\documentclass[11pt,a4paper]{amsart}
\usepackage[foot]{amsaddr}
\usepackage{ifxetex}
\ifxetex
  \usepackage[no-math]{fontspec}
\else
\fi
\usepackage{amsmath}
\usepackage{amsfonts}
\usepackage{amssymb}
\usepackage{amsthm}
\usepackage{fullpage}
\usepackage{microtype}
\ifxetex % fix font incompatibility
  \usepackage[libertine]{newtxmath}
\else
  \usepackage{newtxmath}
\fi
\usepackage[tt=false]{libertine} %% it is ugly 
\usepackage{caption}
\usepackage{tcolorbox}
\usepackage{bbm}
\usepackage{hyperref, color}
\hypersetup{colorlinks=true,citecolor=blue, linkcolor=blue, urlcolor=blue}
\usepackage[linesnumbered,boxed,ruled,vlined]{algorithm2e}
\usepackage{bm}
\usepackage{bbm}
\usepackage[numbers]{natbib}
\usepackage{xcolor}
\usepackage{enumerate} 
\usepackage{enumitem}
\usepackage{ragged2e}
\usepackage{tabularx}
\usepackage{array}
\usepackage{longtable}
\usepackage{hhline}
\usepackage{multirow}
\usepackage{cleveref}
\newcolumntype{L}[1]{>{\raggedright\arraybackslash}p{#1}}
\newcolumntype{C}[1]{>{\centering\arraybackslash}m{#1}}
\newcolumntype{R}[1]{>{\raggedleft\arraybackslash}p{#1}}

\usepackage{makecell}
\usepackage{footnote}
\usepackage{float}
\makesavenoteenv{tabular}

\renewcommand{\epsilon}{\varepsilon}

\newtheorem{theorem}{Theorem}[section]

\newtheorem{claim}[theorem]{Claim}
\newtheorem*{claim*}{Claim}
\newtheorem{condition}[theorem]{Condition}
\newtheorem{property}[theorem]{Property}

\newtheorem{fact}[theorem]{Fact}
\newtheorem{lemma}[theorem]{Lemma}
\newtheorem{proposition}[theorem]{Proposition}

\newtheorem{corollary}[theorem]{Corollary}
\theoremstyle{definition}

\newtheorem{definition}[theorem]{Definition}
\newtheorem{remark}[theorem]{Remark}
\newtheorem*{remark*}{Remark}

\newcommand{\unpin}[1]{{#1}^{\+O}}

\def\\+Ex{\mathop{\mathbf{\+E}}\nolimits}

\renewcommand{\Pr}[2][]{ \ifthenelse{\isempty{#1}}
  {\mathop{\mathbf{Pr}}\left[#2\right]} {\mathop{\mathbf{Pr}}_{#1}\left[#2\right]} }

  \newcommand{\dtv}{d_{\rm TV}}
\newcommand{\cp}{\textnormal{cp}}

\newcommand{\coup}{\textnormal{coup}}
\newcommand{\trun}{\textnormal{trun}}
\newcommand{\lp}{\textnormal{lp}}
\newcommand{\good}{\textnormal{good}}

\newcommand{\poly}{\textnormal{poly}}
\newcommand{\Couple}{\mathsf{Couple}}

\newcommand{\abs}[1]{\left\vert#1\right\vert}

\newcommand{\defeq}{\triangleq}
\newcommand{\supp}{\mathrm{supp}}

\newcommand{\vbl}{\textnormal{vbl}}

\newcommand{\vio}[1]{{\textnormal{\texttt{False}}(#1)}}

\def\*#1{\mathfrak{#1}} % Use \*A for \mathfrak{A}
\def\+#1{\mathcal{#1}} % Use \+A for \mathcal{A}
\def\-#1{\mathrm{#1}} % Use \-A for \mathrm{A}
\def\=#1{\mathbb{#1}} % Use \=A for \mathbb{A}

\usepackage[textsize=tiny]{todonotes}

\usepackage{xifthen}

\newcommand{\E}{\mathcal{E}}

\newcommand{\eps}{\varepsilon}

\makeatletter
\def\prob#1#2#3{\goodbreak\begin{list}{}{\labelwidth\z@ \itemindent-\leftmargin
                        \itemsep\z@  \topsep6\p@\@plus6\p@
                        \let\makelabel\descriptionlabel}
                \item[\it Name]#1
               \item[\it Instance]                #2
                \item[\it Output]#3
                \end{list}}
\makeatother

\newcommand{\vgood}{\ensuremath{V_\textnormal{good}}}
\newcommand{\vbad}{\ensuremath{V_\textnormal{bad}}}
\newcommand{\egood}{\ensuremath{\+E_\textnormal{good}}}
\newcommand{\ebad}{\ensuremath{\+E_\textnormal{bad}}}

\newbool{doubleblind}
\setbool{doubleblind}{true}

\newcommand{\zc}[1]{\textcolor{red}{ZC: #1}}
\newcommand{\al}[1]{\textcolor{red}{AL: #1}}

\title{Counting random $k$-SAT near the satisfiability threshold}

\author{Zongchen Chen, Aditya Lonkar, Chunyang Wang, Kuan Yang, Yitong Yin}
\address[Zongchen Chen, Aditya Lonkar]{School of Computer Science, Georgia Institute of Technology, Atlanta, GA, USA. \textnormal{Email: \texttt{chenzongchen@gatech.edu}, \texttt{alonkar9@gatech.edu}.}}
\address[Chunyang Wang, Yitong Yin]{State Key Laboratory for Novel Software Technology, New Cornerstone Science Laboratory, Nanjing University, 163 Xianlin Avenue, Nanjing, Jiangsu Province, 210023, China. \textnormal{Email: \texttt{wcysai@smail.nju.edu.cn}, \texttt{yinyt@nju.edu.cn}.}}
\address[Kuan Yang]{John Hopcroft Center for Computer Science, Shanghai Jiao Tong University, Shanghai, China. \textnormal{Email: \texttt{kuan.yang@sjtu.edu.cn}.} 
The research of K.\ Yang is supported by the NSFC grant No.\ 62102253.}

\begin{document}

\begin{abstract}
We present efficient counting and sampling algorithms for random $k$-SAT when the clause density satisfies
$\alpha \le \frac{2^k}{\poly(k)}$.
In particular, the exponential term $2^k$ matches the satisfiability threshold $\Theta(2^k)$ for the existence of a solution and the (conjectured) algorithmic threshold $2^k (\ln k) / k$ for efficiently finding a solution. 
%This shows that for random $k$-SAT, the sampling and counting problems are not significantly computationally harder than their satisfiability and algorithmic counterparts. 
Previously, the best-known counting and sampling algorithms required far more restricted densities $\alpha\lesssim 2^{k/3}$~\cite{HWY23}. 
Notably, our result goes beyond the lower bound $d\gtrsim 2^{k/2}$ for worst-case $k$-SAT with bounded-degree $d$~\cite{BGGGS19}, 
showing that for counting and sampling, the average-case random $k$-SAT model is computationally much easier than the worst-case model.

At the heart of our approach is a new refined analysis of the recent novel coupling procedure by ~\cite{WY24}, utilizing the  structural properties of random constraint satisfaction problems (CSPs). %\al{Shall we rather use constraint satisfaction problems (CSPs) the first time?} 
Crucially, our analysis avoids reliance on the $2$-tree structure used in prior works, which cannot extend beyond the worst-case threshold $2^{k/2}$.
%Crucially, our analysis does not require the $2$-tree structure in~\cite{WY24} which causes looseness in previous analysis. 
Instead, we employ a witness tree similar to that used in the analysis of the Moser-Tardos algorithm~\cite{moser2010constructive} for the Lov\'{a}sz Local lemma, which may be of independent interest.
Our new analysis provides a universal framework for efficient counting and sampling for random atomic CSPs, including, for example, random hypergraph colorings. At the same time, it immediately implies as corollaries several structural and probabilistic properties of random CSPs that have been widely studied but rarely justified, including replica symmetry and non-reconstruction.

\end{abstract}	

\maketitle

\setcounter{tocdepth}{1}
\tableofcontents

%\wcy{The title is a placeholder, ready to be changed at any time, so are the titles of the sections}

\section{Introduction}

Random constraint satisfaction problems (CSPs) have attracted lots of attention in both computer science and statistical physics in recent years. Typically, a random CSP consists of a set of $m = \alpha n$ constraints imposed on $n$ variables with finite domains, where the constraints are randomly generated according to a specific rule and $\alpha > 0$ is a constant representing the density of the instance.
The primary goal is to find a feasible solution satisfying all constraints or more generally, optimize a random objective function specified by the constraints. 
Usually, with sufficiently small $\alpha$, the problem is easy and can be solved in polynomial time, while as $\alpha$ grows, the problem becomes computationally hard at a certain point. 
It is, thus, important to understand the computational complexity of random CSPs with respect to the density $\alpha$. 
%particularly as $n$ and $m$ approach infinity.
%We focus on the case of $k$-uniform random CSPs, which are random CSPs where each constraint contains exactly $k$ variables, including many important subcases such as random $k$-SAT, random Not-All-Equal $k$-SAT, and random $k$-uniform hypergraph $q$-colorings.

Perhaps the most notable and important example of random CSPs in computer science is the random $k$-SAT problem. 
Let $v_1,\dots,v_n$ be $n$ Boolean variables taking values in $\{\mathsf{True},\mathsf{False}\}$.
We construct a $k$-SAT formula $\Phi = \Phi(k, n, m = \lfloor\alpha n \rfloor)$ by selecting $m = \lfloor \alpha n \rfloor$ clauses independently where each clause has size $k$ and is obtained by selecting each literal independently and uniformly at random from $\{v_1,\dots, v_n, \lnot v_1, \dots, \lnot v_n\}$.
The fundamental problem for random $k$-SAT is to understand for what value of $\alpha$, a solution (i.e., satisfying assignment of variables) exists and can be found efficiently by an algorithm.
In recent years, significant progress has been made toward these problems.
For random $k$-SAT, numerical experiments and heuristic arguments~\cite{mezard2002analytic, mezard2005clustering} support the \emph{Satisfiability Conjecture}, which posits the existence of a threshold $\alpha_{\mathrm{sat}} = \alpha_{\mathrm{sat}}(k) = 2^k \ln 2 - (1 + \ln 2)/2 + o_k(1)$ that can be described explicitly, such that for a random $k$-SAT instance $\Phi(k,n,m=\lfloor\alpha n \rfloor)$,
\begin{align}\label{eq:k-SAT-sat}
\lim\limits_{n\to\infty} \Pr{\Phi(k,n,m)\text{ is satisfiable}}=
\begin{cases}
1, & \alpha > \alpha_{\mathrm{sat}};\\
0, & \alpha < \alpha_{\mathrm{sat}}.
\end{cases}
\end{align}
Building on a long line of works~\cite{kirousis1998approximate, friedgut1999sharp, achlioptas2002asymptotic, achlioptas2003threshold, coja2014asymptotic, ding2022satisfiability}, Ding, Sly, and Sun~\cite{ding2022satisfiability} finally established \eqref{eq:k-SAT-sat} for sufficiently large $k$, confirming the prediction arising from statistical physics. For other canonical random CSPs, such as random $k$-uniform hypergraph $q$-colorings, the satisfiability threshold for the existence of a coloring remains unclear, with the current best upper and lower bounds differing by an additive $\ln 2+o_q(1)$ factor~\cite{hatami2008sharp, ayre2019hypergraph}.

The next question is to design an efficient algorithm to find a solution at densities below the satisfiability threshold where a solution exists with high probability. 
It is tempting to hope that whenever a solution exists (with high probability) one can also find it in polynomial time.
However, the best-known algorithm for random $k$-SAT, developed by Coja-Oghlan~\cite{coja2010better}, works for densities up to $(1 - o_k(1))2^k (\ln k) / k$, leaving a gap of $\Theta(\log k/ k)$ to the satisfiability threshold. 
Meanwhile, with the hope of approaching the satisfiability threshold $\alpha_{\mathrm{sat}}$, statistical physicists have developed sophisticated and novel message-passing algorithms (the \emph{cavity method}) such as belief propagation, survey propagation, and their advanced versions. 
While numerical experiments suggest that these algorithms perform pretty well for small values of $k$, it has been rigorously proved that in general they fail beyond the density $2^k (\ln k) / k$ \cite{Hetterich16,CO17}.

It has been noted that the density $2^k (\ln k) / k$ marks a phase transition in the geometry of the solution space; see \Cref{figure:transition} for an illustration.
When the density $\alpha$ is below $2^k (\ln k) / k$, the solution space consists of a single giant component that contains almost all solutions and any two solutions from the component can be connected by a path of solutions such that adjacent pairs have Hamming distance $o(n)$. 
Meanwhile, once the density $\alpha$ goes slightly beyond $2^k (\ln k) / k$, the set of solutions is divided into exponentially many small clusters such that each cluster is well-connected, but any two distinct clusters are $\Omega(n)$ away in Hamming distance.
In this regime, a random solution contains many frozen variables with high probability and there exist long-range correlations between variables~\cite{achlioptas2008algorithmic}.
For this reason, the threshold $\alpha_{\mathrm{clust}} \approx 2^k (\ln k) / k$ is called the \emph{clustering threshold}.

Such long-range correlations and the emergence of frozen variables in the clustering phase can be characterized by the \emph{overlap gap property}, which is the barrier for a large family of popular algorithms, including local search algorithms and message-passing algorithms like belief and survey propagation.
%More generally, it is described as the \emph{overlap gap property}.
Thus, the algorithmic threshold for the searching problem is conjectured to coincide with the clustering threshold.
This has been partially established by the recent work \cite{BH22} of Bresler and Huang who proved that the class of \emph{low degree polynomial algorithms} fail at density $4.91 (2^k (\ln k) / k)$, by establishing a stronger version of the overlap gap property. 
We refer to the survey of Gamarnik \cite{Gamarnik21} for more discussions on the overlap gap property and its implication on the intractability of random CSPs.

In this paper, we go beyond searching a solution for random $k$-SAT and consider the even harder problems of \emph{sampling} a uniformly random solution and \emph{counting} the number of solutions. 
Sampling and counting are crucial subroutines for many other statistical and computational tasks such as inference, testing, or prediction.
%Specifically, we study the computational complexity of sampling an almost uniform solution of random CSP formulas as the density of the formula changes. This problem is closely related to estimating the number of solutions, $Z(\Phi)$, to a random CSP formula $\Phi$. 
There have been several recent works on this topic~\cite{montanari2007counting,GGGY21,HWY23,chen2024fast}. For random $k$-SAT, the current best algorithm is given by He, Wu, and Yang~\cite{HWY23}, who presented an efficient algorithm for almost uniform sampling random $k$-SAT solutions up to densities $\alpha\lesssim 2^{k/3}$. Yet, there remains a huge exponential gap between this bound and the clustering threshold $\alpha_{\mathrm{clust}} \approx 2^k (\ln k) / k$ for searching algorithms. 

Meanwhile, message-passing algorithms such as belief and survey propagation from the cavity method are naturally inference algorithms for estimating marginal probabilities and are believed to work all the way up to the clustering threshold.
Given the close relationship between inference and counting/sampling, this seems to suggest that counting and sampling may also be tractable up to the clustering threshold, which is the conjectured threshold for searching algorithms. 
Therefore, it is natural to ask the following question:

\begin{center}
    \emph{For random $k$-SAT, are counting and sampling tractable up to the algorithmic threshold for searching?}
    %(which is likely the clustering threshold)?
\end{center}

\subsection{Counting and sampling for random $k$-SAT}

%We consider the instance of random $k$-SAT and random $k$-uniform hypergraph $q$-colorings unified under the model of random $k$-uniform atomic CSP formulas.
We establish the tractability of approximate counting and almost uniform sampling for random $k$-SAT up to densities $\alpha \le 2^k / \poly(k)$. %and random $k$-uniform hypergraph $q$-colorings $\alpha\lesssim q^{(1-o(1))k}$, respectively. 
In particular, the exponential term $2^k$ matches the satisfiability and algorithmic threshold for random $k$-SAT and extends beyond the lower bound for worst-case $k$-SAT instances (see \Cref{rmk:bounded-degree-CSP} for more discussion).

%We first present our result for random $k$-SAT. 
The random $k$-SAT model is formally defined as follows.

 \begin{definition}[random $k$-SAT formulas]\label{definition:k-SAT-model}
    For $k\geq 3$, define $\Phi(k,n,m)$ as the law of the $k$-SAT formula chosen uniformly at random from all $k$-SAT formulas with $n$ variables and $m$ constraints. 
    
    Specifically, the random $k$-SAT formula $\Phi=(V,C)\sim\Phi(k,n,m)$ is generated as follows:
    \begin{itemize}
        \item The variable set is defined as $V=\{v_1,v_2,\dots,v_n\}$.
        \item The constraint set is defined as $C=\{c_1,c_2,\dots,c_m\}$, where each constraint $c_i$ consists of exactly $k$ literals $\ell_{i,1},\ell_{i,2},\dots,\ell_{i,k}$, with each literal $\ell_{i,j}$ chosen uniformly at random from all $2n$ literals $\{v_1,v_2,\dots,v_n,\neg v_1,\neg v_2,\dots,\neg v_n\}$.
\end{itemize}
\end{definition}

We present  results for efficient (approximate) counting and (almost uniform) sampling of random $k$-SAT up to densities $\alpha \le {2^{k}}/{\poly(k)}$, improving upon the previous best regime $\alpha\lesssim 2^{k/3}$~\cite{HWY23}. 

\begin{theorem}[counting and sampling random $k$-SAT solutions]\label{theorem:k-SAT-main}
    %Let $\zeta>0$ be any constant. Let $k\geq k_0(\zeta)$ and $\alpha>\frac{2^{(1-\zeta)k}}{\poly(k)}$. 
    %Let $\alpha \le \frac{2^k}{c_1 k^{c_2}}$
    There exists a universal constant $c \geq 1$ such that the following holds 
    with high probability over the choice of the random $k$-SAT formula $\Phi=(V,C)\sim\Phi(k,n,\lfloor\alpha n \rfloor)$ where $0 < \alpha \le 2^k / k^c$. 
    
    For any $\eps>0$, there exist the following algorithms, both with running time $\left(n/\eps\right)^{\poly(k,\alpha)}$:
\begin{itemize}
    \item (Counting) A deterministic algorithm that outputs $\hat{Z}$ which is an $\eps$-approximation of the number of solutions $Z(\Phi)$ of $\Phi$, i.e., $(1-\eps) Z(\Phi) \le \hat{Z} \le (1+\eps) Z(\Phi)$; 
    \item (Sampling) An algorithm that outputs a random assignment $X\in \{\mathsf{True},\mathsf{False}\}^V$ that is $\eps$-close in total variation distance to $\mu_{\Phi}$, the uniform distribution over all solutions of $\Phi$.
\end{itemize}
%both within running time $\left(\frac{n}{\eps}\right)^{\poly(k,\alpha)}$.
\end{theorem}

%\zc{Remark 1.3 feels too long and distractive to me. I suggest just using the following short paragraph.}

One may hope to estimate the number of solutions $Z(\Phi)$ by simply computing and outputting $\mathbb{E}[Z(\Phi)]$; this however does not provide an effective approximation algorithm. The random $k$-SAT model does not enjoy the \emph{superconcentration property} \cite{bapst2017planting,chatterjee2024number} where $Z(\Phi)$ concentrates around $\mathbb{E}[Z(\Phi)]$ with constant or tiny $\omega(1)$ factors with high probability, and furthermore, the standard trick for boosting $\poly(n)$ approximation ratios to FPTAS \cite{SJ89} does not apply on such random instances.

We remark that while we consider random $k$-SAT in \Cref{theorem:k-SAT-main}, we can easily obtain similar results for random \emph{regular} $k$-SAT of variable-degree $d=k\alpha$ whose analysis would be simpler due to the absence of high-degree vertices.
For the random regular hypergraph independent set problem (equivalently, random regular \emph{monotone} $k$-SAT), it was shown that the Glauber dynamics for sampling is rapidly mixing at density $\alpha = O(2^k/k^2)$ \cite{HSZ19}.

%\subsection{Result for random hypergraph colorings}
Our counting and sampling algorithms apply to general random atomic CSPs.
Here, we present results for random hypergraph colorings. 
A hypergraph $H=(V,\+E)$ is  $k$-uniform if $|e|=k$ for all $e\in \+E$. 
We adopt the following definition for the random generation of $k$-uniform hypergraphs.

\begin{definition}[Erd\H{o}s-R\'{e}nyi hypergraph]\label{definition:random-hypergraph}
   For $k\geq 2$, define $H(k,n,m)$ as the uniform distribution over
   all $k$-uniform hypergraphs with $n$ vertices and $m$ distinct hyperedges. 
\end{definition}

For a hypergraph $H=(V,\+E)$, a proper hypergraph $q$-coloring $X\in [q]^V$ assigns one of the $q$ colors to each $v\in V$, ensuring no hyperedge is monochromatic. 
We present results for efficiently counting and sampling proper $q$-colorings of random $k$-uniform hypergraphs up to densities $\alpha \le q^{k}/\poly(k,q)$.

\begin{theorem}[counting and sampling random $k$-uniform hypergraph $q$-colorings]\label{theorem:q-coloring-main}
    %Let $\zeta>0$ be any constant. Let $k\geq k_0(\zeta)$ and $\alpha>\frac{q^{(1-\zeta)k}}{\poly(k)}$. 
    There exists a universal constant $c \geq 1$ such that the following holds 
    with high probability over the choice of the random hypergraph $H=(V,\+E)\sim H(k,n,\lfloor \alpha n\rfloor)$ where $\alpha \le q^k / (q k)^c$.
    
    For any $\eps>0$, there exist the following algorithms, both with running time $\left(n/\eps\right)^{\poly(\log q, k,\alpha)}$:
\begin{itemize}
    \item (Counting) A deterministic algorithm that outputs $\hat{Z}$ which is an $\eps$-approximation of the number of proper $q$-colorings $Z(H,q)$ of $H$, i.e., $(1-\eps) Z(H,q) \le \hat{Z} \le (1+\eps) Z(H,q)$;
    \item (Sampling) An algorithm that outputs a random assignment $X\in [q]^V$ that is $\eps$-close in total variation distance to $\mu_{H}$, the uniform distribution over all proper $q$-colorings of $H$.
\end{itemize}
%within running time $\left(\frac{n}{\eps}\right)^{\poly(\log q, k,\alpha)}$.
\end{theorem}

An analog of \Cref{theorem:q-coloring-main} holds for random \emph{regular} hypergraph colorings as well.

\begin{remark}[Comparison with worst-case bounded-degree CSPs]
\label{rmk:bounded-degree-CSP}
%Before addressing the question above, 
It is helpful to review the literature on counting and sampling for worst-case bounded-degree CSPs. A random $k$-SAT formula with density $\alpha$ has an average degree of $k\alpha$, making it comparable to a $k$-SAT formula with maximum degree $d=k\alpha$. For bounded-degree CSPs, while a solution exists and can be efficiently found when  $d \lesssim 2^{k}$~\cite{LocalLemma,moser2010constructive} via the celebrated (algorithmic) Lov\'{a}sz local lemma, the problem of sampling and counting solutions becomes intractable when $d \gtrsim 2^{k/2}$~\cite{BGGGS19}. A similar separation occurs for $k$-uniform hypergraph $q$-colorings: while the existence/searching problem can be solved when  $d\lesssim q^{k}$, the sampling/counting problem becomes intractable at $d \gtrsim q^{k/2}$~\cite{galanis2021inapproximability}. This demonstrates that the sampling and counting problems are computationally much harder than searching for worst-case bounded-degree CSPs.
In contrast, our findings demonstrate that for random CSPs, the sampling and counting problems are not significantly computationally harder than their satisfiability and algorithmic counterparts.
%which stands in stark contrast to the separation observed in the worst-case bounded-degree CSP model.
\end{remark}

\begin{remark}
    As mentioned, our main results \Cref{theorem:k-SAT-main,theorem:q-coloring-main} in fact apply to a broad family of random CSPs (see \Cref{theorem:CSP-main} for a formal statement).
    One may wonder if the counting/sampling threshold would always be close to or even match the searching threshold for general random CSPs.
    The answer is probably no in general. A very recent paper \cite{AG24} by El Alaoui and Gamarnik shows that for the symmetric binary perceptron model, another important example of random CSPs, sampling a solution is intractable at \emph{any} constant density (for two common classes of algorithms) while there are known algorithms for finding solutions at sufficiently low density.
    
    We note that the random CSPs considered in this paper are \emph{sparse}, in the sense that both the size of constraints and the (average) degree of variables are constant.
    Meanwhile, the symmetric binary perceptron model is \emph{dense} as the constraint size and variable degrees are both linear.
    This leads to a huge difference in the geometry of the solution space: For the symmetric binary perceptron model at any constant density, most solutions are isolated with linear distance from each other.
    Therefore, it is still possible that the counting/sampling threshold is equivalent to the searching threshold, or at least somewhere near it, for general sparse random CSPs as considered in this paper. 
\end{remark}

%$2$-random hypergraph coloring is NAE-SAT (\cite{ding2014satisfiability,sly2016number,nam2022replica,achlioptas2002colorability}) 

%$q$-random hypergraph coloring (\cite{dyer2015chromatic,ayre2019hypergraph,Gabri2017PhaseTI})

\subsection{Constraint-wise coupling and replica symmetry} %

Given the extensive focus in the literature, we primarily present our results for random $k$-SAT in this subsection; however, the results and proofs are equally applicable to random hypergraph colorings or other random atomic CSPs.

The central tool for our algorithmic results is a constraint-wise coupling developed in \cite{WY24}.
Let $\Phi=(V,C)\sim\Phi(k,n,\lfloor\alpha n \rfloor)$ be a $k$-SAT instance at density $\alpha > 0$, and let $c_0 \in C$ be an arbitrary clause. We define $\mu_\Phi$ as the uniform distribution over all solutions to $\Phi$, and $\mu_{\Phi \setminus c_0}$ as the uniform distribution over solutions to $\Phi \setminus c_0 := (V,C\setminus \{c_0\})$ with the clause $c_0$ removed from $\Phi$.
We establish the following result concerning the $1$-Wasserstein distance between the two distributions $\mu_\Phi$ and $\mu_{\Phi \setminus c_0}$.

\begin{theorem}\label{thm:W1}
    There exists a universal constant $c \geq 1$ such that the following holds 
    with high probability over the choice of the random $k$-SAT formula $\Phi=(V,C)\sim\Phi(k,n,\lfloor\alpha n \rfloor)$ where $0 < \alpha \le 2^k / k^c$. 
    
    For any $c_0 \in C$, it holds that
    \begin{align*}
        W_1\left( \mu_\Phi, \mu_{\Phi \setminus c_0} \right) = O(\log n),
    \end{align*}
    where $W_1(\cdot,\cdot)$ denotes the $1$-Wasserstein distance with respect to the Hamming metric.
    That is, there exists a coupling $(X,Y)$ of $\mu_\Phi$ and $\mu_{\Phi \setminus c_0}$ such that $\mathbb{E}[d_{\-{Ham}}(X,Y)] = O(\log n)$.
\end{theorem}

\Cref{thm:W1} is at the heart of our counting and sampling algorithms from \Cref{theorem:k-SAT-main}.
Intuitively, it states that given a solution $Y$ sampled from $\mu_{\Phi \setminus c_0}$ (which may violate $c_0$), we can flip the value of at most $O(\log n)$ variables with high probability to obtain a truth assignment $X$ in such a way that $X$ is a solution to $\Phi$ and is distributed uniformly at random as $\mu_\Phi$. 
%\al{w.h.p.? Also, would  it be better to shift the previous line to appear after Thm 1.6?}
We establish \Cref{thm:W1} by constructing a recursive coupling procedure and showing it terminates with high probability within $O(\log n)$ iterations. 
We then apply an LP-based algorithm introduced by Moitra \cite{Moi19} which provides fast counting and sampling algorithms; this is also the strategy in \cite{WY24}. %\al{Is it a good idea to mention that the LP technique was introduced by Moitra but our LP is sufficiently different?}

The coupling result in \Cref{thm:W1} can be understood as a way to describe the \emph{decay of correlation} phenomenon. 
In particular, it is similar to various other notions of correlation decay such as disagreement percolation for showing Gibbs uniqueness \cite{BS94}, recursive coupling for proving strong spatial mixing \cite{GMP05}, and especially coupling independence for showing spectral independence \cite{CZ23,chen2024fast}.
For random CSPs like $k$-SAT, the Wasserstein distance 
between the original formula and the formula obtained after the removal of a single clause is reminiscent of the cavity method.
%when a single clause is removed also reminds one of the cavity method.
However, establishing \Cref{thm:W1} for random $k$-SAT is quite different from those previous works on distinct models. Firstly, the $k$-SAT instance is based on hypergraphs while previous rigorous approaches for correlation decay mostly consider models defined on graphs.
Secondly, the hypergraph associated with a random $k$-SAT formula contains a large maximum degree, making the analysis significantly more challenging.
Finally, \Cref{thm:W1} considers the removal of a constraint, while previous approaches often consider flipping the value of a variable.

Before discussing our proof approach for \Cref{thm:W1} in \Cref{subsec:overview}, we first mention a few direct yet important implications of \Cref{thm:W1} for random $k$-SAT at the considered densities.
%We then demonstrate that \Cref{thm:W1} can be applied to immediately give several important structure properties of random CSPs at the considered densities. 

Many empirical studies in statistical physics use heuristics to predict the solution space geometry of random $k$-SAT, with the (predicted) phase transitions illustrated in \Cref{figure:transition}. Extensive research from both computer science and statistical physics has been dedicated to understanding these phases and their corresponding thresholds~\cite{mezard2005clustering,daude2005pairs,achilioptas2006solution,coja2012decimation,lenka2016statistical,budzynski2020asymptotics}. Notably, the satisfiability threshold $\alpha_{\mathrm{sat}}$ has been precisely determined by~\cite{ding2022satisfiability} for sufficiently large $k$. Other important thresholds, such as the clustering threshold $\alpha_{\mathrm{clust}}$ (also referred to as the dynamic phase transition threshold in ~\cite{florent2007gibbs}), where the solution space of a $k$-SAT fragments into an exponential number of clusters, and the condensation threshold $\alpha_{\mathrm{cond}}$, beyond which the solution space is dominated by a few large clusters, have been extensively studied but are not yet fully characterized.

\begin{figure}[t]
    \centering
    \includegraphics[scale=0.15]{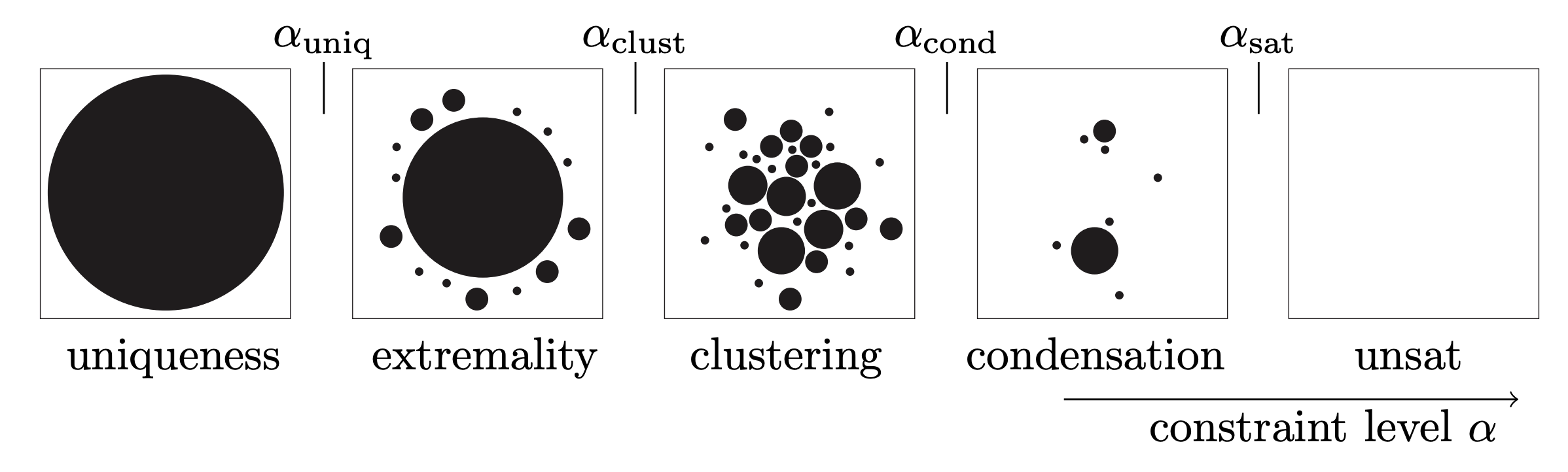}
    \caption{\small The heuristic diagram in ~\cite{ding2022satisfiability} depicts phase transitions in the geometry of the solution space of a random $k$-SAT instance as the density $\alpha$ increases from left to right. 
    %In this diagram, each circle represents a cluster of solutions, where solutions within a cluster are connected by flipping a single bit.   
    }
    \label{figure:transition}
    \end{figure}

%\subsubsection{Absence of long-range correlations}

One of the key tools for statistical physicists in understanding and describing the solution space geometry of random $k$-SAT are the notions of \emph{replica symmetry} and \emph{replica symmetry breaking}~\cite{montanari2008cluster}, first introduced by Parisi~\cite{parisi1979infinite}. Informally, ``replica symmetry'' refers to the idea that two uniformly chosen variables are nearly independent, while ``replica symmetry breaking'' corresponds to the existence of extensive long-range correlations. Following~\cite{florent2007gibbs}, we present the formal definition below.

\begin{definition}[replica symmetry]\label{definition:replica-symmetry}
    Given $k,\alpha$, we say that a random $k$-SAT model with density $\alpha$ is \emph{replica symmetric} if, 
    for $\Phi=(V,C)\sim\Phi(k,n,\lfloor\alpha n \rfloor)$, a uniform satisfying assignment $\sigma\sim \mu_{\Phi}$, and two variables $v_1,v_2\in V$ chosen uniformly at random, the following holds:
    \[
    \lim\limits_{n\to \infty} \abs{\Pr{\sigma(v_1)=\sigma(v_2)=\mathsf{True}}-\Pr{\sigma(v_1)=\mathsf{True}}\Pr{\sigma(v_2)=\mathsf{True}}}=0.
    \]
\end{definition}

\Cref{definition:replica-symmetry} essentially states that the events $\sigma(v_1) = \mathsf{True}$ and $\sigma(v_2) = \mathsf{True}$ are asymptotically independent for large $n$, therefore indicating the absence of long-range correlations in a relatively weak sense, as the typical distance between $v_1$ and $v_2$ is $\Omega(\log n)$. 
Replica symmetry is conjectured to hold up to the condensation threshold $\alpha_{\mathrm{cond}}$, which has been verified on a few other models; see e.g. \cite{CKPZ17,CEJKK18} and the references therein.
For random \emph{soft-constraint} $k$-SAT, it has been shown that replica symmetry is sufficient for the success of the Belief Propagation algorithm~\cite{coja2022belief}.

While the replica symmetry condition has been extensively discussed in the literature, this property is rarely proved rigorously.
In this work, we show replica symmetry holds under the considered densities, which follows immediately from \Cref{thm:W1}. 
%Note that according to ~\cite[Theorem 1.1]{coja2022belief}, \Cref{theorem:replica-symmetry} implies that the Belief Propagation message-passing algorithm works under our density assumptions. 

\begin{theorem}[replica symmetry of random $k$-SAT]\label{theorem:replica-symmetry}
    Under the condition of \Cref{theorem:k-SAT-main}, the random $k$-SAT model with density $\alpha$ is replica symmetric.
\end{theorem}

Another important property that arises from the study of the random $k$-SAT by statistical physicists is the \emph{(non-)reconstruction} property~\cite{mezard2006reconstruction}, which informally requires one being able to estimate the value of one variable given ``far away'' observations~\cite{gerschenfeld2007reconstruction,montanari2011reconstruction}.

For an SAT instance $\Phi=(V,C)$, we use $\vbl(c)$ to denote the set of variables involved in $c$ for each clause $c\in C$. We let $H_{\Phi}=(V,\+E)$ denote the hypergraph where $V$ is the set of variables in $\Phi$ and $\+E=\{\vbl(c)\mid c\in C\}$. For a hypergraph $H$, we use $d_H(\cdot,\cdot)$ to denote the graph-theoretic distance in $H$. Finally, for a hypergraph $H=(V,\+E)$, a vertex $v\in V$, and any $r\geq 1$, we let $\bar{B}_H(v,r)\defeq \{u\in V\mid d_H(u,v)\geq r\}$ to denote the set of vertices in $V$ with distance to $v$ at least $r$. We then follow ~\cite{montanari2011reconstruction} to give the following formal definition of the non-reconstruction property.

\begin{definition}[non-reconstruction]
    Given $k,\alpha$, we say that a random $k$-SAT model with density $\alpha$ is \emph{non-reconstructible} if, 
    for $\Phi=(V,C)\sim\Phi(k,n,\lfloor\alpha n \rfloor)$ with the uniform distribution over satisfying assignments $\mu=\mu_{\Phi}$ and the induced hypergraph $H=H_{\Phi}$,  the following holds for any $v\in V$:
    \[
    \lim\limits_{r\to \infty}\limsup\limits_{n\to \infty}\mathbb{E}\left[\dtv\left(\mu_{\{v\} \cup \bar{B}_H(v,r)},\mu_v \otimes \mu_{\bar{B}_H(v,r)}\right)\right]=0.
    \]
    We say that the model is \emph{reconstructible} otherwise. Here, $\mu_{v}$, $\mu_{\bar{B}_H(v,r)}$, $\mu_{\{v\} \cup \bar{B}_H(v,r)}$ denote the marginal distributions induced from $\mu$ on the subset of variables $\{v\},\bar{B}_H(v,r)$, $\{v\} \cup \bar{B}_H(v,r)$, respectively.
\end{definition}

The non-reconstruction property, which also reflects the absence of long-range correlations, is a stronger condition than replica symmetry~\cite{coja2022belief}. In fact, non-reconstructibility is a necessary condition for the rapid mixing of Glauber dynamics or other local Markov chains~\cite{berger2005glauber}. Furthermore, non-rigorous statistical mechanics calculations~\cite{mezard2002analytic} and approximations up to the second order~\cite{montanari2011reconstruction} strongly suggest that the threshold for non-reconstruction aligns with the clustering threshold $\alpha_{\mathrm{clust}}$. 
%\al{Why is ``non-'' in brackets?}. 
However, there is no rigorous proof of this as far as we know.
In this work, we demonstrate that non-reconstruction holds under the considered densities, which is again a direct corollary of \Cref{thm:W1}.

\begin{theorem}[non-reconstruction of random $k$-SAT]\label{theorem:non-reconstruction}
    Under the condition of \Cref{theorem:k-SAT-main}, the random $k$-SAT model with density $\alpha$ is non-reconstructible.
\end{theorem}

We further establish \emph{looseness} of variables under the considered densities, an intuitive way of characterizing the connectivity of the solution space \cite{achlioptas2008algorithmic,chen2024fast}; see \Cref{definition:looseness} and \Cref{theorem:looseness} for details.

\subsection{Technique Overview}
\label{subsec:overview}

We establish all of the results mentioned above in the context of atomic CSPs. A CSP is considered atomic if each constraint is violated by exactly one assignment in its domain. Any constraint with a constant number $N$ of violating configurations can be decomposed into $N$ atomic constraints. Both SAT and hypergraph coloring instances fall under this category of atomic CSPs.

Our algorithm for counting and sampling atomic CSP solutions is inspired by the recently developed recursive coupling procedure for bounded-degree atomic CSPs in ~\cite{WY24}. Through the novel recursive coupling procedure, ~\cite{WY24} successfully dispenses with the freezing paradigm equipped by the previous approaches~\cite{beck1991algorithmic,Moi19,feng2021sampling,vishesh21towards}, therefore bypassing the technical barrier and achieving a $ q^{k}\gtrsim d^{2+o_q(1)}$ bound for counting/sampling bounded-degree CSPs, where $d$ is the maximum degree of the instance. The freezing paradigm and its variant are also applied by all recent counting/sampling algorithms for random CSPs~\cite{GGGY21,HWY23,chen2024fast}. Hence, we naturally circumvent this barrier for random CSPs.

However, the bound in~\cite{WY24} has an exponent of $2 + o_q(1)$, which only approaches $2$ for large domain sizes and rises to $\approx 4.82$ for the worst-case of atomic CSPs with Boolean domains, such as for $k$-SAT. The way we improve this exponent to $1 + o(1)$ is by leveraging \emph{a new analysis of the coupling procedure based on the structural properties of random CSPs}. Random CSPs enjoy good structural properties such as large constraint expansion, as already observed and utilized in \cite{HWY23}. Our main novelty lies in designing a new analysis for the coupling in ~\cite{WY24} that takes advantage of this structural property. Technically, we replace the $2$-tree witness employed in the analysis of the coupling in ~\cite{WY24} by a denser witness structure, constructed similarly as the witness tree structure for analyzing the Moser-Tardos algorithm~\cite{moser2010constructive}, thereby achieving the improved bound. 

This refined analysis of the coupling procedure immediately establishes properties such as replica symmetry and non-reconstruction at the same density. 
%The proof of the connectivity of the solution space is achieved by additionally demonstrating a monotone property with respect to the coupling. We believe that both the refined analysis and the monotone property are of independent interest and may have further applications in both bounded-degree and random CSPs.
We believe this new analysis is of independent interest and may have further applications in both bounded-degree and random CSPs.

\subsection{Organization}

This paper is organized as follows.

In \Cref{sec:pre}, we formally define atomic CSPs and introduce relevant preliminaries and notations.

In \Cref{sec:structure}, we define the main structural condition (\Cref{condition:main-condition}), stated generally with respect to atomic CSPs. We then claim that this condition suffices for efficient sampling and counting of solutions (\Cref{theorem:CSP-main}), and prove that both the random $k$-SAT instance in \Cref{theorem:k-SAT-main} and the random hypergraph coloring instance in \Cref{theorem:q-coloring-main} satisfy this condition. 

\Cref{sec:coupling} and \Cref{sec:lp} together proves \Cref{theorem:CSP-main}, which concludes the proof of \Cref{theorem:k-SAT-main} and \Cref{theorem:q-coloring-main}. Specifically, \Cref{sec:coupling} introduces the coupling in ~\cite{WY24} with our new analysis and establishes \Cref{thm:W1}. \Cref{sec:lp} shows how to effectively convert this coupling into an efficient counting and sampling algorithm.

In \Cref{sec:connectivity}, we prove the properties related to the correlation decay phenomenon and the geometry of the solution space, namely Theorems \ref{theorem:replica-symmetry}, \ref{theorem:non-reconstruction} and \ref{theorem:looseness}, using the refined analysis of the coupling developed in \Cref{sec:coupling}.

\section{Preliminaries}\label{sec:pre}

\subsection{Atomic CSPs and related notations}
A constraint satisfaction problem (CSP) is described by a collection of constraints defined on a set of variables. 
Formally, an instance of a constraint satisfaction problem, 
called a \emph{CSP formula}, 
is denoted by $\Phi=(V,[q],\+C)$.
Here, $V$ is a set of $n=|V|$ random variables, where each random variable $v\in V$ has
 a finite domain $[q]$. The collection of local constraints is given by $\+C$, where each constraint
 $c\in \+C$ is a function defined as $c:[q]^{\vbl(c)}\to\{\mathsf{True},\mathsf{False}\}$  over a subset of variables denoted as  $\vbl(c)\subseteq V$. For any subset of of constraints $\+E\subseteq \+C$, denote $\vbl(\+E)\defeq\bigcup_{c\in \+E}\vbl(c)$. 
An assignment $\sigma\in [q]^V$ is called \emph{satisfying} for $\Phi$ if %$\Phi(\={x})=\True$, where
\[
\Phi(\sigma)\defeq\bigwedge\limits_{c\in \+C} c\left(\sigma_{\vbl(c)}\right)=\mathsf{True}.
\]
Furthermore, we say $\Phi$ is \emph{satisfiable} if at least one satisfying assignment to $\Phi$ exists. We use $\Omega_{\Phi}$ to denote the set of satisfying assignments to $\Phi$, and use $Z(\Phi)=|\Omega_{\Phi}|$ to denote the number of satisfying assignments to $\Phi$.

We say a constraint $c\in \+C$ is defined by \emph{atomic bad events}, or simply, \emph{atomic}, if it is violated by exactly one configuration in $[q]^{\vbl(c)}$. For an atomic constraint $c$, we use $\vio{c}$ to denote its only violating configuration in $[q]^{\vbl(c)}$. Moreover, when all constraints in $\+C$ are atomic, we say $\Phi$ is atomic. It is important to note that any constraint with a constant number $N$ of violating configurations can be decomposed into $N$ atomic constraints. As a result, both SAT and hypergraph coloring instances fall within the category of atomic CSPs. 
\subsubsection{Notations for (partial) assignments}

For a partial assignment $\sigma\in [q]^\Lambda$ specified over a subset of variables $\Lambda\subseteq V$, we use $\Lambda(\sigma)=\Lambda$ to denote the set of assigned variables in $\sigma$. For any partial assignment $\sigma$ and any $S\subseteq \Lambda(\sigma)$, we use $\sigma_S$ to denote $\bigotimes_{v\in S}\sigma(v)$. We further write $\sigma_v=\sigma_{\{v\}}$ for $v\in V$.

 For any two assignments $\sigma,\tau$ such that $\Lambda(\sigma)\cap \Lambda(\tau)=\emptyset$, we define $\sigma\land \tau\in [q]^{\Lambda(\sigma)\cup \Lambda(\tau)}$ as the concatenation of $\sigma$ and $\tau$ such that for any $v\in \Lambda(\sigma)\cup \Lambda(\tau)$,
\[
(\sigma\land \tau)(v)=\begin{cases}\sigma(v) & v\in \Lambda(\sigma),\\ \tau(v)& v\in \Lambda(\tau).\end{cases}
\]

We will use $\varnothing$ to specifically denote an empty assignment,  distinguishing from the empty set $\emptyset$.

\subsubsection{Notations for events and probability measures}

We begin by specifying some notations for events and probability measures related to the CSP.

\begin{definition}[simple notations for events]\label{definition:notation-events}
For simplicity of notation, we will use:
\begin{itemize}
    \item a constraint $c\in \+C$ to denote the event that this constraint is satisfied;
    \item a subset of constraints $\+E\subseteq \+C$ to denote the event that all constraints in $\+E$ are satisfied;
    \item a partial assignment $\sigma$ to denote the event that the assignment on $\Lambda(\sigma)$ is precisely $\sigma$. 
\end{itemize}

\end{definition}

Note that under this definition, the notation $\sigma\land \tau$ as a concatenation of assignments defined previously is consistent with its interpretation as an event, where  $\sigma\land \tau$ is considered as the logical ``and'' of the two events $\sigma$ and $\tau$. 

 We use $\+P$ to denote the uniform product distribution over the space $[q]^V$.
We use $\mu=\mu_{\Phi}$ to denote the distribution over all satisfying assignments of $\Phi$ induced  by $\+P$, i.e.
\[
    \mu_{\Phi}\defeq \+P\left(\cdot \mid \+C\right).
\]
$\mu_{\Phi}$ is well-defined only when $\Phi$ is satisfiable.

When the variable set $V$ and the domain $[q]$ are clear, we define the following notations for (conditional) distributions for a given set of constraints $\+E$ defined over $V$, and some assignment $\sigma$ defined over $\Lambda(\sigma)$:
\[
\mu_{\+E}\defeq \+P(\cdot \mid \+E),\quad \mu^{\sigma}_{\+E}\defeq \+P(\cdot \mid \+E\land \sigma).
\]

For a probability distribution $\mu$ and some subset of variables $\Lambda\subseteq V$, we use $\mu_{\Lambda}$ to denote the marginal distribution induced by $\mu$ on $\Lambda$. We use commas to separate multiple subscripts; for example, we use $\mu^{\sigma}_{\+E,\Lambda}$ to denote the marginal distribution induced by $\mu^{\sigma}_{\+E}$ on $\Lambda$.

\subsubsection{Pinned formula and pinned constraints}

For a subset of variables $\Lambda\subseteq V$ and a partial assignment $\sigma\in [q]^{\Lambda}$ specified on $\Lambda$, the pinned formula $\Phi=(V,[q],\+C)$ under $\sigma$, denoted by $\Phi^\sigma=(V^\sigma,[q],\+C^\sigma)$, 
is a new CSP formula such that $V^\sigma=V\setminus \Lambda(\sigma)$, 
and the $\+C^\sigma$ is obtained from $\+C$ by: 
\begin{enumerate}
\item
replacing each original constraint $c\in \+C$ with the corresponding pinned constraint $c^\sigma$, where $\vbl(c^{\sigma})=\vbl(c)\setminus\Lambda(\sigma)$ and $c^{\sigma}(\tau)=c(\tau\land \sigma_{\Lambda(\sigma)\cap\vbl(c)})$ for any $\tau\in[q]^{\vbl(c^{\sigma})}$;
%removing all the constraints that have already been satisfied by $\sigma$;
 %\footnote{Recall that a constraint $c$ is satisfied by a partial assignment $\sigma$ if $c$ is satisfied by all full assignments that extend $\sigma$.} 
\item
removing all the resulting constraints that have already been satisfied.
%for each of the remaining constraints, %not yet satisfied by $\sigma$, 
%replacing the variables $v\in \Lambda(\sigma)$ with their values $\sigma(v)$.
\end{enumerate}

%
%Note that $\Phi^\sigma$ is well-defined when $\sigma$ does not violate any constraint in $\+C$. 
Whenever a pinning $\sigma$ is applied to a CSP formula $\Phi=(V,[q],\+C)$, we always assume that $\sigma$ does not violate any constraint in $\+C$.
Under this assumption,  $\Phi^\sigma$ is always well-defined,
%It is easy to see that when $\Phi^\sigma$ is well-defined, 
and we have $\mu_{\Phi^\sigma}=\mu^{\sigma}_{V\setminus\Lambda(\sigma)}$. 
We use $\+C^*$ to denote the set of all possible constraints obtained from pinning some constraint in $\+C$ with a non-violating $\sigma$, including the unpinned constraints in $\+C$.
Finally, for each (possibly pinned) constraint $c\in \+C^*$, we use $\unpin{c}$ to denote its original unpinned constraint in $\+C$. 

\subsubsection{The incidence hypergraph}

We define the underlying \emph{incidence hypergraph} of CSP formulas as follows.

\begin{definition}[incidence (hyper-)graphs for variables and constraints]\label{definition:incidence-hypergraph}
    Given a CSP formula $\Phi=(V,[q],\+C)$ , we define two incidence (hyper-)graphs for variables and constraints respectively:
    \begin{itemize}
        \item We define $H_\Phi = (V, \+E)$ to be the hypergraph (with multiple edges allowed), where $V$ is the set of variables in $\Phi$, and $\+E = \{\vbl(c) \mid c \in \+C\}$.
        \item Let $G_\Phi = \-{Lin}(H_{\Phi})$ to denote the line graph of $H_\Phi$, namely, the vertices in $G_\Phi$ are clauses in $\Phi$, and two clauses $c_1, c_2$ are adjacent in $G_\Phi$ if and only if $\vbl(c_1) \cap \vbl(c_2) \neq \emptyset$.
    \end{itemize}
    In particular, when the context of $\Phi$ is clear, we say a subset of variables $V'\subseteq V$ is connected if the induced sub-graph $H_\Phi[V']$ is connected and say a subset of constraints $\+C'\subseteq \+C$ is connected if the induced sub-graph $G_\Phi[\+C']$ is connected.
    \end{definition}

    Similar to the density of random formulas (the ratio between the number of clauses and the number of variables), we define the density of hypergraphs as the ratio between the number of hyperedges and the number of vertices.

\subsection{Lov\'{a}sz local lemma}
The Lov\'{a}sz local lemma is a gem in the probabilistic method of combinatorics and has inseparable connections with the solution space of CSPs~\cite{LocalLemma}. By viewing the violation of each constraint as a bad event, the celebrated (variable framework) Lov\'{a}sz local lemma gives a sufficient criterion for a CSP solution to exist:

\begin{theorem}[Erd\H{o}s and Lov\'{a}sz~\cite{LocalLemma}]\label{locallemma}
    Given a CSP formula $\Phi=(V,[q],\+C)$, if the following holds
    %there is a function $x : \+{C} \rightarrow (0, 1) $ such that for any $c \in \+{C}$,\\
    \begin{align}\label{llleq}
    \exists x\in (0, 1)^{\+C}\quad \text{ s.t.}\quad \forall c \in \+C:\quad
        {\+P\left[\neg c\right]\leq x(c)\prod_{\substack{c'\in \+{C}\\ {\vbl}(c)\cap {\vbl}(c')\neq \emptyset}}(1-x(c'))},
    \end{align}
    then  
    $$
        {\+{P}\left[ \bigwedge\limits_{c\in \+C} c\right]\geq \prod\limits_{c\in \+C}(1-x(c))>0}.
    $$
\end{theorem}

When the condition \eqref{llleq} is met, 
the probability of any event in the uniform distribution $\mu$ over all satisfying assignments can be well approximated by the probability of that event in the product distribution.
This was observed in \cite{haeupler2011new}:

\begin{theorem}[Haeupler, Saha, and Srinivasan \cite{haeupler2011new}]\label{HSS}
Given a CSP formula $\Phi=(V,[q],\+C)$, if $\eqref{llleq}$ holds, 
then for any event $A$ that is determined by the assignment on a subset of variables $\vbl(A)\subseteq V$, 
\[
   \+{P}\left[A\mid \bigwedge\limits_{c\in \+C}  c\right]\leq \+P\left(A\right)\prod_{\substack{c\in \+{C}\\ \vbl(c)\cap \vbl(A)\neq \emptyset}}(1-x(c))^{-1}.
\]
\end{theorem}

\section{Structural properties of random hypergraphs}\label{sec:structure}

Our strategy to prove main results in \Cref{theorem:k-SAT-main} and \Cref{theorem:q-coloring-main} is to establish that both instances possess certain \emph{structural properties} that enable efficient counting and sampling algorithms. In this section, we will develop a ``nice property'' (formally defined in \Cref{definition:nice-hypergraph}) for the underlying incidence hypergraph. We claim that once this nice property holds for the hypergraph with specific parameters, we can efficiently count and sample solutions. Finally, we demonstrate that the conditions in \Cref{theorem:k-SAT-main} and \Cref{theorem:q-coloring-main} fulfill the desired property.

Recall the definition of incidence hypergraphs in \Cref{definition:incidence-hypergraph}. We then present our claim as \Cref{theorem:CSP-main}, with the ``nice property'' formally defined later in this section in \Cref{definition:nice-hypergraph}.  We note that \Cref{condition:main-condition} only makes assumptions to the hypergraph $H_{\Phi}$, but not the atomic bad events for each constraint.

\begin{condition}[structural condition for atomic CSPs]\label{condition:main-condition}
 For the atomic CSP instance    $\Phi=(V,[q],\+C)$, its incidence hypergraph $H_{\Phi}$ is $(k, \alpha,\eps_1,\eps_2,\eta,\rho,p_1,p_2)$-nice with parameters satisfying:
     \begin{itemize}
        \item $k \ge 30$, $\alpha \le q^k$, $\eta k \ge 4$, and $\mathrm{e}(\rho k\alpha)^{\eta} = 1$;
        \item $\eps_1 = 2\eta$, $p_1 = 6k^7$, $\eps_2 = \frac{12k^5}{(1-\eta)\eta p_1} = \frac2{k^2 (1-\eta)\eta}\le \frac1k$, and $p_2 = \mathrm{e}k^2$.
     \end{itemize}
\end{condition}

%\textcolor{red}{Kuan: after discussing with Zongchen, we prefer to choosing $\eta = \Theta(1/k), \eps = \Theta(\eta), \rho = 1/(\mathrm{e}^k\alpha), p_1 = \Omega(k^7), \eps_2 = O(k^6/p_1)$.}
\begin{theorem}[counting/sampling for atomic CSPs]\label{theorem:CSP-main}
Assume $\Phi=(V,[q],\+C)$ is a satisfiable atomic CSP instance satisfying \Cref{condition:main-condition}. 
%\zc{Can we say: ``Then $\Phi$ is satisfiable, and for any $\eps>0$, ...'' That is, does \Cref{condition:main-condition} imply $\Phi$ is satisfiable? If not, shall we include in \Cref{condition:main-condition} another condition that $\Phi$ is satisfiable?}
%
Then for any $\eps>0$, there exist the following algorithms, both with running time $\left(n/\eps\right)^{\poly(\log q,k,\alpha)}$:
\begin{itemize}
\item (counting) A deterministic algorithm that outputs $\hat{Z}$ which is an $\eps$-approximation of the number of solutions $Z(\Phi)$ of $\Phi$, i.e., $(1-\eps) Z(\Phi) \le \hat{Z} \le (1+\eps) Z(\Phi)$;
\item (sampling) An algorithm that outputs a random assignment $X\in [q]^V$ that is $\eps$-close in total variation distance to $\mu_{\Phi}$, the uniform distribution over all solutions of $\Phi$.
\end{itemize}
\end{theorem}

Also, we claim that both conditions in \Cref{theorem:k-SAT-main,theorem:q-coloring-main} satisfy the desired structural properties. 
%Note that \Cref{theorem:CSP-main} and \Cref{theorem:reduction} together imply \Cref{theorem:k-SAT-main} and \Cref{theorem:q-coloring-main}. 

\begin{theorem}\label{theorem:reduction}
    Both the random $k$-SAT instance under the condition in \Cref{theorem:k-SAT-main} and the random hypergraph coloring instance under the condition in \Cref{theorem:q-coloring-main} satisfy \Cref{condition:main-condition} with high probability.
%    \zc{Added ``with high probability''. Is this theorem a combination of \Cref{lemma:structural-lemma-kSAT,lemma:structural-lemma-uniform}? Also we need to say the instance is satisfiable/colorable with high probability, either by \Cref{condition:main-condition} or citing references.}
\end{theorem}

Note that the condition in \Cref{theorem:k-SAT-main}
 is below the explicit threshold $\alpha < 1.3836 \cdot 2^k/k$, as established in \cite[Theorem 1.3]{frieze1996analysis}. Similarly, the condition in \Cref{theorem:q-coloring-main} falls below the explicit threshold $\alpha < q^{k-1} \ln k$, as given in \cite[Theorem 1.1]{dyer2015chromatic}. As a result, the instances described in \Cref{theorem:k-SAT-main,theorem:q-coloring-main}
are satisfiable with high probability under respective conditions. Consequently, \Cref{theorem:CSP-main,theorem:reduction} together imply \Cref{theorem:k-SAT-main,theorem:q-coloring-main}.

The plan goes as follows: For the rest of this section, we formally define the structural properties and verify that \Cref{theorem:reduction} holds. \Cref{theorem:CSP-main} will be proved later in \Cref{sec:coupling,sec:lp}.

We now start to introduce the desired structural properties of the CSP formula. Most of the properties and proofs of random CSP formulas were presented in \cite{GGGY21,HWY23} using different parameters. Here, we will adapt the proofs to our parameters and show that all properties hold in the uniform hypergraph model as well. 

Let $\+H_{k}$ be the set of all $k$-uniform hypergraphs, and $\+H_{\le k}$ be the set of hypergraphs where each hyperedge contains at most $k$ vertices. Now we can describe the following properties for hypergraphs.

\begin{property}[bounded maximum degree]\label{property:maximum-degree}
Given a hypergraph $H=(V,\+E)$. The maximum degree $\Delta=\Delta(H)$ is at most $4k\alpha+6\log n$.
\end{property}

\begin{property}[edge expansion]
\label{property:edge-expansion}
    For $\eta, \rho \in (0,1)$, we say a hypergraph $H=(V,\+E) \in \+H_{\le k}$ satisfies $(\eta,\rho)$-edge expansion if 
    for any $\ell \le \rho |\+E|$ and any $\ell$ hyperedges $e_1,\dots, e_\ell \in \+E$, it holds
    \begin{align*}
        \left| \bigcup_{i=1}^\ell e_i \right| \ge (1-\eta) k \ell.
    \end{align*}
\end{property}

\begin{property}[bounded neighbourhood growth]
\label{property:bounded-neighbourhood-growth}
    For any $e \in \+E$ and any $\ell \ge 1$, the number of connected subsets of hyperedges containing $e$ of size $\ell$ is at most $n^3 (p_2 \alpha)^\ell$.
\end{property}

Similar to the approaches in \cite{GGGY21, HWY23}, a crucial component of our algorithm for random instances involves identifying high-degree vertices, as well as those whose marginal distributions may be significantly affected by high-degree ones. We begin with the identifying subroutine.

\begin{definition}
    [high-degree vertices]\label{definition:high-degree}
    Given a hypergraph $H = (V, \+E)$ with average degree $d$, and a subset of vertices $V' \subseteq V$, let $\-{HD}(V')\defeq \{v\in V'\mid \-{deg}(v)\geq p_1\alpha\}$ denote the set of \emph{high-degree} vertices in $V'$. 
\end{definition}

\begin{algorithm}[htbp]
\caption{\textsf{IdentifyBad}$(V_0)$ \cite{HWY23}}\label{alg:identify-bad}
\SetKwInOut{Instance}{Instance}
\SetKwInOut{Input}{Input}
\SetKwInOut{Output}{Output}
\SetKwIF{WP}{ElseIf}{Else}{with probability}{do}{else if}{else}{endif}
%  \SetKwInput{KwPar}{Parameter}
\Instance{a hypergraph $\+H=(V, \+E)$ with average degree $d$;}
\Input{a set of vertices $V_0\subseteq V$;}
\Output{a set of bad vertices $\vbad(V_0)$ starting from $V_0$ and a set of bad hyperedges $\ebad(V_0)$;}
 Initialize $\vbad(V_0) \gets \mathrm{HD}(V_0)$ and $\ebad(V_0) = \emptyset$\;
\While{$\exists e\in\+E \setminus\ebad(V_0)$ such that $\abs{ e \cap\vbad(V_0)} > \eps_1 k$}{
 Update $\vbad(V_0) \gets\vbad(V_0) \cup e$ and $\ebad(V_0)\gets\ebad(V_0)\cup\{e\}$
}
 \Return{$\vbad(V_0)$ and $\ebad(V_0)$}
\end{algorithm}

Throughout, we will use the notation $\vbad = \vbad(V)$ and $\ebad = \ebad(V)$ if the hypergraph $H = (V, \+E)$ is clear from the context. The set of \emph{good vertices} $\vgood\defeq V\setminus \vbad$ and the set of \emph{good constraints} $\egood\defeq \+E\setminus \ebad$ is defined to be the set of remaining variables/hyperedges.

\begin{fact}[bounded-degree for good vertices]\label{fact:good-vertex-degree-bound}
For every good vertex $v \in V_\good$, it holds $\deg(v) \le p_1 \alpha$.
\end{fact}

\begin{fact}[bounded fraction of good vertices]\label{fact:good-vertex-fraction}
 For every good hyperedge $e \in \+E_\good$, it holds $(1-\varepsilon_1) k \le |e \cap V_\good| \le k$.
\end{fact}

The following structural properties are useful in our counting and sampling algorithms.

\begin{property}
    [bounded number of bad vertices]\label{property:bounded-bad-vertices}
    Given a hypergraph $H = (V, \+E)$, for any $V_0 \subseteq V$, $\abs{\vbad(V_0)} \le 4 
    \eps_1^{-1}\abs{\-{HD}(V_0)}$, where $\vbad(V_0)$ is obtained from \Cref{alg:identify-bad}.
\end{property}

\begin{property}[bounded fraction of bad hyperedges]
\label{property:bounded-bad-fraction}
    Given a hypergraph $H = (V, \+E)$, let $\ebad = \ebad(V)$ be the set of bad hyperedges obtained from \Cref{alg:identify-bad}. For any $\ell \ge \log n$ and any connected subset of hyperedges in $\-{Lin}(H)$ of size $\ell$, the number of bad hyperedges among them is at most $\eps_2 \ell$.
\end{property}

\begin{definition}[nice hypergraph]\label{definition:nice-hypergraph}
    We say a hypergraph $H=(V,\+E)$ is $(k, \alpha,\eps_1,\eps_2,\eta,\rho,p_1,p_2)$-\emph{nice} if with the choice of $p_1$ at \Cref{definition:high-degree} and $\eps_1$ at \Cref{alg:identify-bad}, the hypergraph $H$:
    \begin{itemize}
    \item is in $\+H_{\le k}$ and has density $\alpha$ where $\alpha \le q^k$;
    \item satisfies bounded maximum degree defined at \Cref{property:maximum-degree};
    \item satisfies $(\eta,\rho)$-constraint expansion defined at \Cref{property:edge-expansion};
    \item satisfies bounded neighbourhood growth with parameter $p_2$ at \Cref{property:bounded-neighbourhood-growth};
    \item satisfies bounded number of bad vertices with parameter $\eps_1$ at \Cref{property:bounded-bad-vertices};
    \item satisfies bounded fraction of bad hyperedges with parameter $\eps_2$ at \Cref{property:bounded-bad-fraction};
    \item has no connected components of size $\ell \ge \log n$ in the line graph induced by bad hyperedges.
    \end{itemize}
\end{definition}

%\textcolor{red}{Kuan: I don't think ``satisfiable'' is a property for hypergraphs. Probably we need another condition for instances...}
%\zc{Maybe we can just add one sentence in the proofs of Theorems 1.2 and 1.4 that the formula is satisfiable with high probability and cite some references; seems to be the easiest way.}

%\begin{definition}[high-degree variables, good and bad variables/hyperedges]\label{definition:good-bad}
%Given a hypergraph $H=(V,\+E)$, two parameters $D,\eta$, we propose the following definitions.

%Let $\-{HD}(H)\defeq \{v\in V\mid \-{deg}(v)\geq D\}$ denote the set of \emph{high degree} variables. 

%The set of \emph{bad vertices} $\vbad\subseteq V$ and the set of \emph{bad hyperedges} $\ebad\subseteq \+C$ are defined through the following deterministic procedure:
%\begin{enumerate}
%    \item Initially, set $\vbad\gets \-{HD}(H)$ and $\ebad\gets \emptyset$;
%    \item While there exists $c\in \+C\setminus \ebad$ such that $|\vbl(c)\cap \vbad|>\eps_1 k$: update $\vbad\gets \vbad\cup \vbl(c)$ and $\ebad\gets \ebad\cup \{c\}$.
%\end{enumerate}
%The set of \emph{good vertices} $\vgood\defeq V\setminus \vbad$ and the set of \emph{good constraints} $\egood\defeq \+C\setminus \ebad$ is defined to be the set of remaining variables/constraints.
%\end{definition}

The following key lemma shows that if all parameters $(k, \alpha, \eps_1, \eps_2, \eta, \rho, p_1, p_2)$ satisfy \Cref{condition:main-condition}, then the underlying incidence hypergraphs of random $k$-SAT formulas is nice with probability $1 - o(1/n)$.

\begin{lemma}\label{lemma:structural-lemma-kSAT}
    For any fixed parameters $k, \alpha, \eta, \rho, \eps_1, p_1, \eps_2, p_2$ where $\eta k \ge 4$, $\-e(\rho k \alpha)^\eta \le 1$, $\eps_1 = 2\eta$, $6k^5\le p_1 \le \-e^{k-2}\alpha$, $\eps_2 = 12k^5/((1-\eta)\eta p_1)$ and $p_2 = \-ek^2$, with probability $1 - o(1/n)$ over the choice of random $k$-SAT formulas $\Phi = \Phi(k, n, m)$ with density $\alpha$, $H_\Phi$ is $(k, \alpha, \eps_1, \eps_2, \eta, \rho, p_1, p_2)$-nice.
\end{lemma}

The proof of \Cref{lemma:structural-lemma-kSAT} is deferred to \Cref{sec:structural-property-proof}. In fact, the same result also holds for random $k$-uniform hypergraphs. The following lemma states that if the incidence hypergraph of a random $k$-SAT formula satisfies some structural property with high probability, then a random hypergraph also satisfies the same property with high probability.

\begin{lemma}
    \label{lemma:technical-lemma-random-models}
    Suppose $k \ge 3$ and $\alpha$ are constants. Let $\mathbb{P}$ be a property for hypergraphs (i.e., $\+P$ is a subset of hypergraphs). If the probability of $H_\Phi$ belonging to $\mathbb{P}$ is $1-o(1/n)$, over the choice of random $k$-SAT formulas $\Phi \sim \Phi(k, n, \lfloor\alpha n \rfloor)$ with density $\alpha$, then a random $k$-uniform hypergraph $H \sim H(k, n, \lfloor \alpha n\rfloor)$ with density $\alpha$ belongs to $\mathbb{P}$ with probability $1 - o(1/n)$ as well.
\end{lemma}

\begin{proof}
    Let $m = \lfloor\alpha n \rfloor$, and $\+E_{k,n,m}$ be the event that the incidence hypergraphs $H_\Phi$ is a $k$-uniform hypergraph with $n$ vertices and $m$ distinct hyperedges. It is easy to see that
    \[
        \mathbf{Pr}_{\Phi \sim \Phi(k, n, m)}[H_\Phi = H \mid \+E_{k,n,m}] = \mathbf{Pr}_{H(k, n, m)}[H]\,,
    \]
    namely, the distribution of the incidence hypergraph $H_\Phi$ conditional on it being $k$-uniform and having $m$ distinct hyperedges is the uniform distribution $H(k,n,m)$.

    Note that if $n$ is sufficiently large, we obtain that
    \begin{align*}
        \mathbf{Pr}_{\Phi \sim \Phi(k, n, m)}[H_\Phi \text{ is $k$-uniform}] &= \Bigl(\frac{n-1}{n}\cdot \frac{n-2}{n}\cdots \frac{n-k+1}{n}\Bigr)^{m}\\
        &\ge \Bigl(1 - \frac{k-1}{n}\Bigr)^{(k-1)m} \quad \ge \mathrm{e}^{-k^2\alpha}\,,
    \end{align*}
    and
    \begin{align*}
        \mathbf{Pr}_{\Phi \sim \Phi(k, n, m)}[H_\Phi \text{ has two identical hyperedges}] &\le \binom{m}{2} \frac{k!}{n^k} \le \frac{k!\alpha^2}{2n^k}\,,
    \end{align*}
    which further implies that
    \[
        \mathbf{Pr}_{\Phi \sim \Phi(k, n, m)}[\+E_{k,n,m}] \ge \-e^{-k^2\alpha} -  \frac{k!\alpha^2}{2n^k} \ge \frac{1}{2}\-e^{-k^2\alpha}
    \]
    for sufficiently large $n$. Thus, it follows that
    \begin{align*}
        &\phantom{{}= {}} \mathbf{Pr}_{\Phi \sim \Phi(k, n, m)}[H_\Phi \not\in \mathbb{P}]\\
        &\ge \mathbf{Pr}_{\Phi \sim \Phi(k, n, m)}[H_\Phi \not\in \mathbb{P} \mid \+E_{k,n,m}] \cdot \mathbf{Pr}_{\Phi \sim \Phi(k, n, m)}[\+E_{k,n,m}]   \\
        &\ge \frac{1}{2}\mathrm{e}^{-k^2\alpha}\cdot \mathbf{Pr}_{H\sim H(k, n, m)}[H \not\in \=P]\,.
    \end{align*}
    Since $\mathbf{Pr}_{\Phi \sim \Phi(k, n, m)}[H_\Phi \not\in \mathbb{P}] = o(1/n)$, we conclude that $\mathbf{Pr}_{H\sim H(k, n, m)}[H \not\in \=P] = o(1/n)$.
\end{proof}

As a corollary, we have the following lemma immediately.

\begin{lemma}
    \label{lemma:structural-lemma-uniform}
    For any fixed parameters $k, \alpha, \eta, \rho, \eps_1, p_1, \eps_2, p_2$ where $\eta k \ge 4$, $\-e(\rho k \alpha)^\eta \le 1$, $\eps_1 = 2\eta$, $6k^5\le p_1 \le \-e^{k-2}\alpha$, $\eps_2 = 12k^5/((1-\eta)\eta p_1)$ and $p_2 = \-ek^2$, with probability $1 - o(1/n)$ over the choice of all $k$-uniform hypergraphs with $n$ vertices and density $\alpha$, a random $k$-uniform hypergraph $H$ is $(k, \alpha, \eps_1, \eps_2, \eta, \rho, p_1, p_2)$-nice.
\end{lemma}

We conclude this section by noting that \Cref{theorem:reduction} is simply a combination of \Cref{lemma:structural-lemma-kSAT,lemma:structural-lemma-uniform}. Here, we do not need the lower bound on $k$ in the conditions of \Cref{theorem:k-SAT-main,theorem:q-coloring-main} as the case of small $k$ can be handled by taking a sufficiently large $c$.

\section{Recursive coupling of random CSPs}\label{sec:coupling}
In this section, we establish the following theorem regarding the decay of correlations for CSP instances that satisfy the structural properties introduced in the previous section.
We note that \Cref{thm:W1} from introduction is an immediate consequence of \Cref{theorem:reduction,theorem:correlation-decay}.

\begin{theorem}\label{theorem:correlation-decay}
    Let $\Phi=(V,[q],\+C)$ be an satisfiable atomic CSP formula satisfying \Cref{condition:main-condition}.
    Let $c_0\in \+C$ be an arbitrary constraint. 
    There exists a coupling $(X,Y)$ of $\mu_{\+C}$ and $\mu_{\+C\setminus \{c_0\}}$ such that for any integer $\log n \leq M \leq \rho m$, it holds that
    \[
    \Pr{d_{\-{Ham}}(X,Y)\geq k\cdot M}\leq 2^{-M}.
    \]
\end{theorem}

\Cref{theorem:correlation-decay} asserts that for CSP instances satisfying the specified structural properties, 
the discrepancy between uniform satisfying assignments induced by any particular constraint decays exponentially.

%\Cref{theorem:correlation-decay} directly implies the absence of long-range correlations, including replica symmetry in  \Cref{theorem:replica-symmetry} and non-reconstruction in \Cref{theorem:non-reconstruction}. These will  be formally proved in \Cref{sec:connectivity}.

\subsection{The coupling procedure}
The coupling in \Cref{theorem:correlation-decay} is the constraint-wise recursive coupling introduced in~\cite{WY24}. 
This coupling procedure, denoted as $\Couple(\+E,\+F,\sigma,\tau)$, is formally described in \Cref{Alg:couple}.
The procedure takes as inputs:
\begin{itemize}
    \item a pair of collections of pinned atomic constraints $\+E,\+F\subseteq \+C^*$, corresponding to two formulas;
    \item a pair of partial assignments $\sigma,\tau\in\+[q]^\Lambda$, specified on the same subset $\Lambda\subseteq V$ of variables.
\end{itemize}
It is assumed that both pinned formulas $\+E^\sigma$ and $\+F^\tau$ are satisfiable.
The objective of the procedure is to generate a pair of random assignments $(X,Y)\in [q]^V\times [q]^V$, such that marginally $X\sim \mu_{\+E}^{\sigma}$ and $Y\sim\mu_{\+F}^{\tau}$, while minimizing the discrepancy between $X$ and $Y$.
%This coupling is critical for establishing the correlation decay asserted in \Cref{theorem:correlation-decay}.

\begin{algorithm}[t]
\caption{$\Couple(\+E,\+F,\sigma,\tau)$~\cite{WY24}} \label{Alg:couple}
\SetKwInOut{Instance}{Instance}
\SetKwInOut{Input}{Input}
\SetKwInOut{Output}{Output}
\SetKwIF{WP}{ElseIf}{Else}{with probability}{do}{else if}{else}{endif}
%  \SetKwInput{KwPar}{Parameter}
\Instance{an atomic CSP formula $\Phi=(V,[q],\+C)$;}
\Input{two subsets of pinned formulas $\+E,\+F\subseteq \+C^*$, and two partial assignments $\sigma,\tau\in [q]^{\Lambda}$ specified on the same subset $\Lambda\subseteq V$ of variables;} 
\Output{a pair of assignments $(X,Y)\in [q]^V\times [q]^V$;}
\If{$\+E^\sigma=\+F^\tau$\label{Line:couple-return-cond}}{
%for atomic CSP, it holds that $\mu_{\+E}^{\sigma}$ and $\mu_{\+F}^{\tau}$ are identical over unassigned variables in $V\setminus \Lambda$\;
let $(X,Y)$ be drawn according to the coupling of $\mu_{\+E}^{\sigma}$ and $\mu_{\+F}^{\tau}$ that always satisfies $X_{V\setminus \Lambda}=Y_{V\setminus \Lambda}$\;\label{Line:couple-perfect}
    \Return $(X,Y)$\;\label{Line:couple-return}
}
\eIf{$\+F^\tau\not\subseteq \+E^\sigma$ \label{Line:couple-swap-cond}}{
   choose the smallest $c\in \+F^\tau\setminus \+E^\sigma$\; \label{Line:couple-choose}
\eWP{$\mu_{\+E}^{\sigma}(c)$ \label{Line:couple-psat-cond}}{
    \Return $\Couple\left(\+E\cup\{c\},\+F,\sigma,\tau\right)$\; \label{Line:couple-psat-return}
}
{\label{Line:couple-else}
  
   let $\pi=\vio{c}$ and draw a random $\rho\sim \mu_{\+F,\vbl(c)}^\tau$\;\label{Line:couple-sample}
    \Return $\Couple\left(\+E,\+F, \sigma\land\pi, \tau\land\rho\right)$\; \label{Line:couple-assign-return}
}}
{\label{Line:couple-else-1}
choose the smallest $c\in \+E^\sigma\setminus \+F^\tau$\; \label{Line:couple-choose-1}
\eWP{$\mu_{\+F}^{\tau}(c)$ \label{Line:couple-psat-cond-1}}{
    \Return $\Couple\left(\+E,\+F\cup \{c\},\sigma,\tau\right)$\; \label{Line:couple-psat-return-1}
}
{
    draw a random $\pi\sim \mu_{\+E,\vbl(c)}^\sigma$ and let $\rho=\vio{c}$ \; \label{Line:couple-sample-1}
    \Return $\Couple\left(\+E,\+F, \sigma\land\pi, \tau\land\rho\right)$\; \label{Line:couple-assign-return-1}
}
}
\end{algorithm}

The validity of this coupling is ensured by the following proposition. 
A similar correctness result was proven under a stronger local lemma condition in~\cite[Lemma 3.3]{WY24}.

%The following correctness guarantee is given in ~\cite{WY24}. We include the proof for completeness.

\begin{proposition}\label{lemma:coupling-correctness}
Assume that the atomic CSP formula $\Phi=(V,[q],\+C)$ is satisfiable.
%Let atomic CSP formula $\Phi=(V,[q],\+C)$ satisfy \Cref{condition:main-condition}. 
For any constraint $c_0\in \+C$, the procedure $\Couple(\+C\setminus \{c_0\},\+C,\varnothing,\varnothing)$ terminates with probability 1 and returns a pair of random assignments $(X,Y)\in [q]^V\times [q]^V$ such that marginally $X\sim \mu_{\+C\setminus\{c_0\}}$ and $Y\sim\mu_{\+C}$.
\end{proposition}

\begin{proof}[Proof sketch]
The proof follows the same inductive framework as the proof of \cite[Lemma 3.3]{WY24}. 
We provide a brief outline here.

First, by applying structural induction in the top-down order of recursion, one can verify the following induction hypothesis for each recursive call $\Couple(\+E,\+F,\sigma,\tau)$:
\[
\Lambda(\sigma)=\Lambda(\tau),\quad \+P[\+E\land \sigma]>0, \quad\text{ and }\quad \+P[\+F\land \tau]>0.
\]
This induction holds as long as the initial formula with constraints $\+C$ is satisfiable,
and it ensures that the coupling procedure remains well-defined throughout the recursion.

Next, observe that  in each recursive step, either the size of the symmetric difference $\+E^{\sigma}\triangle \+F^{\tau}$ decreases by one, or the number of unassigned variables in $\sigma$ and $\tau$ is reduced by at least one. 
Since \Cref{Alg:couple} terminates when $\+E^{\sigma}=\+F^{\tau}$, the procedure $\Couple(\+E,\+F,\sigma,\tau)$ eventually terminates  due to the finiteness of both the number of constraints and the number of variables.

% Since the procedure is well-defined, we next observe that \Cref{Alg:couple} terminates when 
% $\+E^{\sigma}=\+F^{\tau}$. In each recursive step, either the size of $\+E^{\sigma}\triangle \+F^{\tau}$ 
%   decreases by one, or the number of unassigned variables in $\sigma$ and $\tau$
% reduces by at least one. Due to the finiteness of the number of constraints and variables, the procedure 
% $\Couple(\+E,\+F,\sigma,\tau)$ terminates eventually.

Finally, by applying structural induction in the bottom-up order of recursion, we can verify the following induction hypothesis to ensure the correctness of the coupling:
\[
\Couple(\+E,\+F,\sigma,\tau) \text{ returns an $(X,Y)$ such that marginally $X\sim\mu_{\+E}^{\sigma}$ and $Y\sim\mu_{\+F}^{\tau}$}.
\]
This induction follows the same steps as the one given in the proof of \cite[Lemma 3.3]{WY24}. 
The correctness holds as long as the coupling procedure is well-defined and terminates, which we have already established.
\end{proof}

%Note that as in \Cref{Alg:witness} we construct the witness tree in chronological order (whereas in ~\cite{moser2010constructive} the witness tree is constructed in the reverse-chronological order), the witness tree $T$ can actually be maintained on-the-fly during the run of $\Couple(\+C\setminus \{c_0\},\+C,\varnothing,\varnothing)$. To this end, we extract the operation from Lines \ref{Line:witness-loop}-\ref{Line:witness-add} of \Cref{Alg:witness} and define the following operator $\Join$ for a witness tree $T$ and an (unpinned) constraint $c\in \+C$ such that $T\Join \+C$ as the tree $T$ after executing Lines \ref{Line:witness-loop}-\ref{Line:witness-add} of \Cref{Alg:witness}. We also let $T\Join \+C$ be the single root $c$ if $T=\emptyset$. 

\begin{comment}

\end{comment}

\subsection{The witness tree and witness assignment}
The novelty of \Cref{theorem:correlation-decay} lies in an improved analysis of the coupling procedure (\Cref{Alg:couple}).
This new analysis of the coupling establishes an exponential decay of correlation by leveraging new structural properties of the underlying hypergraph for the CSP formula.
%In prior work~\cite{WY24}, the decay of correlation was established under a stronger local lemma condition. 
The main slackness in previous analysis of the coupling~\cite{WY24} arises from the use of a combinatorial structure called a $2$-tree, originally introduced in~\cite{Alon91}, serving as a witness for a large discrepancy between satisfying assignments. 
To achieve an improved analysis, we employ a witness tree structure, 
constructed similarly to the witness tree used in the analysis of the Moser-Tardos algorithm~\cite{moser2010constructive}.
This  replaces the 2-tree with a new certificate for the discrepancy in the coupling. 
This use of the Moser-Tardos witness tree is crucial for approaching the satisfiability threshold.

%Note that in \Cref{Alg:couple}, the discrepancy between the outcomes is restricted to variables on constraints that are accessed by the algorithm. It is helpful to record the (original) constraints accessed by the algorithm.  

We begin with a formal definition of the execution log for \Cref{Alg:couple}.
\begin{definition}[execution log]\label{definition:execution-log}
    Given a run of \Cref{Alg:couple} from %$\Couple(\+E,\+F,\sigma,\tau)$,
$\Couple(\+C\setminus \{c_0\},\+C,\varnothing,\varnothing)$,
the \emph{execution log} $L=L(\+C,c_0)=(c_1,c_2,\dots,c_{\ell})$ is a random sequence of (unpinned) constraints from $\+C$, constructed as follows:
\begin{itemize}
    \item initialize $L$ as the empty list;
    \item whenever $\Couple(\+E,\+F, \sigma\land \pi, \tau\land \rho)$ is recursively called at \Cref{Line:couple-assign-return} or \Cref{Line:couple-assign-return-1},
    append the original unpinned constraint $\unpin{c}\in \+C$ to the end of $L$.
\end{itemize}
For any sequence $(c_1,c_2,\dots,c_{\ell})$ of (unpinned) constraints from $\+C$, we say that $(c_1,c_2,\dots,c_{\ell})$ is a \emph{proper execution log} if it appears in the support of $L=L(\+C,c_0)$, i.e.,
\[
\Pr{L=(c_1,c_2,\dots,c_{\ell})}>0.
\]
\end{definition}

%The execution log $L=L(\+C,c_0)$ defines a stochastic process on states space $\+C$, and a sequence $(c_1,c_2,\dots,c_{\ell})$ is a proper execution log if it is in the support of $L$.

%From \Cref{definition:execution-log}, one can see that each execution of \Cref{Alg:couple} corresponds to exactly one proper execution log. However, one proper execution log may correspond to multiple executions of \Cref{Alg:couple}. 

Recall the definition of the incidence hypergraph in \Cref{definition:incidence-hypergraph}. The following notion of a witness tree is inspired by the witness tree introduced for analyzing the Moser-Tardos algorithm~\cite{moser2010constructive}.
A key distinction in our context is that in the execution log constructed as in \Cref{definition:execution-log}, a constraint can never appear twice.
Thus, the witness tree here is a subgraph of $G_{\Phi}$.

\begin{definition}[witness tree]\label{definition:witness-tree}
A \emph{witness tree} $T$ is a finite rooted tree with vertex set $V(T)\subseteq \+C$ such that $T$ is a subgraph of $G_{\Phi}$.

Given a witness tree $T$ and a constraint $c\in \+C$, define $T\oplus c$ as the witness tree $T'$ obtained as follows:
\begin{itemize}
\item if $T=\emptyset$, then $T'$ is the rooted tree containing a single vertex $c$;
\item otherwise, if $\exists c'\in V(T) $ such that $\vbl(c)\cap \vbl(c')\neq \emptyset$, then $T'$ is obtained by adding $c$ as a child of the deepest such $c'$ in $T$, breaking ties by choosing the lexicographically smallest $c'$ in case multiple $c'$s with the same largest depth exist;
\item otherwise, $T'=T$.
\end{itemize}
\label{definition:witness-tree-by-execution-log}
Given a proper execution log $L=(c_1,c_2,\dots,c_{\ell})$, the \emph{witness tree of the execution log $L$}, denoted by $T=T(L)$, 
is defined as follows: Let $T_i=T_{i-1}\oplus c_i$ for $1\le i\le \ell$ and initialize $T_0=\emptyset$.
Thus, $T(L)=T_{\ell}$.

% \begin{align}
% T\oplus c\defeq\begin{cases} \text{a single root $c$} & T=\emptyset \\ 
% \text{The tree $T$ after Lines \ref{Line:witness-condition}-\ref{Line:witness-add} of \Cref{Alg:witness}} & \text{otherwise};\end{cases}
% \end{align}

% \begin{algorithm}
% \caption{Construction of witness tree from execution log} \label{Alg:witness}
%   \SetKwInput{KwPar}{Parameter}
% \KwIn{a proper execution log $L=(c_1,c_2,\dots,c_{\ell})$;} 
% \KwOut{a witness tree $T$.}
% Initialize $T$ as the rooted tree that contains a single vertex $c_1$\;
% \For{$i=2$ to $\ell$\label{Line:witness-loop}}{
%     \If{$\exists c\in V(T) $ s.t. $\vbl(c)\cap \vbl(c_i)\neq \emptyset $\label{Line:witness-condition}}{
%         Choose such $c$ with the largest depth, breaking ties lexicographically\;\label{Line:witness-choose}
%         Add $c_i$  as a child of $c$ in $T$\;\label{Line:witness-add}
%     }
% }
% \Return $T$\;
% \end{algorithm}

%Moreover, given any witness tree $\tau$ any proper bad sequence $B$, we say that $\tau$ occurs in $B$ if there exists $1\leq t\leq \abs{B}$ such that $\tau(B,t) = \tau$.
\end{definition}

%Note that \Cref{definition:witness-tree-by-execution-log} is well-defined, i.e., the $T$ produced by \Cref{Alg:witness} is always a witness tree. This is by that all edges $(c,c_i)$ added at \Cref{Line:witness-add} of \Cref{Alg:witness} are edges in  $\-{Lin}(H_{\Phi})$ by the condition at \Cref{Line:witness-condition} of \Cref{Alg:witness}.

%The following lemma ensures the well-definedness of witness trees and guarantees that the witness tree produced by a proper execution log includes all appeared constraints as its vertices.

The following proposition states several basic properties of the witness tree. The proof is straightforward and is therefore omitted.

\begin{proposition}\label{lemma:witness-tree-size}\label{lemma:witness-local-total-ordering}
Let  $L$ be a proper execution log, and let $T=T(L)$ be the rooted tree constructed by \Cref{definition:witness-tree}.
For each $c\in T$, denote the depth of $c$ in $T$ as $\-{dep}(c)$. The following properties hold:
\begin{enumerate}
    \item $T$ is a witness tree,  i.e.,  $T$ is a rooted tree and is a subgraph of $G_{\Phi}$.
    \item $|V(T)|=|L|$, meaning that every constraint in $L$ is included as a vertex in $T$.
    \item For any distinct $c, c'\in V(T)$ with $\-{dep}(c)=\-{dep}(c')$, it holds  that $\vbl(c)\cap \vbl(c')=\emptyset$.\label{item:witness-local-total-ordering-1}
    \item For any $c,c'\in V(T)$, if $\vbl(c)\cap \vbl(c')\neq \emptyset$ and $c'$ appears later than $c$ in the execution log $L$, then $\-{dep}(c')>\-{dep}(c)$.  \label{item:witness-local-total-ordering-2}
\end{enumerate}
\end{proposition}

Besides the witness tree, another main source of randomness in the coupling procedure of \Cref{Alg:couple} comes from the non-violating partial assignments drawn in \Cref{Line:couple-sample,Line:couple-sample-1} of \Cref{Alg:couple}. 
This is formalized by the following notion of witness assignment.

\begin{definition}[witness assignment]\label{definition:cumulative-non-violating-assignment}
Given a run of \Cref{Alg:couple} from $\Couple(\+C\setminus \{c_0\},\+C,\varnothing,\varnothing)$,
the \emph{witness assignment}
is a random partial assignment  $\varsigma$ constructed as follows:
\begin{itemize}
    \item initially, $\varsigma=\varnothing$, the empty assignment;
    \item whenever  a recursive call $\Couple(\+E,\+F,\sigma\land \pi,\tau\land \rho)$ is made in \Cref{Line:couple-assign-return}  of \Cref{Alg:couple}, with $\pi=\vio{c}$ being the unique violating assignment of constraint $c$, update $\varsigma\gets \varsigma\land \rho$;
    \item whenever a recursive call $\Couple(\+E,\+F,\sigma\land \pi,\tau\land \rho)$ is made in \Cref{Line:couple-assign-return-1}  of \Cref{Alg:couple}, with $\rho=\vio{c}$ being the unique violating assignment of constraint $c$, update $\varsigma\gets \varsigma\land \pi$. 
\end{itemize}
\end{definition}

For any witness tree $T$, we define
\[
\vbl(T)\defeq\bigcup\limits_{c\in V(T)}\vbl(c),
\]
which represents the set of all variables involved in the constraints that appear as vertices in $T$.
Based on  \Cref{lemma:witness-tree-size,definition:cumulative-non-violating-assignment}, 
it follows that the witness assignment $\varsigma$ is a partial assignment specified over the variables in $\vbl(T)$, i.e.~$\varsigma\in [q]^{\vbl(T)}$, where $T=T(L)$ is the witness tree of the execution log $L$ that belongs to a run of \Cref{Alg:couple} from $\Couple(\+C\setminus \{c_0\},\+C,\varnothing,\varnothing)$.

The witness tree, together with the witness assignment, fully determines the random choices made by \Cref{Alg:couple}, except for the optimal coupling step at the end of the recursion.
%This is formally stated by the following lemma.

\begin{lemma}\label{lemma:conditional-deterministic}
Let $T$ be any tree in $G_{\Phi}$ rooted at $c_0$, and let $\varsigma\in [q]^{\vbl(T)}$ be a partial assignment. 
Given the witness tree of the execution log being $T$ and the witness assignment being $\varsigma$, 
a run of \Cref{Alg:couple} from $\Couple(\+C\setminus \{c_0\},\+C,\varnothing,\varnothing)$ is fully determined, 
except for the random choices made in \Cref{Line:couple-perfect}.
\end{lemma}

\begin{proof}
For each $c\in V(T)$, define $\mathrm{var}(c)$ as the set of variables within $\vbl(c)$,  excluding all  variables within $\vbl(c')$ for any $c'\in V(T)$ such that $\-{dep}(c')<\-{dep}(c)$ in $T$ and $\vbl(c')\cap \vbl(c)\neq \emptyset$.
% \[
%     \-{pre}(c)\defeq \{c'\in V(T)\mid \-{dep}(c')<\-{dep}(c)\land \vbl(c')\cap \vbl(c)\neq \emptyset\},
% \]
% and
% \[
%     \mathrm{var}(c)\defeq \vbl(c)\setminus \bigcup\limits_{\substack{c'\in V(T)\\\-{dep}(c')<\-{dep}(c)\\\vbl(c')\cap \vbl(c)\neq \emptyset}}\vbl(c').
% \] 
Note that  $\mathrm{var}(c)$ is determined by the witness tree $T$.

Then, the random choices in \Cref{Line:couple-psat-cond} (and \Cref{Line:couple-psat-cond-1}) can be determined by checking whether $\unpin{c}\in V(T)$ and $\vbl(c)=\mathrm{var}(\unpin{c})$ for the pinned constraint $c\in\+C^*$ picked in \Cref{Line:couple-choose} (and \Cref{Line:couple-choose-1}). Recall that $\unpin{c}$ denotes the original unpinned constraint from $\+C$ corresponding to $c$.

The random assignments generated in \Cref{Line:couple-assign-return,Line:couple-assign-return-1} are recorded by the witness assignment $\varsigma$, and thus can be fully recovered from $\varsigma$.

Therefore, a run of \Cref{Alg:couple} from $\Couple(\+C\setminus \{c_0\},\+C,\varnothing,\varnothing)$ can be deterministically simulated, except for the random assignments generated in \Cref{Line:couple-perfect}.
%
% We can deterministically simulate \Cref{Alg:couple} (except for the random assignment generated in \Cref{Line:couple-perfect}) as follows:
% \begin{enumerate}
%     \item For each pinned constraint $c$ chosen in \Cref{Line:couple-choose} or \Cref{Line:couple-choose-1} of \Cref{Alg:couple} , the condition in \Cref{Line:couple-psat-cond} is satisfied if and only if $\unpin{c}\in V(T)$ and $\vbl(c)=\-{var}(\unpin{c})$; \label{item:witness-tree-conditioned-probability-bound-1}
%     \item whenever a recursive call  $\Couple(\+E,\+F,\sigma\land \pi,\tau\land \rho)$ is made in \Cref{Line:couple-assign-return}  of \Cref{Alg:couple}, it must hold that $\rho=\varsigma_{\vbl(c)}$;\label{item:witness-tree-conditioned-probability-bound-2}
%     \item whenever a recursive call  $\Couple(\+E,\+F,\sigma\land \pi,\tau\land \rho)$ is made in \Cref{Line:couple-assign-return-1}  of \Cref{Alg:couple},  it must hold that $\pi=\varsigma_{\vbl(c)}$.\label{item:witness-tree-conditioned-probability-bound-3}
% \end{enumerate}
% For \Cref{item:witness-tree-conditioned-probability-bound-1}, the only if direction is easy to see. For the if direction, note that any simplification of $\unpin{c}$ may reappear in $\+E^{\sigma}\triangle \+F^{\tau}$ during the run of \Cref{Alg:couple} only when some variables on $c$ gets assigned. Items \ref{item:witness-tree-conditioned-probability-bound-2} and \ref{item:witness-tree-conditioned-probability-bound-3} are direct by \Cref{definition:cumulative-non-violating-assignment}.
\end{proof}

\subsection{Analysis of the coupling}

In this subsection, we utilize the witness tree and witness assignment defined previously to prove \Cref{theorem:correlation-decay}.

% Recall the definition of execution logs in \Cref{definition:execution-log} and the procedure of constructing witness trees built from execution logs in \Cref{definition:witness-tree-by-execution-log}. Note that the execution logs (with respect to \Cref{Alg:couple}) is built by appending constraints one at a time on-the-fly, we can also directly build witness trees on-the-fly (with respect to \Cref{Alg:couple}) through the following operator. For a witness tree $T$ and a constraint $c\in \+C$, we define
% \[
% T\oplus c\defeq\begin{cases} \text{A single root $c$} & T=\emptyset \\ 
% \text{The tree $T$ after Lines \ref{Line:witness-condition}-\ref{Line:witness-add} of \Cref{Alg:witness}} & \text{otherwise};\end{cases}
% \]

The following definition describes a random process to simulate \Cref{Alg:couple} using both the witness tree and witness assignment,
employing the principle of deferred decisions.

\begin{definition}[$M$-truncated process for simulating \Cref{Alg:couple} with explicitly identified randomness]\label{definition:coupling-equivalent-process}
Let  $\*X\sim \mu_{\+C\setminus \{c_0\}}$ and $\*Y\sim \mu_{\+C}$ be drawn independently beforehand. 
Define the random process 
$$P^{\cp}=\left\{(\+E_t,\+F_t,\sigma_t,\tau_t,T_t,\varsigma_t)\right\}_{t\ge 0}$$ 
starting from the initial state $(\+E_0,\+F_0,\sigma_0,\tau_0,T_t,\varsigma_t)=(\+C\setminus \{c_0\},\+C,\varnothing,\varnothing,\emptyset,\varnothing)$ as follows:

    \begin{enumerate}
        \item If $\+{E}_t^{\sigma_t}=\+F_t^{\tau_t}$ or $|V(T_t)|=M$, the process stops, and $(\+E_t,\+F_t,\sigma_t,\tau_t,T_t,\varsigma_t)$ is the outcome.
        \item Otherwise, if \(\+F_t^{\tau_t} \not\subseteq \+E_t^{\sigma_t}\),
        let \(c\) be the smallest pinned constraint in \(\+F_t^{\tau_t} \setminus \+E_t^{\sigma_t}\). 
        We then set \((\+E_{t+1}, \+F_{t+1}, \sigma_{t+1}, \tau_{t+1}, T_{t+1}, \varsigma_{t+1})\) as
       \[
        \begin{cases}
        \left(\+E_{t}\cup \{c\},\+F_{t},\sigma_{t},\tau_{t},T_{t},\varsigma_t\right)   & 
         \text{if $c$ is satisfied by $\*X$};\\%:\\ &\quad \text{let $\pi=\*X_{\vbl(c)}$ and $\rho=\*Y_{\vbl(c)}$};\\
             \left(\+E_{t},\+F_{t},\sigma_{t}\land \*X_{\vbl(c)},\tau_{t}\land \*Y_{\vbl(c)},T_{t}\oplus \unpin{c},\varsigma_t\land \*Y_{\vbl(c)}\right)  & \text{otherwise}.
        \end{cases}
        \]\label{item:coupling-equivalent-process-2}
        Here, $\unpin{c}\in\+C$ denotes the original unpinned version of $c$, and $T_{t}\oplus \unpin{c}$ is given in \Cref{definition:witness-tree}.
        \item Otherwise, if $\+F_{t}^{\tau_t}\subseteq \+E_{t}^{\sigma_t}$, 
        let $c$ be the smallest pinned constraint in $\+E_{t}^{\sigma_t}\setminus \+F_{t}^{\tau_t}$.
        We then set $(\+E_{t+1},\+F_{t+1},\sigma_{t+1},\tau_{t+1},T_{t+1},\varsigma_{t+1})$ as
        \[
        \begin{cases}
         \left(\+E_{t},\+F_{t}\cup \{c\},\sigma_{t},\tau_{t},T_{t},\varsigma_t\right)   & 
         \text{if $c$ is satisfied by $\*Y$};\\%:\\ &\quad \text{let $\pi=\*X_{\vbl(c)}$ and $\rho=\*Y_{\vbl(c)}$};\\
         \left(\+E_{t},\+F_{t},\sigma_{t}\land \*X_{\vbl(c)},\tau_{t}\land \*Y_{\vbl(c)},T_{t}\oplus \unpin{c},\varsigma_t\land \*X_{\vbl(c)}\right)   & \text{otherwise}.
        \end{cases}
        \]\label{item:coupling-equivalent-process-3}
    \end{enumerate}
Let $\mu^{\cp}$ denote the distribution of the outcome $(\+E^{\cp},\+F^{\cp},\sigma^{\cp},\tau^{\cp},T^{\cp},\varsigma^{\cp})$ of this process, and let $\+L^{\cp}=\supp(\mu^{\cp})$ be its support. Let $\+V^{\cp}$ denote the set of all possible $(\+E,\+F,\sigma,\tau,T,\varsigma)$ such that 
$$\Pr{(\+E,\+F,\sigma,\tau,T,\varsigma)\in P^{\cp}}>0.$$

%Let $\+L^{\cp}_K$ be the set of ``truncated'' outcomes, i.e.:
%\[
%\+L^{\cp}_K=\left\{(\+E,\+F,\sigma,\tau,T,\varsigma)\mid \mu^{\cp}(\+E,\+F,\sigma,\tau,T,\varsigma)>0\land |V(T)|= K\right\}.
%\]
%At last, let $\+V^{\cp}$ denote the set of all possible $(\+E,\+F,\sigma,\tau,T,\varsigma)$ with $\Pr{(\+E,\+F,\sigma,\tau,T,\varsigma)\in P^{\cp}}>0$.
\end{definition}

\begin{remark}[distinctions between \Cref{Alg:couple} and \Cref{definition:coupling-equivalent-process}]
We note two main distinctions between \Cref{Alg:couple} and \Cref{definition:coupling-equivalent-process}, 
specifically regarding truncation and the explicitly identified randomness using the principle of deferred decisions. 
Both aspects are beneficial for our analyses of the coupling and algorithmic applications.

\begin{itemize}
    \item In \Cref{definition:coupling-equivalent-process}, 
    the process is truncated once the size of the witness tree reaches the threshold $M$, 
    whereas \Cref{Alg:couple} does not incorporate such truncation.
    \item In \Cref{definition:coupling-equivalent-process}, 
    all randomness utilized by the process $P^\cp$ is identified with the two pre-generated random assignments $\*X\sim \mu_{\+C\setminus {c_0}}$ and $\*Y\sim \mu_{\+C}$. 
    In contrast, \Cref{Alg:couple} generates random choices at the moment they are needed. 
    However, by the principle of deferred decisions, \Cref{definition:coupling-equivalent-process} still faithfully simulates \Cref{Alg:couple} (truncated when the size of the witness tree reaches $M$). 
    Thus, properties we proved for \Cref{Alg:couple} (specifically \Cref{lemma:witness-tree-size} and \Cref{lemma:conditional-deterministic}) also hold for the random process constructed in \Cref{definition:coupling-equivalent-process}.
\end{itemize}
\end{remark}

The above observation is formalized by the following lemma that upper bounds the correlation decay of the coupling by the probability of the truncation of the random process $P^{\cp}$. 
The proof of this lemma is similar to that of \cite[Lemma 3.9]{WY24}, and is included for completeness.

\begin{lemma}\label{lemma:hamming-transform}
Assume $\Phi$ is satisfiable.
%Assume the conditions of \Cref{theorem:correlation-decay}. 
Let $(X,Y)$ be the output of $\Couple(\+C\setminus \{c_0\},\+C,\varnothing,\varnothing)$, and 
let $(\+E^{\cp},\+F^{\cp},\sigma^{\cp},\tau^{\cp},T^{\cp},\varsigma^{\cp})\sim\mu^{\cp}$ be the outcome of the process $P^{\cp}$. 
Then, we have
\[
\Pr{d_{\mathrm{Ham}}(X,Y)\geq k\cdot M}\leq \Pr{|V(T^{\cp})|= M}.
\]
\end{lemma}

\begin{proof}
By the construction of witness tree in \Cref{definition:witness-tree} and \Cref{lemma:witness-tree-size}, each time a recursive call is made to $\Couple(\+E,\+F,\sigma\land \pi, \tau\land \rho)$ at \Cref{Line:couple-assign-return} or \Cref{Line:couple-assign-return-1} of \Cref{Alg:couple}, at most $k$ variables are assigned into $\sigma$ and $\tau$, and the size of the current witness tree $T=T(L)$ increases by one. 
Furthermore, by \Cref{Line:couple-return} of \Cref{Alg:couple}, the Hamming distance between $X$ and $Y$ is upper bounded by the total number of variables assigned in $\sigma$ and $\tau$ during the recursion $\Couple(\+C\setminus \{c_0\},\+C,\varnothing,\varnothing)$. 
Hence, we have
\[
d_{\mathrm{Ham}}(X,Y)\geq k\cdot M\implies |V(T)|\geq M.
\]
We now argue that the witness tree $T=T(L)$, constructed from the execution log $L$ of a run of \Cref{Alg:couple} from $\Couple(\+C\setminus \{c_0\},\+C,\varnothing,\varnothing)$, can be coupled with the witness tree $T^{\cp}$ constructed during the process $P^{\cp}$ described in \Cref{definition:coupling-equivalent-process}.
Specifically, whenever a recursive call to $\Couple(\+E,\+F,\sigma,\tau)$ is made in \Cref{Alg:couple}, with the current witness tree being $T$ and witness assignment  $\varsigma$, 
the process $P^{\cp}$ moves to a state $(\+E,\+F,\sigma,\tau,T,\varsigma)$. 
Under this coupling, the lemma follows immediately.

It is important to note that this coupling between the two processes, \Cref{Alg:couple} and $P^{\cp}$, does not hold trivially, 
because the randomness used in the construction of $P^{\cp}$ is explicitly identified to the random satisfying assignments:
\[
\*X\sim \mu_{\+C\setminus \{c_0\}}\quad \text{ and }\quad\*Y\sim \mu_{\+C}.
\]
Nevertheless, this perfect coupling between \Cref{Alg:couple} and $P^{\cp}$ can be verified by structural induction, proceeding in the top-down order of recursion, with a strengthened hypothesis:
conditioning on the process $P^{\cp}$ being  at the state $(\+E,\+F,\sigma,\tau,T,\varsigma)$, \Cref{Alg:couple} is at the same state, and it further holds:
\begin{equation}\label{eq:initial-assignment-distribution}
\*X\sim \mu^{\sigma}_{\+E}\quad \text{ and }\quad\*Y\sim \mu^{\tau}_{\+F}.
\end{equation}
With \eqref{eq:initial-assignment-distribution}, the transition probabilities to the next state are identical in both processes. Therefore, the two processes are perfectly coupled. Additionally, it can be verified that \eqref{eq:initial-assignment-distribution} continues to hold for each possible branch of the process $P^{\cp}$. This concludes the proof of the lemma.
\end{proof}

%It remains to bound $\Pr{|V(T^{\cp})|= K}$ to prove \Cref{theorem:correlation-decay}. Here, we further stipulate $T^{\cp}$ and $\varsigma^{\cp}_{K}$ as the $T_K$ and $\varsigma_K$ in the process in \Cref{definition:coupling-equivalent-process}, that is, the witness tree and cumulative assignment after $K$ steps of the process.  A lemma crucial to the proof of \Cref{theorem:correlation-decay} is the following tail bound, which bounds the probability of a certain witness tree and a certain cumulative non-violating assignment appearing after first $K$ steps. 

The proof of \Cref{theorem:correlation-decay} now reduces to establishing an upper bound of $2^{-M}$ on the truncation probability $\Pr{|V(T^{\cp})|= M}$.
A crucial step in this process is the following technical lemma, 
which assumes \Cref{condition:main-condition} 
and will also play a key role in the algorithmic implication of the coupling later.
%upper bound for the probability of a specific witness tree in conjunction with a specific witness assignment.

%Before proving \Cref{lemma:witness-tree-conditioned-probability-bound}, we need the following technical lemma.

\begin{lemma}\label{lemma:assignment-probability-bound}
%Assume \Cref{condition:main-condition}. 
Assume $\Phi$ is satisfiable.
Let $\*X\sim \mu_{\+C\setminus \{c_0\}}$ and $\*Y\sim \mu_{\+C}$.
For any subset of variables $S\subseteq V$, 
\begin{align*}
    \forall X\in [q]^{S},
    &&\mu_{\+C\setminus \{c_0\}}(X)
    &=\Pr{\*X_S=X}\leq q^{-|S\cap \vgood|}\cdot (1-\mathrm{e}q^{-(1-\epsilon_1)k})^{-|S|p_1\alpha};\\
    \forall Y\in [q]^{S},
    &&\mu_{\+C}(Y)
    &=\Pr{\*Y_S=Y}\leq q^{-|S\cap \vgood|}\cdot (1-\mathrm{e}q^{-(1-\epsilon_1)k})^{-|S|p_1\alpha}.
\end{align*}
\end{lemma}

\begin{proof}
    We only prove the first inequality, and the second one follows similarly. For any $\tau\in [q]^{\vbad}$ with $\Pr{\*X_{\vbad}=\tau}>0$, the following bound holds:
\begin{equation}\label{eq:good-vertices-assignment-probability-bound}
\Pr{\*X_{\vbl(T)\cap \vgood}=X_{\vbl(T)\cap \vgood}\mid \*X_{\vbad}=\tau}\leq q^{-|\vbl(T)\cap \vgood|}\cdot (1-\mathrm{e}q^{-(1-\epsilon_1)k})^{-|\vbl(T)|p_1\alpha}.
\end{equation}
This is because, given the pinning $\*X_{\vbad}=\tau$, the resulting pinned formula has each variable with a degree of at most $p_1 d$ (by \Cref{fact:good-vertex-degree-bound}) and each constraint containing at least $(1-\epsilon_1)k$ variables (by \Cref{fact:good-vertex-fraction}).
Applying \Cref{HSS} by setting the parameter $x(c)=\mathrm{e}q^{-(1-\epsilon_1)k}$ for each constraint $c$ in the pinned formula, we derive the bound in~\eqref{eq:good-vertices-assignment-probability-bound}.

Now,  we can bound the probability of $\*X_{\vbl(T)}= X$ as follows:
\begin{align*}
\Pr{\*X_{\vbl(T)}= X}= &\sum\limits_{\tau\in [q]^{\vbad}}\Pr{\*X_{\vbad}=\tau}\cdot \Pr{\*X_{\vbl(T)\cap \vgood}=X_{\vbl(T)\cap \vgood}\mid \*X_{\vbad}=\tau}\\
(\text{by \eqref{eq:good-vertices-assignment-probability-bound}})\quad &\leq q^{-|\vbl(T)\cap \vgood|}\cdot (1-\mathrm{e}q^{-(1-\epsilon_1)k})^{-|\vbl(T)|p_1\alpha}.\qedhere
\end{align*}
\end{proof}

% \begin{lemma}\label{lemma:witness-tree-conditioned-probability-bound}
%  For any rooted tree $T$ in $G_{\Phi}$ and any $\varsigma\in [q]^{\vbl(T)}$, we have
% \[
% \Pr{T^{\cp}=T\land \varsigma^{\cp}=\varsigma}\leq q^{-2|\vbl(T)\cap \vgood|}\cdot \left(1-\mathrm{e}q^{-(1-\eps_1)k}\right)^{-2|\vbl(T)|p_1d}.
% \]
% \end{lemma}
% \begin{proof}[Proof of \Cref{lemma:witness-tree-conditioned-probability-bound}]
% Fix any rooted tree $T$ in $G_{\Phi}$ and partial assignment $\varsigma\in [q]^{\vbl(T)}$.
% We may assume $\Pr{T^{\cp}=T\land \varsigma^{\cp}=\varsigma}>0$, because otherwise, the lemma would hold trivially.

% Recall the pre-generated random assignments $\*X\sim  \mu_{\+C\setminus \{c_0\}}$ and $\*Y\sim \mu_{\+C}$, used in the construction of $T^{\cp}$ and $\varsigma^{\cp}$.
% %
% By \Cref{lemma:conditional-deterministic},  there are two partial assignments $\bm{x},\bm{y}\in [q]^{\vbl(T)}$ determined by $T$ and $\varsigma$ such that 
% \begin{equation}\label{eq:assignment-implication}
% T^{\cp}=T\land \varsigma^{\cp}=\varsigma\implies \*X_{\vbl(T)}= \bm{x}\land \*Y_{\vbl(T)}= \bm{y}.
% \end{equation}

% Since $\*X$ and $\*Y$ are independent,  the lemma follows from \eqref{eq:assignment-implication} and \Cref{lemma:assignment-probability-bound}
% \end{proof}

We are now ready to prove \Cref{theorem:correlation-decay}, the main theorem of this section.

\begin{proof}[Proof of \Cref{theorem:correlation-decay}]
Consider the outcome $(\+E^{\cp},\+F^{\cp},\sigma^{\cp},\tau^{\cp},T^{\cp},\varsigma^{\cp})\sim\mu^{\cp}$ of the random process $P^{\cp}$ in \Cref{definition:coupling-equivalent-process}.
We first show that %assuming \Cref{condition:main-condition},
for any rooted tree $T$ in $G_{\Phi}$ and any $\varsigma\in [q]^{\vbl(T)}$, we have
\begin{align}\label{eq:witness-tree-conditioned-probability-bound}
\Pr{T^{\cp}=T\land \varsigma^{\cp}=\varsigma}\leq q^{-2|\vbl(T)\cap \vgood|}\cdot \left(1-\mathrm{e}q^{-(1-\eps_1)k}\right)^{-2|\vbl(T)|p_1\alpha}.
\end{align}
%Given any rooted tree $T$ in $G_{\Phi}$ and any $\varsigma\in [q]^{\vbl(T)}$, we may assume $\Pr{T^{\cp}=T\land \varsigma^{\cp}=\varsigma}>0$, because otherwise, the lemma would hold trivially.
Recall the pre-generated $\*X\sim  \mu_{\+C\setminus \{c_0\}}$ and $\*Y\sim \mu_{\+C}$, used in the construction of the process $P^{\cp}$ and $\varsigma^{\cp}$.
By \Cref{lemma:conditional-deterministic},  there are two partial assignments $X,Y\in [q]^{\vbl(T)}$ determined by $T$ and $\varsigma$ such that 
\begin{equation}\label{eq:assignment-implication}
T^{\cp}=T\land \varsigma^{\cp}=\varsigma\implies \*X_{\vbl(T)}= X\land \*Y_{\vbl(T)}= Y.
\end{equation}
Since $\*X$ and $\*Y$ are independent, the inequality in \eqref{eq:witness-tree-conditioned-probability-bound} follows from \eqref{eq:assignment-implication} and \Cref{lemma:assignment-probability-bound}.

We proceed to prove \Cref{theorem:correlation-decay}.
According to \Cref{definition:coupling-equivalent-process},
all nonempty witness trees $T^{\cp}$ must be rooted at $c_0$. 
Let $\=T^{c_0}_M$ denote the set of witness trees of size $M$ rooted at $c_0$. 
Then, we have
\begin{align}
&\Pr{d_{\-{Ham}}(X,Y)\geq k\cdot M} \nonumber\\
(\text{by \Cref{lemma:hamming-transform}})\quad \leq &  \Pr{|V(T^{\cp})|= M} \nonumber\\
 \leq & \sum\limits_{T\in \=T^{c_0}_M}\sum\limits_{\varsigma\in [q]^{\vbl(T)}}  \Pr{T^{\cp}=T\land \varsigma^{\cp}=\varsigma} \nonumber\\
(\text{by \eqref{eq:witness-tree-conditioned-probability-bound}})\quad \leq & \sum\limits_{T\in \=T^{c_0}_M}\sum\limits_{\varsigma\in [q]^{\vbl(T)}}  q^{-2|\vbl(T)\cap \vgood|}\cdot \left( 1-\mathrm{e}q^{-(1-\eps_1)k} \right)^{-2|\vbl(T)|p_1\alpha} \nonumber\\
\leq &  \sum\limits_{T\in \=T^{c_0}_M}  q^{|\vbl(T)|}\cdot q^{-2|\vbl(T)\cap \vgood|}\cdot \left( 1-\mathrm{e}q^{-(1-\eps_1)k} \right)^{-2|\vbl(T)|p_1\alpha}. \label{eq:used-by-LP}
\end{align}
Note that for $T\in \=T^{c_0}_M$, 
the size of ${\vbl(T)}$ can be easily upper bounded by $kM$. 
We then lower bound $\abs{\vbl(T)\cap \vgood}$ as follows:
\begin{align*}
\abs{\vbl(T)\cap \vgood}
&=\abs{\vbl(T\cap \egood)\cap \vgood}\\
(\text{by \Cref{fact:good-vertex-fraction}})\quad  &\geq \abs{\vbl(T\cap \egood)}-\eps_1kM\\
(\text{by Properties~\ref{property:edge-expansion} and \ref{property:bounded-bad-fraction}})\quad &\geq (1-\eta)(1-\eps_2)kM-\eps_1kM.
\end{align*}
Therefore, we have
\begin{align*}
&\Pr{d_{\-{Ham}}(X,Y)\geq k\cdot M}\\
 \leq{} & \sum\limits_{T\in \=T^{c_0}_M} q^{-(2(1-\eta)(1-\eps_2)-(1+2\eps_1))kM}\cdot \left( 1-\mathrm{e}q^{-(1-\eps_1)k} \right)^{-2kMp_1\alpha}\\
(\text{by \Cref{property:bounded-neighbourhood-growth}})\quad \leq{} &  n^3\cdot (p_2\alpha)^{M} \cdot q^{-(2(1-\eta)(1-\eps_2)-(1+2\eps_1))kM} \cdot ( 1-\mathrm{e}q^{-(1-\eps_1)k} )^{-2kMp_1\alpha}\\
\leq{} & \left( 
        2^{3 \log n / M} 
    \cdot p_2\alpha 
    \cdot q^{-(2(1-\eta)(1-\eps_2)-(1+2\eps_1))k} 
    \cdot \left( 1-\mathrm{e} q^{-(1-\eps_1)k} \right)^{-2k p_1\alpha} 
    \right)^M.
\end{align*}
We upper bound each term in the bracket separately.
\begin{itemize}
    \item Since $M \ge \log n$, it holds $2^{3 \log n / M} \le 8$;
    \item Since $p_2 = \mathrm{e} k^2$, it holds $p_2\alpha = \mathrm{e} k^2 \alpha$;
    \item Since $\eta = \Theta(1/k)$, $\eps_1 = \Theta(1/k)$ and $\eps_2 = \Theta(1/k)$, we deduce that
    \begin{align*}
        2(1-\eta)(1-\eps_2)-(1+2\eps_1)
        \ge 1 - 2(\eta + \eps_1 + \eps_2) = 1 - \Theta\left( \frac{1}{k} \right),
    \end{align*}
    and hence it holds 
    \begin{align*}
        q^{-(2(1-\eta)(1-\eps_2)-(1+2\eps_1))k}
        \le q^{-k + O(1)};
    \end{align*}
    \item Since $\eps_1 = \Theta(1/k)$ and $p_1 = \Theta(k^7)$, for $k$ sufficiently large it holds
    \begin{align*}
        \left( 1-\mathrm{e} q^{-(1-\eps_1)k} \right)^{-2k p_1\alpha}
        &= \exp\left( \Theta\left( \mathrm{e} q^{-(1-\eps_1)k} \cdot 2k p_1\alpha \right) \right)\\
        &= \exp\left( \Theta\left( q^{-k + O(1)} k^8 \alpha \right) \right)\\
        &= O(1),
    \end{align*}
    assuming $\alpha \le \frac{q^k}{(q k)^c}$ for a sufficiently large universal constant $c \ge 8$.
\end{itemize}
Combining everything above, we conclude with
\begin{align*}
    \Pr{d_{\-{Ham}}(X,Y)\geq k\cdot M}
    \le \left( 8 \mathrm{e} k^2 \alpha \cdot q^{-k + O(1)} \cdot O(1) \right)^M
    \le  2^{-M},
\end{align*}
where we assume $\alpha \le \frac{q^k}{(q k)^c}$ for some sufficiently large universal constant $c>0$.
\end{proof}

\subsection{Proofs of replica symmetry and non-reconstruction}\label{sec:connectivity}
%\zc{Require polishing since connectivity is removed.}

In this subsection, we leverage the coupling analysis from previous sections to explore the correlation decay properties for random CSP formulas.  
%In particular, we examine the absence of long-range correlations of the solution space at the considered densities.
Specifically, we will establish the properties of replica symmetry (\Cref{theorem:replica-symmetry}) and non-reconstruction (\Cref{theorem:non-reconstruction}) for random CSPs at the considered densities.
Additionally, we will examine a connectivity property of the solution space, known as the \emph{looseness} property (which will be formally defined later in \Cref{definition:looseness}).

It is important to note that while these theorems are stated  in the context of random $k$-SAT, 
our proofs are applicable to a broader class of random CSPs, including random hypergraph colorings.

%\subsubsection{Absence of long-range correlations}

%We first establish the properties of replica symmetry (\Cref{theorem:replica-symmetry}) and non-reconstruction (\Cref{theorem:non-reconstruction}) for random CSPs.

\begin{proof}[Proof of Theorems \ref{theorem:replica-symmetry} and \ref{theorem:non-reconstruction}]
Recall the decay of correlation property established in \Cref{theorem:correlation-decay}, which implies both replica symmetry and non-reconstruction.

%A slight modification is needed for the proof, though, since for proving Theorems \ref{theorem:replica-symmetry} and \ref{theorem:non-reconstruction}, 
We apply the procedure $\Couple(\+E,\+F,\sigma,\tau)$ described in \Cref{Alg:couple} to construct a coupling between two instances obtained from the same set of constraints $\+C$,  differing only in the pinning on one variable.
Specifically, we consider the coupling procedure $\Couple(\+E,\+F,\sigma,\tau)$ with the following initial states: 
    \begin{itemize}
        \item The initial set of constraints are $\+E=\+F=\+C$;
        \item The initial assignments $\sigma,\tau$ are both specified on just $v$, with $\sigma(v)=x_1$ and $\tau(v)= x_2$.
    \end{itemize}
    \begin{comment}
    
\begin{definition}[coupling procedure with modified initial parameters]\label{definition:coupling-procedure-modified-initialization}
    For a variable $v\in V$, two values $x_1,x_2\in [q]$. we consider the coupling procedure $\Couple(\+E,\+F,\sigma,\tau)$ in \Cref{Alg:couple} with the initial states modified as follows:
    \begin{itemize}
        \item The initial set of constraints $\+E=\+F=\+C$;
        \item The initial partial assignments are both specified on $v$, with $\sigma=x_1,\tau=x_2$.
    \end{itemize}
\end{definition}
\end{comment}
The key proofs for the coupling procedure in \Cref{sec:coupling} apply to this setting.
In particular, the correctness follows directly from the proof of \Cref{lemma:coupling-correctness}. 
Thus, $\Couple(\+E,\+F,\sigma,\tau)$ returns $(X,Y)$ such that $X\sim \mu^{\sigma}_{\+C}$ and $Y\sim \mu^{\tau}_{\+C}$.
Since all pinned constraints in the initial discrepancy set $\+E^{\sigma}\triangle \+F^{\tau}$ must include $v$, they are connected in $G_{\Phi}$. Hence, the proof for \Cref{lemma:witness-tree-size} holds in this case as well.
As a result, following the same reasoning as in \Cref{theorem:correlation-decay}, %under \Cref{condition:main-condition}, 
we can establish that for the pair $(X,Y)$ returned by this $\Couple(\+E,\+F,\sigma,\tau)$, we have
\begin{equation}\label{eq:modified-correlation-decay}
    \Pr{d_{\-{Ham}}(X,Y)\geq k\cdot M}\leq d\cdot 2^{-M}.
\end{equation}
The extra factor of $d$ in the probability bound comes from the need to bound the number of witness trees rooted at a constraint containing $v$, rather than at a single constraint $c_0$. 
While \Cref{condition:main-condition} technically holds with width $k-1$ after pinning one variable,  
the proof of \Cref{theorem:correlation-decay} is robust enough to proceed with this slightly weakened \Cref{condition:main-condition}, thereby implying \eqref{eq:modified-correlation-decay}.
%    This is because \Cref{condition:main-condition} leaves room for the offset by one in the constraint arity $k$ for the proof of \Cref{theorem:correlation-decay}, 

%Now we can directly prove Theorems \ref{theorem:replica-symmetry} and \ref{theorem:non-reconstruction}.
Finally, we conclude the proof:
%Here, we consider \Cref{Alg:couple} with an initial state defined by two sets of pinned constraints obtained from $\+C$,
%differing only in their assignments on $v$ (or $v_1$ for the proof of \Cref{theorem:replica-symmetry}). 

%Note that,  the condition in \Cref{theorem:k-SAT-main} after pinning one variable still satisfies \Cref{theorem:CSP-main}. 
%Additionally, since all constraints in the starting discrepancy set include $v$ and are connected in $G_{\Phi}$, the proof for \Cref{lemma:witness-tree-size} also holds under this setting. 
%Thus, we can similarly prove \Cref{theorem:correlation-decay} under this setting, 
%but with an extra factor of $d$ multiplied on the right-hand-side.
%For this, we now need to bound the number of witness trees rooted at some constraint containing $v$ (or $v_1$), instead of being rooted at a single constraint $c_0$. 
%Therefore, we conclude that:
\begin{itemize}
    \item \Cref{theorem:replica-symmetry} follows from the exponential decay of correlation established in \eqref{eq:modified-correlation-decay};
    \item \Cref{theorem:non-reconstruction} also follows from \eqref{eq:modified-correlation-decay}, as with high probability, the discrepancy at $v$ does not affect the assignment on $\bar{B}_H(v,r)$ when $r$ is large enough.
    \qedhere
\end{itemize}
\end{proof}

%\subsubsection{Looseness of random CSP formulas}

%\subsubsection{Connectivity of the solution space}

Recent work has focused on exploring the solution space of random $k$-SAT through various notions of connectivity. 
%In the diagram shown in \Cref{figure:transition}, solutions within a cluster are connected by flipping a single bit. 
One way to characterize connectivity of the solution space is looseness considered in ~\cite{achlioptas2008algorithmic,chen2024fast}, basically saying one can obtain another solution with some variable flipped without overall flipping too many variables.

\begin{definition}[looseness]\label{definition:looseness}
Let $\Phi=(V,C)$ be a SAT formula with $|V|=n$. 
A variable $v\in V$ is said to be $M$-\emph{loose} with respect to a satisfying assignment $\sigma\in \Omega_{\Phi}$ if
there exists another satisfying assignment $\tau\in \Omega_{\Phi}$ such that $\tau(v)\neq \sigma(v)$ and $d_{\mathrm{Ham}}(\sigma,\tau)\leq M$.

For a random $k$-SAT formula $\Phi\sim \Phi(k,n,\lfloor\alpha n \rfloor)$, we say that $\Phi$ is $M$-\emph{loose} if, 
with high probability over the pair $(\Phi,\sigma)$, where $\sigma\sim \mu_{\Phi}$, 
all variables $v\in V$ are $M$-loose with respect to $\sigma$.
\end{definition}

%Looseness fundamentally differs from connectivity, as it is a local property, whereas connectivity is a global one. Neither property implies the other. For instance, the set of elements in $\{0,1\}^n$ with Hamming weight at least $2n/3$ or at least $n/3$ is $1$-loose but $\Omega(n)$-connected.

Looseness is conjectured to hold for random $k$-SAT up to the clustering threshold $\alpha_{\mathrm{clust}} \approx 2^k (\ln k) / k$ \cite[Conjecture 1]{achlioptas2008algorithmic}.
Here, we prove looseness at the considered densities.
\begin{theorem}[looseness of random $k$-SAT]\label{theorem:looseness}
Under the condition of \Cref{theorem:k-SAT-main}, the random $k$-SAT formula $\Phi(k,n,\lfloor\alpha n \rfloor)$ is $(\poly(k)\log n)$-loose.
\end{theorem}

%We then analyze the looseness of random CSP formulas. 

\begin{proof}
%Recall the coupling procedure (\Cref{Alg:couple}) and the random process with explicitly identified randomness (\Cref{definition:coupling-equivalent-process}). 
We will prove \Cref{theorem:looseness}, 
by modifying the random process introduced in \Cref{definition:coupling-equivalent-process} to find a local neighbor of a solution. The modified process proceeds as follows.

%\begin{definition}[algorithm for finding local neighbor of a satisfying solution]\label{definition:random-process-looseness}
The input consists of an atomic CSP formula $\Phi=(V,[q],\+C)$, a solution $\sigma^{\-{in}}\in \Omega_{\Phi}$, and a variable $v\in V$.  The process  produces an output assignment $\sigma^{\-{out}}\in \Omega_{\Phi}$ as described below:
\begin{enumerate}
    \item Choose an arbitrary value $x\in [q]\setminus\left\{\sigma^{\-{in}}(v)\right\}$.
    \item Run the random process in \Cref{definition:coupling-equivalent-process} with the following alterations:
\begin{itemize}
    \item The initial sets of constraints are $\+E_0=\+F_0=\+C$, and the initial partial assignments are both specified only at $v$, with $\sigma_0(v)=\sigma^{\-{in}}(v)$ and $\tau_0(v)=x$;
    \item Set $\*X$ as $\sigma^{\-{in}}$ and $\*Y\sim \mu_{\+C}^{\tau_0}$. That is, $\*Y\in \Omega^{\+C}$ is chosen uniformly at random conditioned on $\*Y(v)=x$.
    \end{itemize}
     Let $(\bm{\+E},\bm{\+F},\bm{\sigma},\bm{\tau},\bm{T},\bm{\varsigma})$ be the outcome of the random process. \label{item:random-process-looseness-2}
     \item If the outcome $(\bm{\+E},\bm{\+F},\bm{\sigma},\bm{\tau},\bm{T},\bm{\varsigma})$ satisfy that $|\bm{T}|=M$, 
          then let $\sigma^{\-{out}}\gets \sigma^{\-{in}}$. Otherwise, update  $\sigma^{\text{in}}$  by changing the values of the  variables assigned in $\bm{\tau}$ to $\bm{\tau}$, i.e.~set $\sigma^{\-{out}}\gets \bm{\tau}\land \sigma^{\-{in}}_{V\setminus\Lambda(\bm{\tau})}$.  \label{item:random-process-looseness-3}
\end{enumerate}
%\end{definition}

We can make the following observations about this process:
\begin{itemize}
\item Similar to \Cref{Alg:couple} or \Cref{definition:coupling-equivalent-process}, 
this random process  is an idealized algorithm, as it requires sampling from the non-trivial distribution $\mu_{\+C}^{\tau_0}$. The process is constructed to prove the looseness of random CSP formulas (\Cref{theorem:looseness}).
    \item In \Cref{item:random-process-looseness-2}, 
    the sets of pinned constraints $\+E^{\sigma}$ and $\+F^{\tau}$ may have discrepancies at the constraints involving $v$. 
    Since these constraints are connected in $G_{\Phi}$, the proof for \Cref{lemma:witness-tree-size} hold in this setting as well.
    \item In \Cref{item:random-process-looseness-3}, if the outcome $(\bm{\+E},\bm{\+F},\bm{\sigma},\bm{\tau},\bm{T},\bm{\varsigma})$ satisfies $|\bm{T}|<M$, then by \Cref{definition:coupling-equivalent-process}, we must have $\bm{\+E}^{\bm{\sigma}}=\bm{\+F}^{\bm{\tau}}$. Since $\*X=\sigma^{\-{in}}$, the output assignment $\sigma^{\-{out}}$ is a satisfying assignment with $\sigma^{\-{out}}(v)\neq \sigma^{\-{in}}(v)$ and $d_{\-{Ham}}(\sigma^{\-{in}},\sigma^{\-{out}})\leq kM$. 
\end{itemize}

Finally, \Cref{theorem:looseness}  follows from applying the same argument used in the proofs of  \Cref{theorem:replica-symmetry,theorem:non-reconstruction}, by invoking \Cref{theorem:correlation-decay} to the setting where the initial instances differ in one variable rather than in one constraint.
\end{proof}

\section{Linear program for counting and sampling}\label{sec:lp}
In this section, we introduce a linear programming approach that translates the coupling result from the previous section into algorithms for counting and sampling, and prove \Cref{theorem:CSP-main}. 
Similar to the coupling, this linear program is adapted from the one in \cite{WY24}, with key modifications to accommodate the criticality of random instances.

%We then introduce our linear programming approach to turn our correlation decay result algorithmic, and prove \Cref{theorem:CSP-main}. As our coupling is from \cite{WY24}, the framework of the linear program will be much similar to theirs, but the crucial constraint that captures the essence of the analysis would be different.

\subsection{Marginal probabilities}

We introduce the marginal probabilities associated with the coupling procedure, which correspond to the variables of the linear program.

Consider an atomic CSP formula $\Phi=(V,[q],\+C)$, and let $c_0\in \+C$ be an arbitrary constraint. Fix an integer $M\ge 1$. 
Recall the random process $P^{\cp}=P^{\cp}_M=\{(\+E_t,\+F_t,\sigma_t,\tau_t,T_t,\varsigma_t)\}_{t\ge 0}$ and its outcome distribution $\mu^{\cp}=\mu^{\cp}_M$, along with their supports $\+V^{\cp}$ and $\+L^{\cp}$, as constructed in \Cref{definition:coupling-equivalent-process}.
%
%The following defines a family of probabilities arising from a natural one-sided sampling process induced from $P^{\cp}$.
%
%These probabilities correspond to the variables of the linear program that will be introduced later.

We define a family of marginal probabilities induced by the random process $P^{\cp}$.

\begin{definition}[marginal probabilities]\label{definition:imaginary-sampler-quantities}
Let ${X}$ and ${Y}$ be generated as follows:
\begin{itemize}
\item draw $(\bm{\+E},\bm{\+F},\bm{\sigma},\bm{\tau},\bm{T},\bm{\varsigma})\sim\mu^{\cp}$;
\item draw ${X}\sim \mu^{\bm{\sigma}}_{\bm{\+E}}$, and similarly ${Y}\sim \mu^{\bm{\tau}}_{\bm{\+F}}$.
\end{itemize}
For each  $(\+E,\+F,\sigma,\tau,T,\varsigma)\in\+V^{\cp}$, 
we define the following pair of marginal probabilities:
      \[p_{(\+E,\+F,\sigma,\tau,T,\varsigma)}^X\defeq\Pr{(\+E,\+F,\sigma,\tau,T,\varsigma)\in P^{\cp}\mid {X}=\bm{x}},
      \]
      \[p_{(\+E,\+F,\sigma,\tau,T,\varsigma)}^Y\defeq\Pr{(\+E,\+F,\sigma,\tau,T,\varsigma)\in P^{\cp}\mid {Y}=\bm{y}}.
    \]
where $\bm{x},\bm{y}\in[q]^V$ are arbitrary assignments satisfying $\+E\land\sigma$ and $\+F\land\tau$, respectively.
\end{definition}

Note that each marginal probability $p_{(\+E,\+F,\sigma,\tau,T,\varsigma)}^X$ is defined by conditioning on $X$ being an arbitrary assignment  $\bm{x}\in[q]^V$  that satisfies  $\+E\land\sigma$ (and similarly for each $p_{(\+E,\+F,\sigma,\tau,T,\varsigma)}^Y$).
The well-definedness of these probabilities is ensured by the following proposition.
%shows the above quantities are well-defined and gives explicit expressions for these quantities, whose proof is almost identical as \cite[Proposition 4.2]{WY24}. We include here for completeness.

\begin{proposition}\label{lemma:quantities-well-defined}
Assume $\Phi$ is satisfiable. Fix any $(\+E,\+F,\sigma,\tau,T,\varsigma)\in\+V^{\cp}$. 
%The sets $\Omega^{\+E\land\sigma}\defeq\{\pi\in [q]^V\mid \pi\text{ satisfies }\+E\land\sigma\}$ and $\Omega^{\+F\land\tau}\defeq\{\pi\in [q]^V\mid \pi\text{ satisfies }\+F\land\tau\}$ are nonempty.
The following sets are nonempty:
\begin{align*}
\Omega^{\+E\land\sigma}
&\defeq\{\pi\in [q]^V\mid \pi\text{ satisfies }\+E\land\sigma\},\\
\Omega^{\+F\land\tau}
&\defeq\{\pi\in [q]^V\mid \pi\text{ satisfies }\+F\land\tau\}.
\end{align*}
Furthermore, for any $\bm{x},\bm{x}'\in \Omega^{\+E\land\sigma}$ and $\bm{y},\bm{y}'\in \Omega^{\+F\land\tau}$, it holds that
\begin{align*}
%\forall \bm{x},\bm{x}'\in \Omega^{\+E\land\sigma}:&&  
\Pr{(\+E,\+F,\sigma,\tau,T,\varsigma)\in P^\cp\mid {X}=\bm{x}}&=\Pr{(\+E,\+F,\sigma,\tau,T,\varsigma)\in P^\cp\mid {X}=\bm{x}'},\\
%\end{align*}
%and for any $\bm{y},\bm{y}'\in \Omega^{\+F\land\tau}$, it holds that
%\begin{align*}
%\forall \bm{y},\bm{y}'\in \Omega^{\+F\land\tau}:&&  
\Pr{(\+E,\+F,\sigma,\tau,T,\varsigma)\in P^\cp\mid {Y}=\bm{y}}&=\Pr{(\+E,\+F,\sigma,\tau,T,\varsigma)\in P^\cp\mid {Y}=\bm{y}'}.
\end{align*}
This means that $p_{(\+E,\+F,\sigma,\tau,T,\varsigma)}^{X}$ and $p_{(\+E,\+F,\sigma,\tau,T,\varsigma)}^{Y}$ are well-defined. Moreover, it holds that
\begin{equation}\label{eq:real-quantities}
p_{(\+E,\+F,\sigma,\tau,T,\varsigma)}^{X}=\mu_{\+C}(\+F\land \tau) \quad\text{ and } \quad p_{(\+E,\+F,\sigma,\tau,T,\varsigma)}^{Y}=\mu_{\+C\setminus \{c_0\}}(\+E\land \sigma).
 \end{equation}
\end{proposition}

\Cref{lemma:quantities-well-defined} can be proved by following the same argument as in  \cite[Proposition 4.2]{WY24}.
In fact, the same proof applies as long as the coupling is well-defined.
Therefore, we omit the proof here.

The following proposition outlines families of linear constraints satisfied by the marginal probabilities.

\begin{proposition}\label{lemma:quantities-properties}
    Assume $\Phi$ is satisfiable. Fix any $(\+E,\+F,\sigma,\tau,T,\varsigma)\in\+V^{\cp}$.
    The following properties hold for $p_{(\+E,\+F,\sigma,\tau,T,\varsigma)}^X$ and $p_{(\+E,\+F,\sigma,\tau,T,\varsigma)}^Y$:
\begin{enumerate}
    \item It always holds that $p_{(\+E,\+F,\sigma,\tau,T,\varsigma)}^X, p_{(\+E,\+F,\sigma,\tau,T,\varsigma)}^Y\in [0,1]$. In particular, $$p_{(\+C\setminus \{c_0\},\+C,\varnothing,\varnothing,\emptyset,\varnothing)}^X=p_{(\+C\setminus \{c_0\},\+C,\varnothing,\varnothing,\emptyset,\varnothing)}^Y=1.$$
    % \begin{align*}&p_{(\+E,\+F,\sigma,\tau,T,\varsigma)}^X, p_{(\+E,\+F,\sigma,\tau,T,\varsigma)}^Y\in [0,1];\\
    % &\text{In particular, }p_{(\+C\setminus \{c_0\},\+C,\varnothing,\varnothing,\emptyset,\varnothing)}^X=p_{(\+C\setminus \{c_0\},\+C,\varnothing,\varnothing,\emptyset,\varnothing)}^Y=1.
    % \end{align*} 
    \label{item:quantities-properties-1}
    
    \item If $(\+E,\+F,\sigma,\tau,T,\varsigma) \not\in \+L^{\cp}$, then the following holds:
    %where $\+V^{\cp}$ and $\+L^\cp$ are defined in \Cref{definition:coupling-equivalent-process},  
    \label{item:quantities-properties-2}
    \begin{enumerate}
        \item If $\+F^{\tau}\not\subseteq \+E^{\sigma}$, let $c$ be the smallest constraint in $\+F^{\tau}\setminus \+E^{\sigma}$,  
        and define $\pi\defeq \vio{c}$. Then,
        \begin{align*}
            p_{(\+E,\+F,\sigma,\tau,T,\varsigma)}^X 
            =
            &p^X_{(\+E\cup \{c\},\+F,\sigma,\tau,T,\varsigma)}
            =\sum\limits_{\substack{\rho\in [q]^{\vbl(c)}\\ (\+E,\+F,\sigma\land \pi,\tau\land \rho,T\oplus \unpin{c},\varsigma\land \rho)\in \+V^{\cp}}}p^X_{(\+E,\+F,\sigma\land \pi,\tau\land \rho,T\oplus \unpin{c},\varsigma\land \rho)};\\
             p_{(\+E,\+F,\sigma,\tau,T,\varsigma)}^Y 
             =
             &p^Y_{(\+E\cup \{c\},\+F,\sigma,\tau,T,\varsigma)}+p^Y_{(\+E,\+F,\sigma\land \pi,\tau\land \rho,T\oplus \unpin{c},\varsigma\land \rho)},\\
              & \text{for any }\rho\in [q]^{\vbl(c)} \text{ such that }(\+E,\+F,\sigma\land \pi,\tau\land \rho,T\oplus \unpin{c},\varsigma\land \rho)\in \+V^{\cp}.
        \end{align*}
       \label{item:quantities-properties-3-a}
        \item If $\+F^{\tau}\subseteq \+E^{\sigma}$, let $c$ be the smallest constraint in $\+E^{\sigma}\setminus \+F^{\tau}$,
        and define $\rho\defeq\vio{c}$. Then,
        \begin{align*}
        p_{(\+E,\+F,\sigma,\tau,T,\varsigma)}^X
=
&p^X_{(\+E,\+F\cup \{c\},\sigma,\tau,T,\varsigma)}+p^X_{(\+E,\+F,\sigma\land \pi,\tau\land \rho,T\oplus \unpin{c},\varsigma\land \pi)},\\
&\text{for any }\pi\in [q]^{\vbl(c)} \text{ such that }(\+E,\+F,\sigma\land \pi,\tau\land \rho,T\oplus \unpin{c},\varsigma\land \pi)\in \+V^{\cp};\\
p_{(\+E,\+F,\sigma,\tau,T,\varsigma)}^Y
            =
            &p^Y_{(\+E,\+F\cup \{c\},\sigma,\tau,T,\varsigma)}=\sum\limits_{\substack{\pi\in [q]^{\vbl(c)}\\(\+E,\+F,\sigma\land \pi,\tau\land \rho,T\oplus \unpin{c},\varsigma\land \pi)\in \+V^{\cp}}}p^Y_{(\+E,\+F,\sigma\land \pi,\tau\land \rho,T\oplus \unpin{c},\varsigma\land \pi)}.
        \end{align*}\label{item:quantities-properties-3-b}
    \end{enumerate}
        \item Furthermore, it always holds that
        %For any $(\+E,\+F,\sigma,\tau,T,\varsigma)\in \+V^{\cp}$,
        \[
        p_{(\+E,\+F,\sigma,\tau,T,\varsigma)}^X\cdot \frac{ |\Omega^{\+E\land \sigma}|}{|\Omega^{\+C\setminus \{c_0\}}|}= p_{(\+E,\+F,\sigma,\tau,T,\varsigma)}^Y\cdot \frac{|\Omega^{\+F\land \tau}|}{|\Omega^{\+C}|},
        \] \label{item:quantities-properties-4}
        and 
        \begin{align*}
        p^X_{(\+E,\+F,\sigma,\tau,T,\varsigma)},   p^Y_{(\+E,\+F,\sigma,\tau,T,\varsigma)}\leq q^{-|\vbl(T)\cap \vgood|}\cdot \left( 1-\mathrm{e}q^{-(1-\eps_1)k} \right)^{-|\vbl(T)|p_1\alpha}. 
        \end{align*}
\end{enumerate}
\end{proposition}
\begin{proof}
\Cref{item:quantities-properties-1,item:quantities-properties-2} are  derived directly from \Cref{definition:imaginary-sampler-quantities} and the well-definedness ensured in  \Cref{lemma:quantities-well-defined}, both of which are straightforward to verify.

The equation in \Cref{item:quantities-properties-4} follows from \eqref{eq:real-quantities}.
From the same equation \eqref{eq:real-quantities}, we can further derive that $p_{(\+E,\+F,\sigma,\tau,T,\varsigma)}^{X}=\mu_{\+C}(\+F\land \tau)\le \mu_{\+C}(\tau)$ and $p_{(\+E,\+F,\sigma,\tau,T,\varsigma)}^{Y}=\mu_{\+C\setminus \{c_0\}}(\+E\land \sigma)\le \mu_{\+C\setminus \{c_0\}}(\sigma)$.
Together with \Cref{lemma:assignment-probability-bound},  these imply the final inequalities in \Cref{item:quantities-properties-4}.
\end{proof}

%The following bound for the marginal probabilities $p_{(\+E,\+F,\sigma,\tau,T,\varsigma)}^X$ and $p_{(\+E,\+F,\sigma,\tau,T,\varsigma)}^Y$, based on the size of the witness tree, is derived from \Cref{lemma:assignment-probability-bound}.
%a similar argument as used in our analysis of the correlation decay property (\Cref{lemma:witness-tree-conditioned-probability-bound}).

% \begin{proposition}\label{lemma:quantities-property-bad-event}
%      Assume \Cref{condition:main-condition}. The following holds for any $(\+E,\+F,\sigma,\tau,T,\varsigma)\in \+L^{\trun}$: 
% \begin{align*}
%       p^X_{(\+E,\+F,\sigma,\tau,T,\varsigma)},   p^Y_{(\+E,\+F,\sigma,\tau,T,\varsigma)}\leq q^{-|\vbl(T)\cap \vgood|}\cdot \left( 1-\mathrm{e}q^{-(1-\eps_1)k} \right)^{-|\vbl(T)|p_1d}. %\qquad \forall(\+E,\+F,\sigma,\tau,T,\varsigma)\in \+L^{\trun}.
% \end{align*}
% \end{proposition}

%From \Cref{lemma:quantities-well-defined}, we know that $p_{(\+E,\+F,\sigma,\tau,T,\varsigma)}^{X}=\mu_{\+C}(\+F\land \tau)\le \mu_{\+C}(\tau)$ and $p_{(\+E,\+F,\sigma,\tau,T,\varsigma)}^{Y}=\mu_{\+C\setminus \{c_0\}}(\+E\land \sigma)\le \mu_{\+C\setminus \{c_0\}}(\sigma)$. Then the following proposition follows from \Cref{lemma:assignment-probability-bound}.

\subsection{Setting up the linear program}
Next, we construct a linear program that that mimics the marginal probabilities $p_{(\+E,\+F,\sigma,\tau,T,\varsigma)}^X$ and $p_{(\+E,\+F,\sigma,\tau,T,\varsigma)}^Y$.
%in order to bootstrap the ratio between the numbers of assignments satisfying $\+C$ and those satisfying $\+C\setminus \{c_0\}$. %Before that, some technicalities need to be introduced. 

\subsubsection{The coupling tree}
To set up the linear program, 
we first construct a recursion tree for the coupling procedure $\Couple(\+C\setminus \{c_0\},\+C,\varnothing,\varnothing)$, 
truncating it when the size of the witness tree $T$ exceeds $M$.
%We define the following notion of the recursion tree for the coupling procedure $\Couple(\+C\setminus \{c_0\},\+C,\varnothing,\varnothing)$, truncated whenever the witness tree $T$ has size $|V(T)|\geq K$. 
%The defintion is meant to capture the structure of the random process in \Cref{definition:coupling-equivalent-process}. 
%The construction is quite similar to the one presented in ~\cite{WY24}, with the only differences being the information recorded through the coupling (here we record the witness tree and the cumulative non-violating assignment) and the truncation condition.

%We now dive a little bit deeper into the analysis of the correlation decay property in \Cref{sec:correlation-decay}. We still fix an atomic CSP $\Phi=(V,[q],\+C)$ and a constraint $c_0\in \+C$. We define the following ($K$-truncated) coupling tree structure, which captures all possible executions of $\Couple(\+C\setminus \{c_0\},\+C,\varnothing,\varnothing)$, truncated whenever $|B|\geq K$.

\begin{definition}[$M$-truncated coupling tree]\label{definition:coupling-tree}
The \emph{$M$-truncated coupling tree}, denoted $\+T=\+T_M(\Phi,c_0)$, is a finite rooted tree, 
where each node in $\+T$ corresponds to a tuple $(\+E,\+F,\sigma,\tau,T,\varsigma)\in\+V^{\cp}$.
The tree $\+T$ is constructed inductively as follows: \begin{enumerate}
    \item The root of $\+T$ corresponds to $(\+C\setminus \{c_0\},\+C,\varnothing,\varnothing,\emptyset,\varnothing)$, and has depth $0$.
    \label{item:coupling-tree-1}
    \item 
    For $i=0,1,\ldots$, consider all nodes $(\+E,\+F,\sigma,\tau,T,\varsigma)\in V(\+T)$ of depth $i$ in the current tree $\+T$.\label{item:coupling-tree-2}
    \begin{enumerate}
        \item If $\sigma$ violates $\+E$ or $\tau$ violates $\+F$ or $\+E^{\sigma}=\+F^{\tau}$ or $|V(T)|=M$, then $(\+E,\+F,\sigma,\tau,T,\varsigma)$ is left as a leaf node in $\+T$. \label{item:coupling-tree-2-a}
        \item Otherwise, if $\+F^{\tau}\not\subseteq \+E^{\sigma}$, then pick the smallest $c\in \+F^{\tau}\setminus \+E^{\sigma}$ and add $(\+E\cup \{c\},\+F,\sigma,\tau,T,\varsigma)$ as a child of $(\+E,\+F,\sigma,\tau,T,\varsigma)$ in $\+T$.
        Furthermore, for each $\rho\in [q]^{\vbl(c)}$ and $\pi=\vio{c}$, add $(\+E,\+F,\sigma\land \pi,\tau\land \rho,T\oplus \unpin{c},\varsigma\land \rho)$ as a child of $(\+E,\+F,\sigma,\tau,T,\varsigma)$ in $\+T$.\label{item:coupling-tree-2-b}
        \item Otherwise, it holds that $\+F^{\tau}\subseteq \+E^{\sigma}$. Then, pick the smallest $c\in \+E^{\sigma}\setminus \+F^{\tau}$ and add $(\+E,\+F\cup \{c\},\sigma,\tau,T,\varsigma)$ as a child of $(\+E,\+F,\sigma,\tau,T,\varsigma)$ in $\+T$.
        Furthermore, for each $\pi\in [q]^{\vbl(c)}$ and $\rho=\vio{c}$, add $(\+E,\+F,\sigma\land \pi,\tau\land \rho,T\oplus \unpin{c},\varsigma\land \pi)$ as a child of $(\+E,\+F,\sigma,\tau,T,\varsigma)$ in $\+T$.\label{item:coupling-tree-2-c}
    \end{enumerate}
\end{enumerate}
Let $\+L$ denote the set of leaf nodes in $\+T$. 
We further define the following sets of leaf nodes:
\begin{itemize}
    \item $\+L_{\coup}\defeq \{(\+E,\+F,\sigma,\tau,T,\varsigma)\in \+L\mid \+E^{\sigma}=\+F^{\tau}\}$ as the set of ``coupled'' leaf nodes in $\+T$;
    \item $\+L_{\trun}\defeq \{(\+E,\+F,\sigma,\tau,T,\varsigma)\in \+L\mid |V(T)|=M\}$ as the set of ``truncated'' leaf nodes in $\+T$;
    \item $\+{L}_{\text{valid}}\defeq\+L_{\trun}\cup \+L_{\coup}$ as the set of ``valid'' leaf nodes in $\+T$;
    \item $\+L_{\text{invld}}\defeq \+L\setminus \+{L}_{\text{valid}}$ as the set of ``invalid'' leaf nodes in $\+T$.
\end{itemize}
\end{definition}

\begin{proposition}\label{prop:coupling-tree-size}
For any satisfiable $\Phi=(V,[q],\+C)$, any $c_0\in \+C$ and $M\ge 1$, the $M$-truncated coupling tree $\+T=\+T_M(\Phi,c_0)$ has a depth of at most $M\Delta k+1$ and a branching number of at most $q^{2k}$, where $\Delta=\Delta(\Phi)$ is the maximum degree of $G_{\Phi}$.
\end{proposition}
\begin{proof}
By contradiction, assume there exists a node $(\+E', \+F',\sigma',\tau',T',\varsigma')\in V(\+T)$ with depth $M\Delta k+2$. 
We track the size of $\+E^{\sigma}\triangle\+F^{\tau}$ along the path from the root to $(\+E', \+F',\sigma',\tau',T',\varsigma')$, denoting this size by $t$. 
Initially, by \Cref{definition:linear-program}, we have $t=1$. 
At each intermediate node $(\+E, \+F,\sigma,\tau,T,\varsigma)\in V(\+T)$ along the path: we either add some constraint $c$ into $\+E$ or $\+F$, reducing $t$ by 1,
or we assign values to $\sigma$ and $\tau$ on $\vbl(c)$, which increases $t$ by at most $k\Delta-1$ 
(since at most $k\Delta$ new elements can be added into $\+E^{\sigma}\triangle\+F^{\tau}$, and $c$ is removed from $\+E^{\sigma}\triangle\+F^{\tau})$. 
In the latter case, the size of $T$ grows by one according to \Cref{lemma:witness-tree-size}. 
Let $i$ denote the number of times this latter operation is executed.

Since the depth of $(\+E', \+F',\sigma',\tau',T',\varsigma')$ is $M\Delta k+2$, the above step is repeated for $M\Delta k+2$ times. 
Finally, at the node $(\+E', \+F',\sigma',\tau',T',\varsigma')$, we still have $t=|\+E'^{\sigma'}\triangle\+F'^{\tau'}|\geq 0$. 
Therefore,
\[
(k\Delta -1)\cdot i+1-( M\Delta k+2-i) \geq 0,
\]
This implies that $|T'|=i>M$, which contradicts the truncation condition in \Cref{item:coupling-tree-2-a} of \Cref{definition:coupling-tree}.

Finally, it is easy to observe that each node in $\+T$ has at most $q^{2k}$ children.
\end{proof}

% \begin{proof}
% We only present the proof for $p^X_{(\+E,\+F,\sigma,\tau,T,\varsigma)}$, as the argument for $p^Y_{(\+E,\+F,\sigma,\tau,T,\varsigma)}$  follows analogously. 
% Recalling from \Cref{lemma:quantities-well-defined} that $p^X_{(\+E,\+F,\sigma,\tau,T,\varsigma)}=\mu_{\+C}(\+F\land \tau)$,  observe:
% \[
% \mu_{\+C}(\+F\land \tau)\leq \mu_{\+C}(\tau)\leq q^{-|\vbl(T)\cap \vgood|}\cdot \left( 1-\mathrm{e}q^{-(1-\eps_1)k} \right)^{-|\vbl(T)|p_1d},
% \]
% where the last inequality follows from \Cref{lemma:assignment-probability-bound}. This  proves the lemma.
% %The only nontrivial step is the last inequality, which follows from the observation that, after pinning $\tau_{\vbl(T)\cap \vbad}$, 
% %the resulting pinned formula has each variable with a degree of at most $p_1 d$ (by \Cref{fact:good-vertex-degree-bound}),
% %and  each constraint containing at least $(1-\epsilon_1)k$ variables (by \Cref{fact:good-vertex-fraction}).
% %Applying \Cref{HSS}, this inequality is established by setting  $x(c)=\mathrm{e}2^{-(1-\epsilon_1)k}$ for each constraint $c$ in the pinned formula.
% \end{proof}

\subsubsection{The linear program}
We now present the linear program,
constructed on the $M$-truncated coupling tree from \Cref{definition:coupling-tree}.
Each node of the tree, denoted as $(\+E,\+F,\sigma,\tau,T,\varsigma)$, is associated with two variables mimicking $p_{(\+E,\+F,\sigma,\tau,T,\varsigma)}^X$ and $p_{(\+E,\+F,\sigma,\tau,T,\varsigma)}^Y$. 
The linear constraints of this LP are derived from the properties listed in \Cref{lemma:quantities-properties}. 
%The linear program listed below share much similarity with the linear program in ~\cite{WY24}, only with the last truncation constraint capturing the analysis of the coupling being very much different.

\begin{definition}[linear program induced by the coupling] \label{definition:linear-program}
Let $\+T=\+T_M(\Phi,c_0)$ denote the $M$-truncated coupling tree, constructed according to \Cref{definition:coupling-tree}. 
Given two parameters $0\leq r_- \leq r_+$,
we define a linear program using variables $\hat{p}_{(\+E,\+F,\sigma,\tau,T,\varsigma)}^X$ and $\hat{p}_{(\+E,\+F,\sigma,\tau,T,\varsigma)}^Y$ for all $(\+E,\+F,\sigma,\tau,T,\varsigma)\in V(\+T)$:
% We now set up a linear program to mimic the quantities $ p_{(\+E,\+F,\sigma,\tau,T,\varsigma)}^X, p_{(\+E,\+F,\sigma,\tau,T,\varsigma)}^Y$ defined in \Cref{definition:imaginary-sampler-quantities}, but now on the coupling tree $\+T$.
% Formally,
% given parameters $r_- \leq r_+$ and $\+T$, we check whether the following linear
% program in variables $\hat{p}_{(\+E,\+F,\sigma,\tau,T,\varsigma)}^X$ and $\hat{p}_{(\+E,\+F,\sigma,\tau,T,\varsigma)}^Y$ is feasible:
\begin{enumerate}[label=\Roman*.]
    \item \emph{Range constraints}: 
\begin{align*}
&\hat{p}_{(\+C\setminus \{c_0\},\+C,\varnothing,\varnothing,\emptyset)}^X
=\hat{p}_{(\+C\setminus \{c_0\},\+C,\varnothing,\varnothing,\emptyset)}^Y=1;\\
&\hat{p}_{(\+E,\+F,\sigma,\tau,T,\varsigma)}^X,\hat{p}_{(\+E,\+F,\sigma,\tau,T,\varsigma)}^Y\in [0,1], \qquad \forall(\+E,\+F,\sigma,\tau,T,\varsigma)\in V(\+T).
\end{align*}
\label{item:linear-program-range-constraints}
\item \emph{Non-leaf constraints}: 
For each non-leaf node $(\+E,\+F,\sigma,\tau,T,\varsigma) \in V(\+T)\setminus\+L$: \label{item:linear-program-non-leaf-constraints}
    \begin{enumerate}
        \item If $\+F^{\tau}\not\subseteq \+E^{\sigma}$, let $c$ be the smallest constraint in $\+F^{\tau}\setminus \+E^{\sigma}$ and $\pi=\vio{c}$:
\begin{align*}
\hat{p}_{(\+E,\+F,\sigma,\tau,T,\varsigma)}^X
&=\hat{p}^X_{(\+E\cup \{c\},\+F,\sigma,\tau,T,\varsigma)}=\sum\limits_{\rho\in [q]^{\vbl(c)}}\hat{p}^X_{(\+E,\+F,\sigma\land \pi,\tau\land \rho,T\oplus \unpin{c},\varsigma\land \rho)};\\  
\hat{p}_{(\+E,\+F,\sigma,\tau,T,\varsigma)}^Y
&=\hat{p}^Y_{(\+E\cup \{c\},\+F,\sigma,\tau,T,\varsigma)}+\hat{p}^Y_{(\+E,\+F,\sigma\land \pi,\tau\land \rho,T\oplus \unpin{c},\varsigma\land \rho)}, \qquad \forall\rho\in [q]^{\vbl(c)}.
\end{align*}
\label{item:linear-program-2-a}
        \item Otherwise, if $\+F^{\tau}\subseteq \+E^{\sigma}$, let $c$ be the smallest constraint in $\+E^{\sigma}\setminus \+F^{\tau}$ and $\rho=\vio{c}$:
\begin{align*}
            \hat{p}_{(\+E,\+F,\sigma,\tau,T,\varsigma)}^X
            &=\hat{p}^X_{(\+E,\+F\cup \{c\},\sigma,\tau,T,\varsigma)}+\hat{p}^X_{(\+E,\+F,\sigma\land \pi,\tau\land \rho,T\oplus \unpin{c},\varsigma\land \pi)}, \qquad \forall\pi\in [q]^{\vbl(c)};\\
            \hat{p}_{(\+E,\+F,\sigma,\tau,T,\varsigma)}^Y
            &=\hat{p}^Y_{(\+E,\+F\cup \{c\},\sigma,\tau,T,\varsigma)}=\sum\limits_{\pi\in [q]^{\vbl(c)}}\hat{p}^Y_{(\+E,\+F,\sigma\land \pi,\tau\land \rho,T\oplus \unpin{c},\varsigma\land \pi)}.
\end{align*}
        \label{item:linear-program-2-b}
    \end{enumerate}
\item \emph{Leaf constraints}: \label{item:linear-program-leaf-constraints}
For each leaf node $(\+E,\+F,\sigma,\tau,T,\varsigma)\in \+L$:
    \begin{enumerate}
    \item If it is a coupled leaf $(\+E,\+F,\sigma,\tau,T,\varsigma)\in\+L_{\coup}$, 
        \[
       r_-\cdot \hat{p}_{(\+E,\+F,\sigma,\tau,T,\varsigma)}^Y\leq \hat{p}_{(\+E,\+F,\sigma,\tau,T,\varsigma)}^X\leq r_+\cdot \hat{p}_{(\+E,\+F,\sigma,\tau,T,\varsigma)}^Y.
        \] \label{item:linear-program-coupled}
  \item If it is an invalid leaf $(\+E,\+F,\sigma,\tau,T,\varsigma)\in \+L_{\text{invld}}$,
  \begin{align*}
      \hat{p}_{(\+E,\+F,\sigma,\tau,T,\varsigma)}^Y=0, \text{ if $\sigma$ violates $\+E$}; \\
      \hat{p}_{(\+E,\+F,\sigma,\tau,T,\varsigma)}^X=0, \text{ if $\tau$ violates $\+F$}.
  \end{align*}
  \label{item:linear-program-invalid}
    \end{enumerate}
        \item \emph{Truncation constraints}: For each truncated leaf node $(\+E,\+F,\sigma,\tau,T,\varsigma)\in \+L^{\trun}$:
          \begin{align*}
         p^X_{(\+E,\+F,\sigma,\tau,T,\varsigma)},   p^Y_{(\+E,\+F,\sigma,\tau,T,\varsigma)}\leq q^{-|\vbl(T)\cap \vgood|}\cdot \left( 1-\mathrm{e}q^{-(1-\eps_1)k} \right)^{-|\vbl(T)|p_1\alpha}  .
\end{align*}
\label{item:linear-program-overflow-constraints}
\end{enumerate}
\end{definition}

\begin{remark}
This linear program closely resembles the one presented in~\cite{WY24},
with the primary distinction being the last class of linear constraints: the truncation constraints.
These constraints replace the ``overflow constraints'' used in the LP from~\cite{WY24}. 
Notably, in their design of the LP, this class of overflow constraints also distinguishes their approach from other LP-based algorithms, 
such as those in \cite{Moi19,guo2019counting,vishesh21towards,GGGY21}, 
and is considered a key step in approaching the critical threshold in their context.
In contrast, the use of truncation constraints here effectively captures the critical behavior of random instances 
and is essential for approaching the critical threshold in this new context.
\end{remark}

\subsection{Analysis of the linear program}
We now establish both the correctness and efficiency of the linear program constructed earlier. 

\subsubsection{Performance of the LP}
First, we show that the feasibility of the LP can be checked efficiently.

\begin{proposition}\label{lemma:lp-properties}
Assume \Cref{condition:main-condition}.
For any $0\leq r_{-}\leq r_{+}$, the feasibility of the linear program in \Cref{definition:linear-program} can be determined within $\exp(M\cdot \poly(k,\log q,\alpha))$ time. 
\end{proposition}

This result follows directly from \Cref{prop:coupling-tree-size}, as the size of the linear program is bounded by the size of the $M$-truncated coupling tree $\+T=\+T_M(\Phi,c_0)$, which is at most $\exp(M\cdot \poly(k,\log q,\alpha))$ according to the bound on maximum degree (\Cref{property:maximum-degree}) in \Cref{condition:main-condition}.

Next, we prove the soundness of this linear program 
by showing that the true values of the marginal probabilities satisfy all the linear constraints.

\begin{lemma}\label{lemma:lp-feasibility}
Assume $\Phi$ is satisfiable.
%Assume \Cref{condition:main-condition}. 
The linear program in \Cref{definition:linear-program} is feasible for 
$$r_{-}=r_{+}=\frac{|\Omega^{\+C\setminus \{c_0\}}|}{|\Omega^{\+C}|}.$$
\end{lemma}
\begin{proof}

For each $(\+E,\+F,\sigma,\tau,T,\varsigma)\in V(\+T)$, we define the following quantities:
\begin{equation}\label{eq:linear-program-solutions}
\hat{p}_{(\+E,\+F,\sigma,\tau,T,\varsigma)}^X=\mu_{\+C}(\+F\land \tau), \quad \hat{p}_{(\+E,\+F,\sigma,\tau,T,\varsigma)}^Y=\mu_{\+C\setminus \{c_0\}}(\+E\land \sigma).
\end{equation}
We will show that they form a feasible solution to the LP described in \Cref{definition:linear-program} with the parameters $r_{-}=r_{+}=\frac{|\Omega^{\+C\setminus \{c_0\}}|}{|\Omega^{\+C}|}$. 
By \Cref{lemma:quantities-well-defined}, these quantities in \eqref{eq:linear-program-solutions} are consistent with the actual values of $p_{(\+E,\+F,\sigma,\tau,T,\varsigma)}^{X},p_{(\+E,\+F,\sigma,\tau,T,\varsigma)}^{Y}$ as defined in \Cref{definition:imaginary-sampler-quantities}, also extending these marginal probabilities to all nodes in $V(\+T)$ rather than just $\+V^{\cp}$.
Note that compared to $\+V^{\cp}$, which contains only those $(\+E,\+F,\sigma,\tau,T,\varsigma)$ corresponding to the final outcomes of the $M$-truncated coupling procedure, the set $V(\+T)$ also contains all $(\+E,\+F,\sigma,\tau,T,\varsigma)$ corresponding to the intermediate steps of the  coupling.

We then verify that \eqref{eq:linear-program-solutions} satisfies all linear constraints of the LP in \Cref{definition:imaginary-sampler-quantities}:
\begin{itemize}
    \item \emph{Range and Non-leaf constraints}: These constraints hold by \Cref{item:quantities-properties-1,item:quantities-properties-2} of \Cref{lemma:quantities-properties}. 
\item \emph{Leaf constraints:} 
\Cref{item:linear-program-invalid} is a direct consequence of \eqref{eq:linear-program-solutions}. 
To verify \Cref{item:linear-program-coupled}, note that for $(\+E,\+F,\sigma,\tau,T,\varsigma)\in \+L_{\coup}$, we have $\+E^\sigma=\+F^{\tau}$, which implies that $|\Omega^{\+E\land \sigma}|=|\Omega^{\+F\land \tau}|$. Thus:
\begin{itemize}
    \item If $|\Omega^{\+E\land \sigma}|=|\Omega^{\+F\land \tau}|>0$, then \Cref{item:linear-program-coupled}  follows directly from \Cref{item:quantities-properties-4} of \Cref{lemma:quantities-properties}.
    \item If $|\Omega^{\+E\land \sigma}|=|\Omega^{\+F\land \tau}|=0$, then $\mu_{\+C}(\+F\land \tau)=\mu_{\+C\setminus \{c_0\}}(\+E\land \sigma)=0$ and \Cref{item:linear-program-coupled} holds.
\end{itemize}
\item \emph{Truncation constraints:} 
The constraints in \Cref{item:linear-program-overflow-constraints} hold by \Cref{item:quantities-properties-4} of \Cref{lemma:quantities-properties}. \qedhere
\end{itemize}
\end{proof}

At last, we show that the feasibility of the linear program implies that $r_{-}$ and $r_{+}$ provide respective 
lower and upper bounds for $\frac{|\Omega^{\+C\setminus \{c_0\}}|}{|\Omega^{\+C}|}$ with bounded multiplicative error. 
With this, we can apply a binary search to approximate the true value of $\frac{|\Omega^{\+C\setminus \{c_0\}}|}{|\Omega^{\+C}|}$.

\begin{lemma}\label{lemma:linear-program-error-bound}
Assume \Cref{condition:main-condition}. 
If the LP in \Cref{definition:linear-program} is feasible for parameters $0\leq r_-\leq r_+$, then it holds that
\[
\left(1-2\cdot 2^{-M}\right)r_-\leq \frac{|\Omega^{\+C\setminus \{c_0\}}|}{|\Omega^{\+C}|}\leq \left(1+2\cdot 2^{-M}\right)r_+.
\]
\end{lemma}

The proof of \Cref{lemma:linear-program-error-bound} relies on the following claim, which provides an upper bound on the estimated average marginal probabilities for truncated nodes of the coupling tree.
\begin{claim}\label{lemma:linear-program-bad-leaf-loss}
Assume \Cref{condition:main-condition}. The following hold for the solution of the LP in \Cref{definition:linear-program}: %with respect to the linear program in \Cref{definition:linear-program} with parameters $0\leq r_-\leq r_+$:
\begin{align*}
    \frac{1}{|\Omega^{\+C\setminus \{c_0\}}|}\sum\limits_{\bm{x}\in \Omega^{\+C\setminus \{c_0\}}}\sum\limits_{\substack{(\+E,\+F,\sigma,\tau,T,\varsigma)\in \+L_{\trun}\\\text{with $\+E\land \sigma$ satisfied by }\bm{x}}}\hat{p}_{(\+E,\+F,\sigma,\tau,T,\varsigma)}^X 
    &\leq 2^{-M},\\
    \frac{1}{|\Omega^{\+C}|}\sum\limits_{\bm{y}\in \Omega^{\+C}}\sum\limits_{\substack{(\+E,\+F,\sigma,\tau,T,\varsigma)\in \+L_{\trun}\\\text{with $\+F\land \tau$ satisfied by }\bm{y}}}\hat{p}_{(\+E,\+F,\sigma,\tau,T,\varsigma)}^Y
    &\leq 2^{-M}.
\end{align*}
\end{claim}
Assuming this claim holds, \Cref{lemma:linear-program-error-bound} can be proved using the same argument as in the proof of \cite[\text{Lemma 4.11}]{WY24}.
For completeness, we provide the proof here.
\begin{proof}[Proof of \Cref{lemma:linear-program-error-bound}]
Let $\hat{p}_{(\+E,\+F,\sigma,\tau,T,\varsigma)}^X$ and $\hat{p}_{(\+E,\+F,\sigma,\tau,T,\varsigma)}^Y$ denote a feasible solution of the linear program in \Cref{definition:linear-program}.
First, we claim the following equations hold for this feasible solution:
%   \begin{claim}\label{lemma:lp-sum-to-1}
\begin{equation}\label{eq:lp-sum-to-1}
\begin{aligned}
    \sum\limits_{\substack{(\+E,\+F,\sigma,\tau,T,\varsigma)\in \+L_{\textnormal{valid}}\\\text{with $\+E\land \sigma$ satisfied by }\bm{x}}}\hat{p}_{(\+E,\+F,\sigma,\tau,T,\varsigma)}^X
    &=1, \quad\textnormal{for all }\bm{x}\in \Omega^{\+C\setminus \{c_0\}},\\
     \sum\limits_{\substack{(\+E,\+F,\sigma,\tau,T,\varsigma)\in \+L_{\textnormal{valid}}\\\text{with $\+F\land \tau$ satisfied by }\bm{y}}}\hat{p}_{(\+E,\+F,\sigma,\tau,T,\varsigma)}^Y
     &=1, \quad\textnormal{for all }\bm{y}\in \Omega^{\+C}.
\end{aligned}   
\end{equation}
%    \end{claim}
These equations follow directly from verifying \Cref{definition:linear-program}. 
    By summing these equations over all $\bm{x}\in\Omega^{\+C\setminus \{c_0\}}$ and all $\bm{y}\in \Omega^{\+C}$, we obtain:
    \begin{equation}\label{eq:ss-sum}
    \begin{aligned}
                  |\Omega^{\+C\setminus \{c_0\}}|
                  &=\sum\limits_{\bm{x}\in \Omega^{\+C\setminus \{c_0\}}}\sum\limits_{\substack{(\+E,\+F,\sigma,\tau,T,\varsigma)\in \+L_{\text{valid}}\\\text{with $\+E\land \sigma$ satisfied by }\bm{x}}}\hat{p}_{(\+E,\+F,\sigma,\tau,T,\varsigma)}^X,\\
                      |\Omega^{\+C}|
                      &=\sum\limits_{\bm{y}\in \Omega^{\+C}}\sum\limits_{\substack{(\+E,\+F,\sigma,\tau,T,\varsigma)\in \+L_{\text{valid}}\\\text{with $\+F\land \tau$ satisfied by }\bm{y}}}\hat{p}_{(\+E,\+F,\sigma,\tau,T,\varsigma)}^Y.
    \end{aligned}
    \end{equation}
%    We will also prove the following claim at the end of this subsection.
Thus, we can express $|\Omega^{\+C\setminus \{c_0\}}|$ as follows:
\begin{align*}
(\text{by \eqref{eq:ss-sum}})&&
|\Omega^{\+C\setminus \{c_0\}}|
 =&\sum\limits_{\bm{x}\in \Omega^{\+C\setminus \{c_0\}}}\sum\limits_{\substack{(\+E,\+F,\sigma,\tau,T,\varsigma)\in \+L_{\text{valid}}\\\text{with $\+E\land \sigma$ satisfied by }\bm{x}}}\hat{p}_{(\+E,\+F,\sigma,\tau,T,\varsigma)}^X\\
&&=&\sum\limits_{\bm{x}\in \Omega^{\+C\setminus \{c_0\}}}\sum\limits_{\substack{(\+E,\+F,\sigma,\tau,T,\varsigma)\in \+L_{\coup}\\\text{with $\+E\land \sigma$ satisfied by }\bm{x}}}\hat{p}_{(\+E,\+F,\sigma,\tau,T,\varsigma)}^X\\
&&&+\sum\limits_{\bm{x}\in \Omega^{\+C\setminus \{c_0\}}}\sum\limits_{\substack{(\+E,\+F,\sigma,\tau,T,\varsigma)\in \+L_{\trun}\\\text{with $\+E\land \sigma$ satisfied by }\bm{x}}}\hat{p}_{(\+E,\+F,\sigma,\tau,T,\varsigma)}^X\\
(\text{by \Cref{lemma:linear-program-bad-leaf-loss}})
&&\leq & \sum\limits_{(\+E,\+F,\sigma,\tau,T,\varsigma)\in \+L_{\coup}}|\Omega^{\+E\land \sigma}|\cdot \hat{p}_{(\+E,\+F,\sigma,\tau,T,\varsigma)}^X+2^{-M}\cdot |\Omega^{\+C\setminus \{c_0\}}|. 
\end{align*}
Thus,  $|\Omega^{\+C\setminus \{c_0\}}|$ can be bounded as:
\begin{align*}
    \left|\Omega^{\+C\setminus \{c_0\}}\right|
    \in &\left[\hat{z}^X,\,\,\left(1+2\cdot 2^{-M}\right)\hat{z}^X \right],\\
    &\text{where }\hat{z}^X\defeq \sum\limits_{(\+E,\+F,\sigma,\tau,T,\varsigma)\in \+L_{\coup}}|\Omega^{\+E\land \sigma}|\cdot \hat{p}_{(\+E,\+F,\sigma,\tau,T,\varsigma)}^X.
\end{align*}
Similarly, $|\Omega^{\+C}|$ can also be bounded as:
\begin{align*}
    \left|\Omega^{\+C}\right|
    \in &\left[\hat{z}^Y,\,\,\left(1+2\cdot 2^{-M}\right)\hat{z}^Y \right],\\
    &\text{where }\hat{z}^Y\defeq \sum\limits_{(\+E,\+F,\sigma,\tau,T,\varsigma)\in \+L_{\coup}}|\Omega^{\+F\land \tau}|\cdot \hat{p}_{(\+E,\+F,\sigma,\tau,T,\varsigma)}^Y.
\end{align*}
Consequently, the ratio ${|\Omega^{\+C\setminus \{c_0\}}|}/{|\Omega^{\+C}|}$ is bounded as:
\begin{align*}
\frac{|\Omega^{\+C\setminus \{c_0\}}|}{|\Omega^{\+C}|}
&\leq 
\left(1+2\cdot 2^{-M}\right)\frac{\sum\limits_{(\+E,\+F,\sigma,\tau,T,\varsigma)\in \+L_{\coup}}|\Omega^{\+E\land \sigma}|\cdot \hat{p}_{(\+E,\+F,\sigma,\tau,T,\varsigma)}^X}{\sum\limits_{(\+E,\+F,\sigma,\tau,T,\varsigma)\in \+L_{\coup}}|\Omega^{\+F\land \tau}|\cdot \hat{p}_{(\+E,\+F,\sigma,\tau,T,\varsigma)}^Y}\\
&\leq \left(1+2\cdot 2^{-M}\right)r_{+},
\end{align*}
where the last inequality follows from the leaf constraints of the LP (\Cref{item:linear-program-leaf-constraints} in \Cref{definition:linear-program}).

By symmetry, we also have:
\begin{align*}
\frac{|\Omega^{\+C\setminus \{c_0\}}|}{|\Omega^{\+C}|}
&\geq \left(1-2\cdot 2^{-M}\right)\frac{\sum\limits_{(\+E,\+F,\sigma,\tau,T,\varsigma)\in \+L_{\coup}}|\Omega^{\+E\land \sigma}|\cdot \hat{p}_{(\+E,\+F,\sigma,\tau,T,\varsigma)}^X}{\sum\limits_{(\+E,\+F,\sigma,\tau,T,\varsigma)\in \+L_{\coup}}|\Omega^{\+F\land \tau}|\cdot \hat{p}_{(\+E,\+F,\sigma,\tau,T,\varsigma)}^Y}\\
&\geq \left(1-2\cdot 2^{-M}\right)r_{-}.
\end{align*}
Combining these results, we conclude:
\[
\left(1-2\cdot 2^{-M}\right)r_-\leq \frac{|\Omega^{\+C\setminus \{c_0\}}|}{|\Omega^{\+C}|}\leq \left(1+2\cdot 2^{-M}\right)r_+.\qedhere
\]
\end{proof}

\subsubsection{Proof of the average marginal bound (\Cref{lemma:linear-program-bad-leaf-loss})}
To prove \Cref{lemma:linear-program-bad-leaf-loss},
%We only prove the first inequality and the other one follows analogously. 
we introduce an auxiliary random process induced by the LP feasible solution. 
%This random process is used solely in proving \Cref{lemma:linear-program-bad-leaf-loss}.
This random process plays a key role in the proof of \Cref{lemma:linear-program-bad-leaf-loss},
and also forms the foundation for the sampling algorithm derived from the LP.

\begin{definition}[random path induced by the LP solution]\label{definition:random-process-lp}
We define the following random process for generating a random path from the root to a leaf in the $M$-truncated coupling tree $\+T=\+T_M(\Phi,c_0)$,
based on a feasible solution $\hat{p}_{(\+E,\+F,\sigma,\tau,T,\varsigma)}^X$, $\hat{p}_{(\+E,\+F,\sigma,\tau,T,\varsigma)}^Y$  to the linear program in \Cref{definition:linear-program}.

Let $\*X^{\lp}\sim \mu_{\+C\setminus \{c_0\}}$ be a random satisfying assignment. 
The random process starts at the root $(\+C\setminus \{c_0\},\+C,\varnothing,\varnothing,\emptyset,\varnothing)$
and proceeds as follows at each non-leaf node $(\+E,\+F,\sigma,\tau,T,\varsigma)$ in $\+T$:
\begin{enumerate}
        \item If $\+F^{\tau}\not\subseteq \+E^{\sigma}$, let $c$ be the smallest constraint in $\+F^{\tau}\setminus \+E^{\sigma}$:
        \begin{itemize}
            \item If $c$ is satisfied by $\*X^{\lp}$, move to the child node $(\+E\cup \{c\},\+F,\sigma,\tau,T,\varsigma)$.
            \item Otherwise, for each $\rho\in [q]^{\vbl(c)}$, move to the child node $(\+E,\+F,\sigma\land \*X^{\lp}_{\vbl(c)},\tau\land \rho,T\oplus \unpin{c},\varsigma\land \rho)$  with probability:
            \[
            \frac{\hat{p}^X_{(\+E,\+F,\sigma\land \*X^{\lp}_{\vbl(c)},\tau\land \rho,T\oplus \unpin{c},\varsigma\land \rho)}}{\hat{p}_{(\+E,\+F,\sigma,\tau,T,\varsigma)}^X}.
            \]
        \end{itemize}
        \item Otherwise, if $\+F^{\tau}\subseteq \+E^{\sigma}$, let $c$ be the smallest constraint in $\+E^{\sigma}\setminus \+F^{\tau}$:
        \begin{itemize}
            \item Move to the child node $(\+E,\+F\cup \{c\},\sigma,\tau,T,\varsigma)$ with probability:
            \[
            \frac{\hat{p}^X_{(\+E,\+F\cup \{c\},\sigma,\tau,T,\varsigma)}}{\hat{p}_{(\+E,\+F,\sigma,\tau,T,\varsigma)}^X};
            \]
            \item Alternatively, move to the child node $(\+E,\+F,\sigma\land \*X^{\lp}_{\vbl(c)},\tau\land \rho,T\oplus \unpin{c},\varsigma\land \*X^{\lp}_{\vbl(c)})$, where $\rho=\vio{c}$,  with probability:  
            \[
            \frac{\hat{p}^X_{(\+E,\+F,\sigma\land \*X^{\lp}_{\vbl(c)},\tau\land \rho,T\oplus \unpin{c},\varsigma\land \*X^{\lp}_{\vbl(c)})}}{\hat{p}_{(\+E,\+F,\sigma,\tau,T,\varsigma)}^X}.
            \]
        \end{itemize}
\end{enumerate}
Let $P^{\lp}$ denote the random  path generated by this process, and let $\mu^{\lp}$ denote its distribution.
\end{definition}

By \Cref{item:linear-program-non-leaf-constraints} of \Cref{definition:linear-program}, 
it is straightforward to verify that the random process in  \Cref{definition:random-process-lp} generates  a root-to-leaf path $P^{\lp}$ in $\+T$.
We also have the following probability bounds.

\begin{lemma}\label{lemma:random-process-lp-probability-bound}
Assume $\Phi$ is satisfiable. For each $(\+E,\+F,\sigma,\tau,T,\varsigma)\in V(\+T)$, it holds that
\[
\Pr{(\+E,\+F,\sigma,\tau,T,\varsigma)\in P^{\lp}}=\mu_{\+C\setminus \{c_0\}}(\+E\land \sigma)\cdot \hat{p}_{(\+E,\+F,\sigma,\tau,T,\varsigma)}^X.
\]
Moreover, conditioned on $(\+E,\+F,\sigma,\tau,T,\varsigma)\in P^\lp$, it follows that
\[
\*X^{\lp}\sim \mu_{\+E}^{\sigma},
\]
for each $(\+E,\+F,\sigma,\tau,T,\varsigma)$ such that 
$\Pr{(\+E,\+F,\sigma,\tau,T,\varsigma)\in P^{\lp}}>0$.
\end{lemma}

\begin{proof}
We prove the lemma by structural induction in the top-down manner. 
The induction basis is the root node $(\+C\setminus \{c_0\},\+C,\varnothing,\varnothing,\emptyset,\varnothing)$, 
where the lemma follows easily from \Cref{definition:random-process-lp}.

     For the induction step, we only consider the case when $\+F^{\tau}\not\subseteq \+E^{\sigma}$, 
     as the complementary case with $\+F^{\tau}\subseteq \+E^{\sigma}$ is straightforward. 
     Let $c$ be the smallest constraint in $ \+F^{\tau}\setminus \+E^{\sigma}$. 
     %From \Cref{definition:random-process-lp}, 
     We have the following cases:  
     \begin{itemize}
        \item \textbf{Case 1}: $c$ is satisfied by $\*X^{\lp}$. The process transitions to $(\+E\cup \{c\},\+F,\sigma,\tau,T,\varsigma)$. By the induction hypothesis, we know that $\*X^{\lp}\sim \mu_{\+E}^{\sigma}$. 
        Therefore, the event that $c$ is satisfied by $\*X^{\lp}$ occurs with probability $\mu_{\+E}^{\sigma}(c)$. 
        Consequently, we have
        \begin{align*}
     &\Pr{(\+E\cup \{c\},\+F,\sigma,\tau,T,\varsigma)\in P^{\lp}}\\
     (\text{by \Cref{definition:random-process-lp}})\quad=&  \Pr{(\+E,\+F,\sigma,\tau,T,\varsigma)\in P^{\lp}}\cdot  \mu_{\+E}^{\sigma}(c)\\
     (\text{by I.H.})\quad=&  \mu_{\+C\setminus \{c_0\}}(\+E\land \sigma)\cdot \hat{p}_{(\+E,\+F,\sigma,\tau,T,\varsigma)}^X\cdot \mu_{\+E}^{\sigma}(c)\\
     (\text{by the chain rule})\quad=& \mu_{\+C\setminus \{c_0\}}((\+E\cup \{c\})\land \sigma)\cdot \hat{p}_{(\+E,\+F,\sigma,\tau,T,\varsigma)}^X.
         \end{align*}
   Additionally, conditioning on moving to $(\+E\cup \{c\},\+F,\sigma,\tau,T,\varsigma)$, it follows that
    \[
        \*X^{\lp}\sim \mu_{\+E\cup \{c\}}^{\sigma},        \]
thus completing the proof of this case.
         
         \item \textbf{Case 2}: $c$ is violated by $\*X^{\lp}$. Let $\pi=\*X^{\lp}_{\vbl(c)}$, then we have $\pi=\vio{c}$.  By the induction hypothesis, we know that $\*X^{\lp}\sim \mu_{\+E}^{\sigma}$. Therefore, the event that $c$ is violated by $\*X^{\lp}$ occurs with probability \[
         \mu_{\+E}^{\sigma}(\neg c) = \mu_{\+E}^{\sigma}(\pi).
         \]
         In this case, for each $\rho\in [q]^{\vbl(c)}$, the process transitions to $(\+E,\+F,\sigma\land \pi,\tau\land \rho,T\oplus \unpin{c},\varsigma\land \rho)$ with probability $\frac{\hat{p}^X_{(\+E,\+F,\sigma\land \pi,\tau\land \rho,T\oplus \unpin{c},\varsigma\land \rho)}}{\hat{p}_{(\+E,\+F,\sigma,\tau,T,\varsigma)}^X}$.
        Therefore, for each $\rho\in [q]^{\vbl(c)}$, we have
          \begin{align*}
     &\Pr{(\+E,\+F,\sigma\land \pi,\tau\land \rho,T\oplus \unpin{c},\varsigma\land \rho)\in P^{\lp}}\\
    (\text{by \Cref{definition:random-process-lp}})\quad=&  \Pr{(\+E,\+F,\sigma,\tau,T,\varsigma)\in P^{\lp}}\cdot \mu_{\+E}^{\sigma}(\pi)\cdot \frac{\hat{p}^X_{(\+E,\+F,\sigma\land \pi,\tau\land \rho,T\oplus \unpin{c},\varsigma\land \rho)}}{\hat{p}_{(\+E,\+F,\sigma,\tau,T,\varsigma)}^X}\\
     (\text{by I.H.})\quad=&  \mu_{\+C\setminus \{c_0\}}(\+E\land \sigma)\cdot \hat{p}_{(\+E,\+F,\sigma,\tau,T,\varsigma)}^X\cdot \mu_{\+E}^{\sigma}(\pi)\cdot \frac{\hat{p}^X_{(\+E,\+F,\sigma\land \pi,\tau\land \rho,T\oplus \unpin{c},\varsigma\land \rho)}}{\hat{p}_{(\+E,\+F,\sigma,\tau,T,\varsigma)}^X}\\
    (\text{by the chain rule})\quad=& \mu_{\+C\setminus \{c_0\}}(\+E\land \sigma\land \pi)\cdot \hat{p}^X_{(\+E,\+F,\sigma\land \pi,\tau\land \rho,T\oplus \unpin{c},\varsigma\land \rho)}.
         \end{align*}
     \end{itemize}
        This concludes the proof for the case when $\+F^{\tau}\not\subseteq \+E^{\sigma}$. As we discussed earlier, it proves the lemma.
\end{proof}

The following corollary follows from \Cref{lemma:random-process-lp-probability-bound} by the definition of conditional probability.

\begin{corollary}\label{corollary:lp-procedure-1-probability-bound}
Assume $\Phi$ is satisfiable. For each $\bm{x}\in \Omega^{\+C\setminus \{c_0\}}$ and each $(\+E,\+F,\sigma,\tau,T,\varsigma)\in V(\+T)$, 
\[
\Pr{(\+E,\+F,\sigma,\tau,T,\varsigma)\in P^{\lp}\mid \*X^{\lp}=\bm{x}}=\begin{cases}
0 & \text{if }\bm{x}\notin \Omega^{\+E\land \sigma},\\
\hat{p}_{(\+E,\+F,\sigma,\tau,T,\varsigma)}^X & \text{if }\bm{x}\in \Omega^{\+E\land \sigma}.
\end{cases}
\]
\end{corollary}

\begin{proof}
For each $\bm{x}\in \Omega$ and each $(\+E,\+F,\sigma,\tau,T,\varsigma)\in V(\+T)$, we have
\begin{equation}\label{eq:lp-procedure-transform}
\begin{aligned}
    \Pr{(\+E,\+F,\sigma,\tau,T,\varsigma)\in P^\lp\mid \*X^{\lp}=\bm{x}}=&\frac{\Pr{(\+E,\+F,\sigma,\tau,T,\varsigma)\in P^\lp\land \*X^{\lp}=\bm{x}}}{\Pr{\*X^{\lp}=\bm{x}}}\\
    (\text{by \Cref{definition:random-process-lp}})\quad =&|\Omega^{\+C\setminus \{c_0\}}|\cdot \Pr{(\+E,\+F,\sigma,\tau,T,\varsigma)\in P^\lp\land \*X^{\lp}=\bm{x}}.
\end{aligned}
\end{equation}
If $\bm{x}\notin  \Omega^{\+E\land \sigma}$,
then by \Cref{lemma:random-process-lp-probability-bound},  the above expression equals $0$. 
Therefore, we assume $\bm{x}\in  \Omega^{\+E\land \sigma}$. 
By \Cref{lemma:random-process-lp-probability-bound}, conditioning on the event $ (\+E,\+F,\sigma,\tau,T,\varsigma)$, the distribution of $\sigma$ is uniform over all assignments in $\Omega^{\+E\land \sigma}$. 
Thus, we have
\begin{align*}
\Pr{(\+E,\+F,\sigma,\tau,T,\varsigma)\in P^{\lp}\land \*X^{\lp}=\bm{x}}=& \Pr{(\+E,\+F,\sigma,\tau,T,\varsigma)\in P^{\lp}}\cdot \frac{1}{|\Omega^{\+E\land \sigma}|}\\
(\text{by \Cref{lemma:random-process-lp-probability-bound}})\quad=& \frac{|\Omega^{\+E\land \sigma}|}{|\Omega^{\+C\setminus \{c_0\}}|}\cdot \hat{p}_{(\+E,\+F,\sigma,\tau,T,\varsigma)}^X\cdot \frac{1}{|\Omega^{\+E\land \sigma}|}\\
=& \frac{1}{|\Omega^{\+C\setminus \{c_0\}}|}\cdot \hat{p}_{(\+E,\+F,\sigma,\tau,T,\varsigma)}^X,
\end{align*}
combining this with \eqref{eq:lp-procedure-transform} completes the proof of the corollary.
\end{proof}

We are now ready to prove \Cref{lemma:linear-program-bad-leaf-loss}.
\begin{proof}[Proof of \Cref{lemma:linear-program-bad-leaf-loss}]
We will prove the first inequality; the second inequality can be proved by following a similar approach.
  
First, we verify the following identity:
\begin{equation}\label{eq:linear-program-leaf-loss-1}
\begin{aligned}
&\frac{1}{|\Omega^{\+C\setminus \{c_0\}}|}\sum\limits_{\bm{x}\in \Omega^{\+C\setminus \{c_0\}}}\sum\limits_{\substack{(\+E,\+F,\sigma,\tau,T,\varsigma)\in \+L_{\trun}\\\text{with $\+E\land \sigma$ satisfied by }\bm{x}}}\hat{p}_{(\+E,\+F,\sigma,\tau,T,\varsigma)}^X\\
%(\star)\quad
=&\sum\limits_{\bm{x}\in \Omega^{\+C\setminus \{c_0\}}}\sum\limits_{\substack{(\+E,\+F,\sigma,\tau,T,\varsigma)\in \+L_{\trun}\\\text{with $\+E\land \sigma$ satisfied by }\bm{x}}}\frac{\Pr{(\+E,\+F,\sigma,\tau,T,\varsigma)\in P^\lp}}{|\Omega^{\+E\land \sigma}|}\\
%(\blacktriangle)\quad
=&\sum\limits_{(\+E,\+F,\sigma,\tau,T,\varsigma)\in \+L_{\trun}}\Pr{(\+E,\+F,\sigma,\tau,T,\varsigma)\in P^\lp},
\end{aligned}
\end{equation}
where the first equality follows from \Cref{lemma:random-process-lp-probability-bound} and the fact that for each $(\+E,\+F,\sigma,\tau,T,\varsigma)\in V(\+T)$, any assignment $\bm{x}$ that satisfies $\+E\land \sigma$ must also satisfy $\+C\setminus \{c_0\}$. 
The second equality follows from this same fact and the exchange of the order of summation.

Combining these results, we have the following: 
\begin{align*}
&\sum\limits_{(\+E,\+F,\sigma,\tau,T,\varsigma)\in \+L_{\trun}}\Pr{(\+E,\+F,\sigma,\tau,T,\varsigma)\in P^\lp}\\
\leq & \sum\limits_{T'\in \=T^{c_0}_M}\sum\limits_{\varsigma'\in [q]^{\vbl(T')}}\sum\limits_{\substack{(\+E,\+F,\sigma,\tau,T,\varsigma)\in P^\lp\\ T=T'\land \varsigma=\varsigma'}} \Pr{(\+E,\+F,\sigma,\tau,T,\varsigma)\in P^\lp}\\
\leq & \sum\limits_{T'\in \=T^{c_0}_M}\sum\limits_{\varsigma'\in [q]^{\vbl(T')}}\sum\limits_{\substack{(\+E,\+F,\sigma,\tau,T,\varsigma)\in P^\lp\\ T=T'\land \varsigma=\varsigma'}} \sum\limits_{\bm{x}\in \Omega^{\+E\land \sigma}}\Pr{ \*X^{\lp}=\bm{x}}\cdot \Pr{(\+E,\+F,\sigma,\tau,T,\varsigma)\in P^\lp \mid \*X^{\lp}=\bm{x}}\\
 =&\sum\limits_{T'\in \=T^{c_0}_M}\sum\limits_{\varsigma'\in [q]^{\vbl(T')}}\sum\limits_{\substack{(\+E,\+F,\sigma,\tau,T,\varsigma)\in P^\lp\\ T=T'\land \varsigma=\varsigma'}} \sum\limits_{\bm{x}\in \Omega^{\+E\land \sigma}} \Pr{ \*X^{\lp}=\bm{x}}\cdot\hat{p}_{(\+E,\+F,\sigma,\tau,T,\varsigma)}^X.
\end{align*}
Here, the last equality follows directly from \Cref{corollary:lp-procedure-1-probability-bound}. 

Note that by \Cref{lemma:conditional-deterministic}, fixing the witness tree $T$ and the witness assignment $\varsigma$ uniquely identifies a node $(\+E,\+F,\sigma,\tau,T,\varsigma)\in P^\lp$ satisfying $T=T'\land \varsigma=\varsigma'$. We then denote this unique node as $N(T',\sigma')$, and the partial assignment $\sigma$ in this unique node as $\sigma(T',\varsigma')$. Therefore, we further have

\begin{align}
&\sum\limits_{T'\in \=T^{c_0}_M}\sum\limits_{\varsigma'\in [q]^{\vbl(T')}}\sum\limits_{\substack{(\+E,\+F,\sigma,\tau,T,\varsigma)\in P^\lp\\ T=T'\land \varsigma=\varsigma'}} \sum\limits_{\bm{x}\in \Omega^{\+E\land \sigma}} \Pr{ \*X^{\lp}=\bm{x}}\cdot\hat{p}_{(\+E,\+F,\sigma,\tau,T,\varsigma)}^X \nonumber\\
(\star)\quad \leq & \sum\limits_{T'\in \=T^{c_0}_M}\sum\limits_{\varsigma'\in [q]^{\vbl(T')}} \Pr{ \*X^{\lp}_{\vbl(T')}=\sigma(T',\varsigma')}\cdot \hat{p}_{N(T',\sigma')}^X \nonumber\\ 
(\blacktriangle)\quad \leq &\sum\limits_{T'\in \=T^{c_0}_M}\sum\limits_{\varsigma'\in [q]^{\vbl(T')}}  \Pr{ \*X^{\lp}_{\vbl(T')}=\sigma(T',\varsigma')}\cdot q^{-|\vbl(T')\cap \vgood|}\cdot \left( 1-\mathrm{e}q^{-(1-\eps_1)k} \right)^{-|\vbl(T')|p_1\alpha} \nonumber\\
(\blacksquare)\quad \leq & \sum\limits_{T'\in \=T^{c_0}_M}\sum\limits_{\varsigma'\in [q]^{\vbl(T')}} q^{-2|\vbl(T')\cap \vgood|}\cdot \left( 1-\mathrm{e}q^{-(1-\eps_1)k} \right)^{-2|\vbl(T')|p_1\alpha} \nonumber\\
\leq & \sum\limits_{T'\in \=T^{c_0}_M} q^{|\vbl(T')|}\cdot q^{-2|\vbl(T')\cap \vgood|}\cdot \left( 1-\mathrm{e}q^{-(1-\eps_1)k} \right)^{-2|\vbl(T')|p_1\alpha} \label{eq:same-as-coupling}\\
(\triangle)\quad\leq &2^{-M}. \nonumber
\end{align}
The $(\star)$ inequality is derived using the argument above. The $(\blacktriangle)$ inequality follows from \Cref{item:linear-program-overflow-constraints} of \Cref{definition:linear-program}. 
The $(\blacksquare)$ inequality follows from \Cref{lemma:assignment-probability-bound}. 
Finally, the $(\triangle)$ inequality holds for the chosen parameters in \Cref{condition:main-condition} and was already established in the proof of \Cref{theorem:correlation-decay}; in particular, \eqref{eq:same-as-coupling} is exactly \eqref{eq:used-by-LP} and hence the same proof applies.
\end{proof}

\subsection{Sampling and counting via linear programming}
In this subsection, we apply the linear program in \Cref{definition:linear-program} to derive sampling and counting algorithms and prove \Cref{theorem:CSP-main}. The proofs essentially follow those in ~\cite[Section 5]{WY24}, and we include for completeness.

We begin with the counting algorithm, which directly follows from the analysis of the linear program presented in the previous subsection.

\begin{proof}[Proof of the counting part of \Cref{theorem:CSP-main}]
Let $\+C=\{c_1,c_2,\dots,c_m\}$. Note that $km=n\alpha$.  
For each $0\leq i\leq m$, define $\+C_i=\{c_1,c_2,\dots,c_i\}$, $\Phi_i=(V,[q],\+C_i)$. 
In particular, $\Phi_m=\Phi$. 
By applying a \emph{constraint-wise self-reduction}, we decompose $Z(\Phi_m)$ into the following telescopic product:
\begin{equation}
    Z(\Phi_m)=Z(\Phi_0)\cdot\frac{Z(\Phi_{m})}{Z(\Phi_0)}=Z(\Phi_0)\cdot \prod\limits_{i=1}^{m}\frac{Z(\Phi_i)}{Z(\Phi_{i-1})}.
\end{equation}
Note that $\Phi_0$ is a trivial CSP formula, so $Z_{\Phi_0}=q^{|V|}$.
Moreover, and crucially, observe that each $\Phi_i$ satisfies \Cref{condition:main-condition} for $0\leq i\leq m$, assuming the original CSP formula $\Phi$ satisfies \Cref{condition:main-condition}. Setting:
\[
M=1+\log \frac{4m}{\eps},
\]
we use the LP in \Cref{definition:linear-program} with binary search to approximate each ratio $\frac{Z(\Phi_i)}{Z(\Phi_{i-1})}=\frac{|\Omega^{\+C_i}|}{|\Omega^{\+C_{i-1}}|}$ within a multiplicative error of $\frac{\varepsilon}{4m}$. 
According to \Cref{lemma:lp-properties,lemma:lp-feasibility,lemma:linear-program-error-bound}, 
this can be done in time $O\left(\left(\frac{n}{\varepsilon}\right)^{\poly(k,\log q.\alpha)}\right)$. 
Therefore, we obtain an estimate of $Z(\Phi)=Z(\Phi_m)$ within a total multiplicative error of $\varepsilon$ in time $O\left(\left(\frac{n}{\varepsilon}\right)^{\poly(k,\log q,\alpha)}\right)$, as required.
\end{proof}

For the sampling part, we establish a constraint-wise self-reduction for sampling by constructing a \emph{dynamic sampler}.
Given an atomic CSP formula $\Phi=(V,[q],\+C)$ and a constraint $c\in \+C\setminus \{c_0\}$, it dynamically updates a current sample $\sigma^{\-{in}}\sim \mu_{\+C\setminus \{c_0\}}$ to a new sample $\sigma^{\-{out}}\sim \mu_{\+C}$. 
This dynamic sampler is derived from the random process in \Cref{definition:random-process-lp}, which was initially introduced to analyze the linear program.

\begin{definition}[dynamic sampler]\label{definition:marginal-sampling-algorithm}
The algorithm takes as input a CSP $\Phi=(V,[q],\+C)$, a constraint $c_0\in \+C$, and  an error bound $\varepsilon\in (0,1)$.
Additionally, the algorithm is also given access to an assignment $\sigma^{\text{in}}\in [q]^V$ that satisfies $\Phi^{c_0}=(V,[q],\+C\setminus \{c_0\})$. 

    The algorithm proceeds as follows to update  $\sigma^{\text{in}}$ to a new assignment $\sigma^{\text{out}}\in [q]^V$:
    \begin{itemize}
    \item   Set the truncation threshold:
        \[
            M=1+\log \frac{4}{\varepsilon}.
        \]
        Construct the LP on the $M$-truncated coupling tree $\+T=\+T_M(\Phi,c_0)$, as described in \Cref{definition:linear-program}.
    \item Use binary search to find an interval $[r_{-},r_{+}]$ such that $r_{-}\geq \frac{4+\varepsilon}{4+2\varepsilon}r_{+}$ and the LP remains feasible for parameters $r_{-}$ and $r_{+}$. 
         Let $\left\{\hat{p}_{(\+E,\+F,\sigma,\tau,T,\varsigma)}^X,\hat{p}_{(\+E,\+F,\sigma,\tau,T,\varsigma)}^Y\right\}_{(\+E,\+F,\sigma,\tau,T,\varsigma)\in V(\+T)}$ be the corresponding LP feasible solution. 
     \item Simulate the random process from \Cref{definition:random-process-lp} using the feasible solution obtained above, 
     replacing the random assignment $\*X^{\lp}$ (used as the random seed for the process) with $\sigma^{\text{in}}$.
     Let $(\bm{\+E},\bm{\+F},\bm{\sigma},\bm{\tau},\bm{T},\bm{\varsigma})\in \+L$ be the random leaf reached by this process.
          \item  If the leaf node $(\bm{\+E},\bm{\+F},\bm{\sigma},\bm{\tau},\bm{T},\bm{\varsigma})\notin \+L_{\coup}$, 
          then set $\sigma^{\text{out}}\in [q]^V$ to be an arbitrary satisfying assignment of $\Phi$. 
          Otherwise, update  $\sigma^{\text{in}}$  by modifying the assigned variables according to $\bm{\tau}$, i.e., set $\sigma^{\text{out}}\gets \bm{\tau}\land \sigma^{\text{in}}_{V\setminus\Lambda(\bm{\tau})}$.
    \end{itemize}
\end{definition}

We will prove the following lemma, which shows that the new sample produced by the dynamic sampler approximates the target distribution $\mu_{\+C}$.

\begin{lemma}\label{lemma:marginal-sampler-tvd-bound}
Assume \Cref{condition:main-condition} and that $\sigma^{\textnormal{in}}\sim \mu_{\+C\setminus \{c_0\}}$.
Then, the assignment $\sigma^{\textnormal{out}}$ produced by the dynamic sampler in \Cref{definition:marginal-sampling-algorithm} follows a distribution $\mu^{\textnormal{out}}$ that satisfies
\[
\dtv(\mu^{\textnormal{out}},\mu_{\+C})\leq \varepsilon.
\]
\end{lemma}

With \Cref{lemma:marginal-sampler-tvd-bound}, we can complete the proof of \Cref{theorem:CSP-main}.
\begin{proof}[Proof of the sampling part of \Cref{theorem:CSP-main}]
Let $\+C=\{c_1,c_2,\dots,c_m\}$. Note that $km=n\alpha$.  
For each $0\leq i\leq m$, define $\+C_i=\{c_1,c_2,\dots,c_i\}$, $\Phi_i=(V,[q],\+C_i)$ and let $\mu_{\Phi_i}$ denote the uniform distribution over satisfying assignments of $\Phi_i$. 
Let $\sigma_0,\sigma_1,\dots,\sigma_m\in [q]^V$ be constricted as follows: 
\begin{itemize}
    \item Sample $\sigma_0\sim \+P$, the uniform product distribution.
    \item For each $1\leq i\leq m$, use the dynamic sampler from \Cref{definition:marginal-sampling-algorithm} to generate $\sigma_i$ from $\sigma_{i-1}$, with the input formula $\Phi_i=(V,[q],\+C_{i})$, input constraint $c_{i}$, error bound $\frac{\varepsilon}{m}$, and current sample $\sigma_{i-1}$. 
    If $\sigma_i$ is not a solution to $\Phi_i$ at any step $0\leq i\leq m$, output an arbitrary final assignment instead. 
\end{itemize}
By applying induction  along with \Cref{lemma:marginal-sampler-tvd-bound}, we can verify that for each $0\leq i\leq m$, the  distribution $\mu_i$ of $\sigma_i$ satisfies
\[
%\forall 0\leq i\leq m, \quad 
\dtv(\mu_i,\mu_{\+C_{i}})\leq \frac{i}{m}\cdot \varepsilon.
\]
The total running time is bounded by $O\left(\left(\frac{n}{\varepsilon}\right)^{\poly(k,\log q.\alpha)}\right)$  according to \Cref{lemma:lp-properties}, which accounts for the overhead of building and solving the linear program.  
Thus, the theorem is proved.
\end{proof}

We then finish this section by proving \Cref{lemma:marginal-sampler-tvd-bound}. 
\begin{proof}[Proof of \Cref{lemma:marginal-sampler-tvd-bound}]
 According to \Cref{definition:marginal-sampling-algorithm} and \Cref{lemma:random-process-lp-probability-bound}, for each $\bm{y}\in \Omega^{\+C}$, we have
  \begin{align*}
\Pr{\sigma^{\text{out}}=\bm{y}}= &\sum\limits_{\substack{(\+E,\+F,\sigma,\tau,T,\varsigma)\in \+L_{\coup}\\\text{with $\+F\land \tau$ satisfied by }\bm{y}}}\mu_{\+C\setminus \{c_0\}}(\+E\land \sigma)\cdot \hat{p}_{(\+E,\+F,\sigma,\tau,T,\varsigma)}^X\cdot \frac{1}{|\Omega^{\+E\land \sigma}|}\\
(\text{by \Cref{item:linear-program-coupled} of \Cref{definition:linear-program}})\quad\geq & r_{-}\cdot \sum\limits_{\substack{(\+E,\+F,\sigma,\tau,T,\varsigma)\in \+L_{\coup}\\\text{with $\+F\land \tau$ satisfied by }\bm{y}}}\mu_{\+C\setminus \{c_0\}}(\+E\land \sigma)\cdot \hat{p}_{(\+E,\+F,\sigma,\tau,T,\varsigma)}^Y\cdot \frac{1}{|\Omega^{\+E\land \sigma}|}\\
 = &r_{-}\cdot \frac{1}{|\Omega^{\+C\setminus \{c_0\}}|}\cdot\sum\limits_{\substack{(\+E,\+F,\sigma,\tau,T,\varsigma)\in \+L_{\coup}\\\text{with $\+F\land \tau$ satisfied by }\bm{y}}} \hat{p}_{(\+E,\+F,\sigma,\tau,T,\varsigma)}^Y.
\end{align*}
 Here, the first equality holds because $\+E^{\sigma}=\+F^{\tau}$ for each $(\+E,\+F,\sigma,\tau,T,\varsigma)\in \+L_{\coup}$, meaning each solution in $\Omega^{\+F\land \tau}$ is generated with equal probability $\frac{1}{\abs{\Omega^{\+F\land \tau}}}=\frac{1}{\abs{\Omega^{\+E\land \sigma}}}$ as established in \Cref{lemma:random-process-lp-probability-bound}. 
 The last equality follows from the identity $\mu_{\+C\setminus \{c_0\}}(\+E\land \sigma)=\frac{|\Omega^{\+E\land \sigma}|}{|\Omega^{\+C\setminus \{c_0\}}|}$, because $\+E\land \sigma\implies \+C\setminus \{c_0\}$ for each $(\+E,\+F,\sigma,\tau,T,\varsigma)\in V(\+T)$, which holds  by induction.

Therefore, we can construct a distribution $\nu$ over $[q]^V$ such that 
the measure of each  $\bm{y}\in \Omega^{\+C}$ is
\begin{align}
%\forall \bm{y}\in \Omega^{\+C},\qquad  
\nu(\bm{y})
&=\frac{1}{|\Omega^{\+C}|}\cdot\sum\limits_{\substack{(\+E,\+F,\sigma,\tau,T,\varsigma)\in \+L_{\coup}\\\text{with $\+F\land \tau$ satisfied by }\bm{y}}}\hat{p}_{(\+E,\+F,\sigma,\tau,T,\varsigma)}^Y,\label{eq:definition-nu}
\end{align}
and meanwhile the total variation distance between $\nu$ and the distribution $\mu^{\text{out}}$ of $\sigma^{\text{out}}$ is bounded as:
\begin{align}
\dtv\left(\nu,\mu^{\text{out}}\right)
&\le\abs{\frac{1}{|\Omega^{\+C}|}-r_{-}\cdot \frac{1}{|\Omega^{\+C\setminus \{c_0\}}|}}\sum\limits_{\bm{y}\in \Omega^{\+C}}\sum\limits_{\substack{(\+E,\+F,\sigma,\tau,T,\varsigma)\in \+L_{\coup}\\\text{with $\+F\land \tau$ satisfied by }\bm{y}}}\hat{p}_{(\+E,\+F,\sigma,\tau,T,\varsigma)}^Y.\label{eq:nu-design}
\end{align}

Given that  $\left(1-\frac{\varepsilon}{4}\right)r_{-}\cdot \frac{|\Omega^{\+C\setminus \{c_0\}}|}{|\Omega^{\+C}|}\leq \left(1+\frac{\varepsilon}{4}\right)r_+$ and $r_{-}\geq \frac{4+\varepsilon}{4+2\varepsilon}r_{+}$ are ensured by the process, 
we obtain
\begin{equation}\label{eq:dtv-transform}
\left(1-\frac{\varepsilon}{2}\right)\cdot \frac{|\Omega^{\+C\setminus \{c_0\}}|}{|\Omega^{\+C}|}\leq r_{-}\leq \left(1+\frac{\varepsilon}{2}\right)\cdot \frac{|\Omega^{\+C\setminus \{c_0\}}|}{|\Omega^{\+C}|}.
\end{equation}
Combining \eqref{eq:nu-design} and \eqref{eq:dtv-transform}, we have
\begin{equation}\label{eq:dtv-1}
 \dtv(\nu,\mu^{\text{out}})\leq \frac{\varepsilon}{2}\cdot \frac{1}{|\Omega^{\+C}|}\sum\limits_{\bm{y}\in \Omega^{\+C}}\sum\limits_{\substack{(\+E,\+F,\sigma,\tau,T,\varsigma)\in \+L_{\coup}\\\text{with $\+F\land \tau$ satisfied by }\bm{y}}}\hat{p}_{(\+E,\+F,\sigma,\tau,T,\varsigma)}^Y\leq \frac{\varepsilon}{2}.
 \end{equation}

Now, consider the total variation distance between $\nu$ and $\mu_{\+C}$, the uniform distribution over $\Omega^{\+C}$. 
By \eqref{eq:definition-nu} and \eqref{eq:lp-sum-to-1}, for each $\bm{y} \in \Omega^{\+C}$ we have 
  $\nu(\bm{y})\leq \frac{1}{|\Omega^{\+C}|}=\mu_{\+C}(\bm{y})$.
Thus, we have
\begin{align*}
\dtv(\mu_{\+C},\nu)=&\sum\limits_{\bm{y}\in \Omega^{\+C}}\left(\mu_{\+C}(\bm{y}-\nu(\bm{y})\right)\\
=&\sum\limits_{\bm{y}\in \Omega^{\+C}}\left(\frac{1}{|\Omega^{\+C}|}-\frac{1}{|\Omega^{\+C}|}\cdot\sum\limits_{\substack{(\+E,\+F,\sigma,\tau,T,\varsigma)\in \+L_{\coup}\\\text{with $\+F\land \tau$ satisfied by }\bm{y}}} \hat{p}_{(\+E,\+F,\sigma,\tau,T,\varsigma)}^Y\right)\\
(\text{by \eqref{eq:lp-sum-to-1}})\quad= &\frac{1}{|\Omega^{\+C}|}\sum\limits_{\bm{y}\in \Omega^{\+C}}  \sum\limits_{\substack{(\+E,\+F,\sigma,\tau,T,\varsigma)\in \+L_{\coup}\\\text{with $\+F\land \tau$ satisfied by }\bm{y}}} \hat{p}_{(\+E,\+F,\sigma,\tau,T,\varsigma)}^Y\\
(\text{by \Cref{lemma:linear-program-bad-leaf-loss}})\quad\leq & \frac{\varepsilon}{2}.
\end{align*}
Combining with \eqref{eq:dtv-1} and by triangle inequality, we conclude
\[
\dtv(\mu_{\+C},\mu^{\text{out}})\leq \dtv(\mu_{\+C},\nu)+\dtv(\nu,\mu^{\text{out}})=\varepsilon. \qedhere
\]
%finishing the proof of the lemma. 
\end{proof}

\appendix

\bibliographystyle{alpha}
\bibliography{references.bib}

\newcommand{\etalchar}[1]{$^{#1}$}
\begin{thebibliography}{KMRT{\etalchar{+}}07}

\bibitem[AC08]{achlioptas2008algorithmic}
Dimitris Achlioptas and Amin {Coja-Oghlan}.
\newblock Algorithmic barriers from phase transitions.
\newblock In {\em FOCS}, pages 793--802, 2008.

\bibitem[ACG19]{ayre2019hypergraph}
Peter Ayre, Amin {Coja-Oghlan}, and Catherine Greenhill.
\newblock Hypergraph coloring up to condensation.
\newblock {\em Random Structures \& Algorithms}, 54(4):615--652, 2019.

\bibitem[Alo91]{Alon91}
Noga Alon.
\newblock A parallel algorithmic version of the local lemma.
\newblock In {\em FOCS}, pages 586--593. IEEE, 1991.

\bibitem[AM02]{achlioptas2002asymptotic}
D.~Achlioptas and C.~Moore.
\newblock The asymptotic order of the random $k$-{SAT} threshold.
\newblock In {\em FOCS}, pages 779--788, 2002.

\bibitem[AP03]{achlioptas2003threshold}
Dimitris Achlioptas and Yuval Peres.
\newblock The threshold for random $k$-{SAT} is $2^k \ln 2 - {O}(k)$.
\newblock In {\em STOC}, page 223–231. ACM, 2003.

\bibitem[ART06]{achilioptas2006solution}
Dimitris Achlioptas and Federico Ricci-Tersenghi.
\newblock On the solution-space geometry of random constraint satisfaction
  problems.
\newblock In {\em STOC}, page 130–139. ACM, 2006.

\bibitem[BCE17]{bapst2017planting}
Victor Bapst, Amin {Coja-Oghlan}, and Charilaos Efthymiou.
\newblock Planting colourings silently.
\newblock {\em Comb., Prob. and Comput.}, 26(3):338–366, 2017.

\bibitem[Bec91]{beck1991algorithmic}
J{\'o}zsef Beck.
\newblock An algorithmic approach to the {L}ov{\'a}sz local lemma.
\newblock {\em Random Struct. Algorithms}, 2(4):343--365, 1991.

\bibitem[BGG{\etalchar{+}}19]{BGGGS19}
Ivona Bez{\'{a}}kov{\'{a}}, Andreas Galanis, Leslie~A. Goldberg, Heng Guo, and
  Daniel {\v{S}}tefankovi{\v{c}}.
\newblock Approximation via correlation decay when strong spatial mixing fails.
\newblock {\em {SIAM} J. Comput.}, 48(2):279--349, 2019.

\bibitem[BH22]{BH22}
Guy Bresler and Brice Huang.
\newblock The algorithmic phase transition of random $k$-{SAT} for low degree
  polynomials.
\newblock In {\em FOCS}, pages 298--309. IEEE, 2022.

\bibitem[BKMP05]{berger2005glauber}
Noam Berger, Claire Kenyon, Elchanan Mossel, and Yuval Peres.
\newblock Glauber dynamics on trees and hyperbolic graphs.
\newblock {\em Probab. Theory Relat. Fields}, 131(3):311--340, 2005.

\bibitem[BS20]{budzynski2020asymptotics}
Louise Budzynski and Guilhem Semerjian.
\newblock The asymptotics of the clustering transition for random constraint
  satisfaction problems.
\newblock {\em {J. Stat. Phys.}}, October 2020.

\bibitem[CCM{\etalchar{+}}24]{chatterjee2024number}
Arnab Chatterjee, Amin {Coja-Oghlan}, Noela M\"{u}ller, Connor Riddlesden,
  Maurice Rolvien, Pavel Zakharov, and Haodong Zhu.
\newblock The number of random $2$-{SAT} solutions is asymptotically
  log-normal.
\newblock In {\em RANDOM 2024}, volume 317, pages 39:1--39:15, Dagstuhl,
  Germany, 2024. Schloss Dagstuhl -- Leibniz-Zentrum f{\"u}r Informatik.

\bibitem[CGG{\etalchar{+}}24]{chen2024fast}
Zongchen Chen, Andreas Galanis, Leslie~Ann Goldberg, Heng Guo, Andrés
  Herrera-Poyatos, Nitya Mani, and Ankur Moitra.
\newblock Fast sampling of satisfying assignments from random $k$-{SAT} with
  applications to connectivity.
\newblock {\em SIAM J. Disc. Math.}, 2024.

\bibitem[CMR22]{coja2022belief}
Amin {Coja-Oghlan}, Noela M{\"u}ller, and Jean~B. Ravelomanana.
\newblock Belief propagation on the random $k$-{SAT} model.
\newblock {\em Ann. Appl. Probab.}, 32(5):3718 -- 3796, 2022.

\bibitem[CO17]{CO17}
Amin Coja-Oghlan.
\newblock Belief propagation guided decimation fails on random formulas.
\newblock {\em J. ACM}, 63(6):1--55, 2017.

\bibitem[COEJ{\etalchar{+}}18]{CEJKK18}
Amin Coja-Oghlan, Charilaos Efthymiou, Nor Jaafari, Mihyun Kang, and Tobias
  Kapetanopoulos.
\newblock Charting the replica symmetric phase.
\newblock {\em Commun. Math. Phys.}, 359:603--698, 2018.

\bibitem[COF14]{CF14Walksat}
Amin Coja-Oghlan and Alan Frieze.
\newblock Analyzing walksat on random formulas.
\newblock {\em SIAM Journal on Computing}, 43(4):1456--1485, 2014.

\bibitem[{Coj}10]{coja2010better}
Amin {Coja-Oghlan}.
\newblock A better algorithm for random $k$-{SAT}.
\newblock {\em SIAM J. Comput.}, 39(7):2823--2864, 2010.

\bibitem[{Coj}14]{coja2014asymptotic}
Amin {Coja-Oghlan}.
\newblock The asymptotic $k$-{SAT} threshold.
\newblock In {\em STOC}, page 804–813. ACM, 2014.

\bibitem[COKPZ17]{CKPZ17}
Amin Coja-Oghlan, Florent Krzakala, Will Perkins, and Lenka Zdeborov{\'a}.
\newblock Information-theoretic thresholds from the cavity method.
\newblock In {\em STOC}, pages 146--157. ACM, 2017.

\bibitem[CP12]{coja2012decimation}
Amin {Coja-Oghlan} and Angelica~Y. {Pachon-Pinzon}.
\newblock The decimation process in random $k$-{SAT}.
\newblock {\em SIAM J. Disc. Math.}, 26(4):1471--1509, 2012.

\bibitem[CZ23]{CZ23}
Xiaoyu Chen and Xinyuan Zhang.
\newblock A near-linear time sampler for the ising model with external field.
\newblock In {\em SODA}, pages 4478--4503. SIAM, 2023.

\bibitem[DFG15]{dyer2015chromatic}
Martin Dyer, Alan Frieze, and Catherine Greenhill.
\newblock On the chromatic number of a random hypergraph.
\newblock {\em J. Comb. Theory Ser. B}, 113(C):68–122, July 2015.

\bibitem[DMMZ05]{daude2005pairs}
Hervé Daudé, Marc Mezard, Thierry Mora, and Riccardo Zecchina.
\newblock Pairs of sat assignments and clustering in random boolean formulae.
\newblock {\em Theore. Comput. Sci.}, 393:260--279, 07 2005.

\bibitem[DSS22]{ding2022satisfiability}
Jian Ding, Allan Sly, and Nike Sun.
\newblock Proof of the satisfiability conjecture for large $k$.
\newblock {\em Ann. Math.}, 196(1):1 -- 388, 2022.

\bibitem[EG24]{AG24}
Ahmed {El Alaoui} and David Gamarnik.
\newblock Hardness of sampling solutions from the symmetric binary perceptron.
\newblock {\em arXiv preprint arXiv:2407.16627}, 2024.

\bibitem[EL75]{LocalLemma}
Paul Erd\H{o}s and L\'aszl\'o Lov\'asz.
\newblock Problems and results on 3-chromatic hypergraphs and some related
  questions.
\newblock {\em Infinite and finite sets, volume 10 of Colloquia Mathematica
  Societatis J\'anos Bolyai}, pages 609--628, 1975.

\bibitem[FB99]{friedgut1999sharp}
Ehud Friedgut and Jean Bourgain.
\newblock Sharp thresholds of graph properties, and the $k$-{SAT} problem.
\newblock {\em Journal of the American Mathematical Society}, 12(4):1017--1054,
  1999.

\bibitem[FHY21]{feng2021sampling}
Weiming Feng, Kun He, and Yitong Yin.
\newblock Sampling constraint satisfaction solutions in the local lemma regime.
\newblock In {\em STOC}, pages 1565--1578. ACM, 2021.

\bibitem[FS96]{frieze1996analysis}
Alan Frieze and Stephen Suen.
\newblock Analysis of two simple heuristics on a random instance of $k$-{SAT}.
\newblock {\em Journal of Algorithms}, 20(2):312--355, 1996.

\bibitem[Gam21]{Gamarnik21}
David Gamarnik.
\newblock The overlap gap property: A topological barrier to optimizing over
  random structures.
\newblock {\em PNAS}, 118(41):e2108492118, 2021.

\bibitem[GGGY21]{GGGY21}
Andreas Galanis, Leslie~Ann Goldberg, Heng Guo, and Kuan Yang.
\newblock Counting solutions to random {SAT} formulas.
\newblock {\em SIAM J. Comput.}, 50(6):1701--1738, 2021.

\bibitem[GGW22]{galanis2021inapproximability}
Andreas Galanis, Heng Guo, and Jiaheng Wang.
\newblock Inapproximability of counting hypergraph colourings.
\newblock {\em ACM Trans. Comput. Theory}, 2022.

\bibitem[GLLZ19]{guo2019counting}
Heng Guo, Chao Liao, Pinyan Lu, and Chihao Zhang.
\newblock Counting hypergraph colorings in the local lemma regime.
\newblock {\em SIAM J. Comput.}, 48(4):1397--1424, 2019.

\bibitem[GM07]{gerschenfeld2007reconstruction}
Antonine Gerschenfeld and Andrea Montanari.
\newblock Reconstruction for models on random graphs.
\newblock In {\em FOCS}, pages 194--204. IEEE, 2007.

\bibitem[GMP05]{GMP05}
Leslie~Ann Goldberg, Russell Martin, and Mike Paterson.
\newblock Strong spatial mixing with fewer colors for lattice graphs.
\newblock {\em SIAM J. Comput.}, 35(2):486--517, 2005.

\bibitem[Het16]{Hetterich16}
Samuel Hetterich.
\newblock Analysing survey propagation guided decimationon random formulas.
\newblock In {\em 43rd International Colloquium on Automata, Languages, and
  Programming (ICALP 2016)}. Schloss Dagstuhl-Leibniz-Zentrum fuer Informatik,
  2016.

\bibitem[HM08]{hatami2008sharp}
Hamed Hatami and Michael Molloy.
\newblock Sharp thresholds for constraint satisfaction problems and
  homomorphisms.
\newblock {\em Random Struct. Algorithms}, 33(3):310--332, 2008.

\bibitem[HSS11]{haeupler2011new}
Bernhard Haeupler, Barna Saha, and Aravind Srinivasan.
\newblock New constructive aspects of the lov{\'{a}}sz local lemma.
\newblock {\em J. {ACM}}, 58(6):28:1--28:28, 2011.
\newblock (Conference version in \emph{FOCS}'10).

\bibitem[HSZ19]{HSZ19}
Jonathan Hermon, Allan Sly, and Yumeng Zhang.
\newblock Rapid mixing of hypergraph independent sets.
\newblock {\em Random Struct. Algorithms}, 54(4):730--767, 2019.

\bibitem[HWY23]{HWY23}
Kun He, Kewen Wu, and Kuan Yang.
\newblock Improved bounds for sampling solutions of random {SAT} formulas.
\newblock In {\em SODA}, pages 3330--3361. SIAM, 2023.

\bibitem[JPV21]{vishesh21towards}
Vishesh Jain, Huy~Tuan Pham, and Thuy~Duong Vuong.
\newblock Towards the sampling lov{\'{a}}sz local lemma.
\newblock In {\em FOCS}, pages 173--183. {IEEE}, 2021.

\bibitem[KKKS98]{kirousis1998approximate}
Lefteris~M. Kirousis, Evangelos Kranakis, Danny Krizanc, and Yannis~C.
  Stamatiou.
\newblock Approximating the unsatisfiability threshold of random formulas.
\newblock {\em Random Struct. Algorithms}, 12(3):253--269, 1998.

\bibitem[KMRT{\etalchar{+}}07]{florent2007gibbs}
Florent Krzakala, Andrea Montanari, Federico Ricci-Tersenghi, Guilhem
  Semerjian, and Lenka Zdeborová.
\newblock Gibbs states and the set of solutions of random constraint
  satisfaction problems.
\newblock {\em PNAS}, 104(25):10318--10323, 2007.

\bibitem[MM06]{mezard2006reconstruction}
Marc M{\'e}zard and Andrea Montanari.
\newblock Reconstruction on trees and spin glass transition.
\newblock {\em J. Stat. Phys.}, 124(6):1317--1350, 2006.

\bibitem[MMZ05]{mezard2005clustering}
M.~M\'ezard, T.~Mora, and R.~Zecchina.
\newblock Clustering of solutions in the random satisfiability problem.
\newblock {\em Phys. Rev. Lett.}, 94:197205, 2005.

\bibitem[Moi19]{Moi19}
Ankur Moitra.
\newblock Approximate counting, the {L}ov{\'{a}}sz local lemma, and inference
  in graphical models.
\newblock {\em J. {ACM}}, 66(2):10:1--10:25, 2019.
\newblock (Conference version in \emph{STOC}'17).

\bibitem[MPZ02]{mezard2002analytic}
M.~Mézard, G.~Parisi, and R.~Zecchina.
\newblock Analytic and algorithmic solution of random satisfiability problems.
\newblock {\em Science}, 297(5582):812--815, 2002.

\bibitem[MRT11]{montanari2011reconstruction}
Andrea Montanari, Ricardo Restrepo, and Prasad Tetali.
\newblock Reconstruction and clustering in random constraint satisfaction
  problems.
\newblock {\em SIAM J. Disc. Math.}, 25(2):771--808, 2011.

\bibitem[MRTS08]{montanari2008cluster}
Andrea Montanari, Federico Ricci-Tersenghi, and Guilhem Semerjian.
\newblock Clusters of solutions and replica symmetry breaking in random
  $k$-satisfiability.
\newblock {\em J. Stat. Mech.: Theory Exp}, 2008(04):P04004, apr 2008.

\bibitem[MS07]{montanari2007counting}
Andrea Montanari and Devavrat Shah.
\newblock Counting good truth assignments of random $k$-{SAT} formulae.
\newblock In {\em SODA}, page 1255–1264. SIAM, 2007.

\bibitem[MT10]{moser2010constructive}
Robin~A. Moser and G{\'a}bor Tardos.
\newblock A constructive proof of the general {L}ov{\'a}sz local lemma.
\newblock {\em J. {ACM}}, 57(2):11, 2010.

\bibitem[MU17]{mitzenmacher2017probability}
Michael Mitzenmacher and Eli Upfal.
\newblock {\em Probability and computing: Randomization and probabilistic
  techniques in algorithms and data analysis}.
\newblock Cambridge university press, 2017.

\bibitem[Par79]{parisi1979infinite}
G.~Parisi.
\newblock Infinite number of order parameters for spin-glasses.
\newblock {\em Phys. Rev. Lett.}, 43:1754--1756, Dec 1979.

\bibitem[RS98]{rabb1998balls}
Martin Raab and Angelika Steger.
\newblock ``balls into bins''' - a simple and tight analysis.
\newblock In {\em RANDOM}, page 159–170, Berlin, Heidelberg, 1998.
  Springer-Verlag.

\bibitem[SJ89]{SJ89}
Alistair Sinclair and Mark Jerrum.
\newblock Approximate counting, uniform generation and rapidly mixing markov
  chains.
\newblock {\em Information and Computation}, 82(1):93--133, 1989.

\bibitem[vdBS94]{BS94}
Jacob van~den Berg and Jeffrey~E Steif.
\newblock Percolation and the hard-core lattice gas model.
\newblock {\em Stoch. Process. Their Appl.}, 49(2):179--197, 1994.

\bibitem[WY24]{WY24}
Chunyang Wang and Yitong Yin.
\newblock A sampling {L}ov{\'a}sz local lemma for large domain sizes.
\newblock arXiv preprint arXiv:2307.14872. To appear in FOCS'24, 2024.

\bibitem[ZK16]{lenka2016statistical}
Lenka Zdeborová and Florent Krzakala.
\newblock Statistical physics of inference: thresholds and algorithms.
\newblock {\em Adv. Phys.}, 65(5):453--552, 2016.

\end{thebibliography}

\section{Proofs of structural properties}\label{sec:structural-property-proof}

In this section, we prove \Cref{lemma:structural-lemma-kSAT}.

The first condition of nice hypergraphs, i.e., $H_\Phi$ is in $\+H_{\le k}$ and has density $\alpha$, is trivial. Then \Cref{property:maximum-degree} follows from a classical result~\cite[Theorem 1]{rabb1998balls}.
\begin{lemma}
    With probability $1 - o(1/n)$ over the random formula $\Phi = \Phi(k, n, m)$ with density $\alpha$, $H_\Phi$ satisfies \Cref{property:maximum-degree}, namely, the maximum degree of variables is at most $4k\alpha + 6 \log n$.
\end{lemma}
Similar to this lemma, the following proposition, which bounds the number of high-degree vertices in $H_\Phi$, also follows from the classical result of the balls-and-bins model. 

\begin{proposition}\label{prop:high-degree}
Assume $k\ge2$, $q \ge 2$ and $\alpha\le q^k$ are constants. Let $p_1$ be a parameter satisfying $p_1 \ge 4k$. Then with probability $1-o(1/n)$ over the random formula $\Phi$, for $H_\Phi = (V, \+E)$, we have
\[
\abs{\{v\in V\mid \deg(v) > p_1\alpha\}}\le \mathrm{e}^{-k}\alpha^{-2}n\,,
\]
i.e., $\abs{\-{HD}(V)}\le \mathrm{e}^{-k}\alpha^{-2}n$.
\end{proposition}

\begin{proof}
    The degrees of the variables in $\Phi$ distribute as a balls-and-bins experiment with $km$ balls and $n$ bins. Let $D_1,\ldots,D_n\sim \mathrm{Pois}(k\alpha)$ be $n$ independent Poisson random variables with parameter $k\alpha$. Then the degrees of the variables in $\Phi$ has the same distribution as $\{D_1,\ldots,D_n\}$ conditioned on the event $\+E_{n,m}$ that $\sum_{i=1}^nD_i=km$ \cite[Chapter 5.4]{mitzenmacher2017probability}. Note that $\sum_{i=1}^nD_i$ is a Poisson random variable with parameter $k\alpha n=km$. Thus 
    \[
    \Pr{\+E_{n,m}}=\mathrm{e}^{-km}\cdot\frac{(km)^{km}}{(km)!}\ge\frac1{\sqrt{2\pi km}}=\frac1{\sqrt{2\pi k\alpha n}}.
    \]
    For any fixed $i \in [n]$, we have
    \begin{align*}
        \Pr{D_i\ge p_1\alpha} &=\Pr{\mathrm{Pois}(k\alpha)\ge p_1\alpha}\le\frac{\mathrm{e}^{-k\alpha}(\mathrm{e} k\alpha)^{p_1\alpha}}{(p_1\alpha)^{p_1\alpha}} \le \mathrm{e}^{-k\alpha} (\mathrm{e}/4)^{4k\alpha} \le \-e^{-k\alpha} 2^{-2k\alpha} \le 2^{-3k\alpha}\,.
    \end{align*}
    Define $U=\{i\in[n]\mid D_i\ge p_1\alpha\}$. Then by Chernoff-Hoeffding bound, we obtain that 
    \[
        \Pr{\abs{U} > \mathrm{e}^{-k}\alpha^{-2}n} \le \Pr{\abs{U} - \mathbf{E}[\abs{U}] > \mathrm{e}^{-k+1}\alpha^{-2}n} < \mathrm{e}^{-2\mathrm{e}^{-2k+2}\alpha^{-4}n} = o(1/n^2)\,,
    \]
    which further gives that
    \begin{align*}
        &\phantom{= {}}\Pr{\abs{\{v\in V(H_\Phi)\mid \deg(v) > p_1\alpha\}}\ge \mathrm{e}^{-k}\alpha^{-2}n}\\ & = \Pr{|U|\ge \mathrm{e}^{-k}\alpha^{-2}n\mid\+E_{n,m}}\\ & \le \sqrt{2\pi k \alpha n} \cdot o(1/n^2) = o(1/n)\,. \qedhere
    \end{align*}
\end{proof}

Now we show that $H_\Phi$ satisfies \Cref{property:edge-expansion} and \Cref{property:bounded-neighbourhood-growth} with probability $1 - o(1/n)$. For \Cref{property:edge-expansion}, we present \Cref{lemma:kSAT-expansion} and its corollary, which are adapted from \cite[Proposition 3.3 \& 3.4]{HWY23}. For \Cref{property:bounded-neighbourhood-growth}, we have \Cref{lemma:bounded-neighbourhood-growth}, which is adapted from \cite[Lemma 8.6]{GGGY21} and \cite[Proposition 3.5]{HWY23}, but gives a better bound.

\begin{lemma}
    \label{lemma:kSAT-expansion}
    For any fixed $k$ and $\alpha$, if $\eta k \ge 4$, $\rho < 1$ and $\mathrm{e}(\rho k\alpha)^\eta \le 1$, then with probability $1 - o(1/n)$ over the choice of random $k$-SAT formula $\Phi = \Phi(k, n, m)$ with density $\alpha$, $H_\Phi$ satisfies \Cref{property:edge-expansion}.
\end{lemma}

\begin{proof}
%If $\ell = 1$, then Property \ref{property:clause-expansion} trivially holds.
%So we assume $\ell \ge 2$.
Let $\ell \le \rho m$ and fix $r = \lfloor (1 - \eta)k\ell \rfloor$. For any $U \subseteq V$ of size $r$ and any $\ell$ edges $e_1, \ldots, e_\ell$, the probability that $e_i \subseteq U$ for every $i$ is
\[
    \biggl(\frac{r}{n}\biggr)^{k\ell} \le \biggl(\frac{(1-\eta)k\ell}{n}\biggr)^{k\ell}\,.
\]
Thus, let $\mathcal{E}_\ell$ be the event that there exists $U \subseteq V$ of size $r$ and $\ell$ edges $e_1, \ldots, e_\ell$ such that $e_i \subseteq U$ for every $i$. We obtain that
\begin{align*}
    \Pr{\mathcal{E}_\ell} &\le \binom{n}{r}\binom{m}{\ell}\biggl(\frac{(1-\eta)k\ell}{n}\biggr)^{k\ell} \le \biggl(\frac{\mathrm{e}n}{r}\biggr)^r \biggl(\frac{\mathrm{e}m}{\ell}\biggr)^\ell\biggl(\frac{(1-\eta)k\ell}{n}\biggr)^{k\ell}\\
    &\le \biggl(\frac{\mathrm{e}n}{(1-\eta)k\ell}\biggr)^{(1-\eta)k\ell} \biggl(\frac{\mathrm{e}m}{\ell}\biggr)^\ell\biggl(\frac{(1-\eta)k\ell}{n}\biggr)^{k\ell} \tag{\text{assuming $n \ge 4r$}} \\
    &\le \Biggl(\frac{\mathrm{e}^{(1-\eta)k+1}((1-\eta)k\ell)^{\eta k}m}{n^{\eta k}\ell}\Biggr)^\ell \\
    &\le \Biggl(\frac{\mathrm{e}^k k^{\eta k}\alpha n \ell^{\eta k}}{n^{\eta k}\ell}\Biggr)^\ell 
    \le \biggl((\mathrm{e}k^\eta)^k\alpha \Bigl(\frac{\ell}{n}\Bigr)^{\eta k-1}\biggr)^\ell.
\end{align*}
If $\ell \le n^{1/3}$, we have
\[
\Pr{\mathcal{E}_\ell} \le (\mathrm{e}k^\eta)^k \alpha n^{-2(\eta k - 1)/3} \le (\mathrm{e}k^\eta)^k \alpha n^{-2}
\]
as long as $\eta k \ge 4$. If $n^{1/3} \le \ell \le \rho m$, noting that $(\rho\alpha)^{\eta k - 1}  \le 1/(\mathrm{e}^k k^{\eta k}\rho \alpha)$, we have
\[
\Pr{\mathcal{E}_\ell} \le \biggl((\mathrm{e}k^\eta)^k\alpha\bigl(\rho\alpha\bigr)^{\eta k - 1}\biggr)^{n^{1/3}} \le \rho^{n^{1/3}} \le n^{-3}\,.
\]
Therefore, by the union bound, the probability that there exists $\ell \le \rho m$ edges $e_1, \ldots, e_\ell$ where
    \begin{align*}
        \left| \bigcup_{i=1}^\ell e_i \right| \le (1-\eta) k \ell
    \end{align*}
is at most
\begin{align*}
    \sum_{\ell = 1}^{\rho m} \Pr{\mathcal{E}_\ell} &\le \frac{(\mathrm{e}k^\eta)^k\alpha n^{1/3}}{n^2} + n^{-2} = o(1/n)\,.\qedhere
\end{align*}
\end{proof}

\begin{corollary}
    \label{cor:bkvars}
    Let $H = (V, \+E) \in \+H_{\le k}$ be a hypergraph satisfying $(\eta, \rho)$-edge expansion. Then for any $\eta < b \le 1$ and $V'\subseteq V$ of size less than $(b-\eta)k\rho \abs{\+E}$, it holds that
    \[
    \abs{V'} \ge (b - \eta)k \abs{\{e \in \+E \mid \abs{e \cap V'} \ge bk\}}\,,
    \]
    namely, the number of hyperedges $e$ such that $\abs{e \cap V'} \ge bk$ is at most $\lfloor\abs{V'}/((b-\eta)k)\rfloor$.
\end{corollary}

\begin{proof}
    For the sake of contradiction, assume that there is a set $V' \subseteq V$ of size $\abs{V'}<(b- \eta)k\rho m$ and $\ell = \lfloor\abs{V'}/((b - \eta)k)\rfloor + 1$ hyperedges $e_1, \ldots, e_\ell$ such that each hyperedge contains at least $bk$ vertices in $V'$. Then it is clear that $\abs{V'} < (b-\eta)k\ell$ and thus
    \[
    \abs{\bigcup_{i = 1}^\ell e_i} \le (1-b)k\ell + \abs{V'} < (1-\eta)k\ell\,.
    \]
    Since $\ell\le \rho m$, it contradicts with \Cref{lemma:kSAT-expansion}.
\end{proof}

\begin{proposition}
    [\text{\cite[Lemma 8.5]{GGGY21}}]
    \label{lemma:probability-of-spanning-tree}
    Let $U$ be any subset of indices of clauses in $\Phi$ and $T$ be any tree on the vertex set $U$. Then the probability that $T$ is a subgraph of $G_\Phi$ is at most $(k^2/n)^{\abs{U} - 1}$.    
\end{proposition}

\begin{lemma}\label{lemma:bounded-neighbourhood-growth}
    Suppose $\alpha \le 2^k$. With probability $1 - o(1/n)$ over the choice of random $k$-SAT formula $\Phi = \Phi(k, n, m)$ with fixed density $\alpha$, $H_\Phi$ satisfies \Cref{property:bounded-neighbourhood-growth}, namely, for every clause $c$ in $\Phi$ and $\ell \ge 1$, there are at most $n^3(\mathrm{e}k^2\alpha)^\ell$ many connected sets of clauses in $G_\Phi$ that contains $c$ and has size $\ell$. 
\end{lemma}

\begin{proof}
If $\ell = 1$, this lemma is trivial. Now we assume $\ell \ge 2$. Let $c$ be an arbitrary clause in $\Phi$ and $U$ be a set of clauses of size $\ell$ where $c\in U$.
Let $T_U$ be the set of all trees with vertex set $U$.
By standard results, we have $\abs{T_U} = \ell^{\ell-2}$.
In addition, any fixed tree $T \in T_U$ is a subgraph of $G_\Phi$ with probability at most $(k^2/n)^{\ell-1}$, by \Cref{lemma:probability-of-spanning-tree}.
Thus, the union bound gives that
$$
\Pr{G_\Phi[U]\text{ is connected}}\le
\ell^{\ell-2}(k^2/n)^{\ell-1}.
$$
Let $Z_{\ell,c}$ be the number of connected sets of clauses with size $\ell$ containing $c$. Then, we have
\begin{align*}
\mathbf{E}[Z_{\ell,c}]
&=\sum_{U:c\in U,|U|=\ell}\Pr{G_\Phi[U]\text{ is connected}}\\
&\le\binom{m-1}{\ell-1}\cdot\ell^{\ell-2}\cdot\biggl(\frac{k^2}n\biggr)^{\ell-1}\\
&\le \frac{(\mathrm{e}m)^{\ell - 1}}{(\ell - 1)^{\ell - 1}}\ell^{\ell - 2}\biggl(\frac{k^2}n\biggr)^{\ell-1}\\
&= \frac{(\mathrm{e}k^2\alpha)^{\ell-1}}{\ell} \Bigl(\frac{\ell}{\ell-1}\Bigr)^{\ell-1}\\
&\le \frac{\mathrm{e}}{\ell}\cdot (\mathrm{e}k^2\alpha)^{\ell-1}\,.
\end{align*}
By Markov's inequality, it further implies that
\begin{align*}
    \Pr{Z_{\ell,c}\ge n^3 (\mathrm{e}k^2\alpha)^\ell}&\le \frac{\mathrm{e}(\mathrm{e}k^2\alpha)^{\ell - 1}}{\ell n^3 (\mathrm{e}k^2\alpha)^\ell} = \frac{1}{n^3 k^2\alpha \ell}\,.
\end{align*}
Finally, using a union bound again, we obtain that
\begin{align*}
    &\phantom{={}} \Pr{\exists\, 2\le \ell\le m \text{ and clause } c \text{ such that } Z_{\ell,c}\ge n^3(\mathrm{e}k^2\alpha)^\ell}\\
    &\le \sum_{\ell = 2}^m \frac{m}{n^3k^2\alpha \ell} \  = \frac{1}{k^2n^2} \sum_{\ell = 2}^m \frac{1}{\ell} \  \le  \frac{\log m}{k^2n^2} \  = o(1/n)\,.\qedhere
\end{align*}
\end{proof}

Next, we show that $H_\Phi$ satisfies \Cref{property:bounded-bad-vertices} and \Cref{property:bounded-bad-fraction} with probability $1 - o(1/n)$ in \Cref{lemma:bounded-bad-vertices} and \Cref{lemma:bounded-ebad-fraction}. To prove these two lemmas, we also need the following properties for random formulas and for \Cref{alg:identify-bad}.

In a hypergraph $H = (V, \+E)$, let $\Gamma_H(V')$ denote the set of neighbors of vertices in $V'$, and let $\Gamma_H^+(V') = V' \cup \Gamma_H(V')$. Then we have the following technical propositions.

\begin{proposition}[\text{\cite[Proposition 3.6]{HWY23}}]
    \label{lemma:bounded-number-neighbors}
    Let $\Phi = (k, n, m)$ be a random $k$-SAT formula with $k \ge 30$ and density $\alpha$. Then with probability $1 - o(1/n)$, we have
    \[
        \abs{\Gamma^+_{H_\Phi}(V')} \le 3k^4\alpha \max\{\abs{V'}, k\log n\}
    \]
    for every connected subset $V'$ of vertices in $H_\Phi$.
\end{proposition}

\begin{proof}
Define $\+E'=\{e\in \+E\;|\; \vbl(e)\cap V'\neq \emptyset\}$. It suffices to bound $|\+E'|\leq 3k^3\alpha\max\{|V'|,k\log n\}$ since $|\Gamma^+_{H_\Phi}(V')|\leq k |\+E'|$. 

We turn our focus to the case when $|V'|\geq \lfloor k\log n\rfloor $. Since $H_{\Phi}(V')$ is connected, there exists a $\+E''\subset \+E'$ such that $V'$ is connected using hyperedges in  $\+E''$ and $|V'|/k\leq \+E''\leq |V'|$. We get that
\[|\+E''|\geq \log n -1 \,.\]
Now, define $\tilde{\+E}=\+E'\setminus \+E''$. Since $k^3\alpha \geq 1$, we will focus on bounding $|{\+E'}|\leq 2k^3\alpha |V'|+|V'|$. Using the fact that $|\+E'|\leq |\tilde{\+E}|+|\+E''|$ and $|\+E''|\leq |V'|$, we focus our attention on bounding $|\tilde{\+E}|\leq 2k^3\alpha |V'|$. Consider any fixed $\+E'',V',\tilde{\+E}$ which satisfy:
\begin{itemize}
    \item $|\+E''|\geq \log n -1$, $|V'|\geq |\+E''|$, $|\tilde{\+E}|\geq 2k^3\alpha |V'|$, and $\tilde{\+E}\cap \+E''=\emptyset$;
    \item Further, $G_{\Phi}(\+E'')$ is connected, $V'\subset \bigcup_{e\in \+E''}\vbl(\+E'')$ and $\vbl(e)\cap V'\neq \emptyset $ for all $e\in\+E''$.
\end{itemize}

Define $r_1=|\+E''|$, $r_2=|V'|$, and $r_3=|\tilde{\+E}|$. Based on these observations, we define the events: 
\begin{itemize}
\item $\mathcal{A}(\+E'',V',\tilde{\+E})$ to be the event when the above conditions are satisfied with $(\+E'', V', \tilde{\+E})$;
    \item $\mathcal{A}(\+E'')$ to be event that $G_{\Phi}[\+E'']$ is connected;
    \item $\mathcal{A}(V',\tilde{\+E})$ to be the event that $\vbl(e)\cap V'\neq \emptyset$ holds for all $e\in\tilde{\+E}$.
\end{itemize}
 Now, using \Cref{lemma:probability-of-spanning-tree}, we get that
\[\Pr{\mathcal{A}(\+E''})\leq r_1^{r_1-2}\cdot \left(\frac{k^2}{n}\right)^{r_1-1}\,.\]
Observe that since $\+E''\cap \tilde{\+E}=\emptyset$, $\mathcal{A}(\+E'')$ and $\mathcal{A}(V',\tilde{\+E})$ are independent. Then,
\[\Pr{\mathcal{A}(V',\tilde{\+E})|\mathcal{A}(\+E'')}=\Pr{\mathcal{A}(V',\tilde{\+E})}\leq \left(k\cdot \frac{r_2}{n}\right)^{r_3}\,.\]
Therefore,
\[\Pr{\mathcal{A}(\+E'',V',\tilde{\+E})}\leq \Pr{\mathcal{A}(V',\tilde{\+E})}\cdot \Pr{\mathcal{A}(\+E''})\leq r_1^{r_1-2}\cdot \left(\frac{k^2}{n}\right)^{r_1-1}\left(k\cdot \frac{r_2}{n}\right)^{r_3}\,.\]
We now use union bound over all possible (valid) sizes of $\+E'',V',\tilde{\+E}$ to upper bound the probability that under the conditions mentioned above, $r_3\geq 2k^3\alpha r_2$. We will show this is upper bounded by $o(1/n)$, which then implies our result. By the union bound,
\[\Pr{\exists \text{ such }\mathcal{A}(\+E'',V',\tilde{\+E})}\leq \sum_{r_1\geq \log n -1} \sum_{r_2\geq r_1} \sum_{r_3\geq 2k^3\alpha r_2}\binom{m}{r_1}\binom{kr_1}{r_2}\binom{m}{r_3} r_1^{r_1-2}\cdot \left(\frac{k^2}{n}\right)^{r_1-1}\left(k\cdot \frac{r_2}{n}\right)^{r_3}\,.\]
To simplify the above expression, we make following assumptions:
$\frac{\-ek^2\alpha}{(2k^2/\-e)^{k^3\alpha}}\leq \frac{1}{8}$ and $\frac{\-ek}{(2k^2/\-e)^{k^3\alpha}}\leq \frac{1}{8}$. These assumptions are satisfied under $k^3\alpha\geq 1$ and $k\geq 30$. Then, we get that
\[\Pr{\exists \text{ such }\mathcal{A}(\+E'',V',\tilde{\+E})}\leq \frac{4n}{n^2(\log n-1)^2}=o(1/n)\,.\]
For the case when $|V'|<\lfloor k\log n \rfloor$, we can consider another connected set of vertices $V''\supset V'$ such that $|V''|=\lfloor k\log n \rfloor$. Applying the previous argument on $V''$, the claim follows.
\end{proof}

Using \Cref{lemma:bounded-number-neighbors}, We can also bound the fraction of high-degree variables in any connected set.

\begin{proposition}\label{prop:fraction_of_high-degrees}
Assume $k \ge 2$, $q \ge 2$, and $\alpha \le q^k$ are constants.
Let $p_1$ be a parameter satisfying $6k^5 \le p_1 \le \-e^{k-2}\alpha$.
Then with probability $1-o(1/n)$ over the random formula $\Phi$, the following holds for $H_\Phi = (V, \+E)$:
Let $V'\subseteq V$ be connected in $H_\Phi$ and $|V'|\ge\log n$. Then 
\[
\abs{\{v\in V'\mid\deg(v)\ge p_1 \alpha\}}\le 6k^5 \abs{V'}/p_1\,.
\]
\end{proposition}

\begin{proof}
    We prove this proposition by showing that if $\-{HD}(V')$ is too large then $V'$ has too many neighbors, which contradicts to \Cref{lemma:bounded-number-neighbors}.
    
    We consider the number of edges that contain high-degree vertices in $V'$. Let $$\+E' = \{e \in \+E \mid e \cap \-{HD}(V') \neq \emptyset\}\,, \  \text{ and } \  \  U = \bigcup_{e \in \+E'} e\,.$$
    By setting $\eta = 1/2$ and $\rho = \-e^{-2}/(k\alpha)$ in \Cref{lemma:kSAT-expansion}, we obtain that any $\ell \le \rho m$ edges contain at least $k\ell/2$ distinct vertices with probability $1 - o(1/n)$. Thus, it gives that $\abs{U} \ge k \min\{\abs{\+E'}, \rho m\}/2$. Note that $U = \Gamma^+_{H_\Phi}(\-{HD}(V')) \subseteq \Gamma^+_{H_\Phi}(V')$. So we have $\abs{\Gamma^+_{H_\Phi}(V')} \ge k \min\{\abs{\+E'}, \rho m\}/2$.

    Now we bound the size of $\+E'$. By a double counting on the size of $\{(v, e) \in \-{HD}(V') \times \+E'\mid v \in e\}$, we have $k\abs{\+E'} \ge p_1\alpha \abs{\-{HD}(V')}$, and thus
    \begin{equation}\label{eq:upper-bound-HD}
    \abs{\-{HD}(V')} \le \frac{k}{p_1\alpha} \abs{\+E'}\,.        
    \end{equation}
    If $\abs{\+E'} \le \rho m$, then by \Cref{lemma:bounded-number-neighbors}, it follows that
    \[
    \frac{k\abs{\+E'}}{2} \le  \abs{\Gamma^+_{H_\Phi}(V')} \le  3k^4\alpha \max\{\abs{V'}, k\log n\} \le 3k^5 \alpha \abs{V'}\,.
    \]
    Combining with \eqref{eq:upper-bound-HD}, it yields that $\abs{\+E'} \le 6k^4\alpha\abs{V'}$, and thus
    \[
        \abs{\-{HD}(V')} \le \frac{6k^5}{p_1} \abs{V'}\,.
    \]
    If $\abs{\+E'} > \rho m$, using  \Cref{lemma:bounded-number-neighbors} again, we obtain
    \[
        \frac{k\rho m}{2} \le \abs{\Gamma^+_{H_\Phi}(V')} \le  3k^4\alpha \max\{\abs{V'}, k\log n\} \le 3k^5 \alpha \abs{V'}\,,
    \]
    which gives \[\abs{V'}\ge\frac{k\rho m}{6k^5\alpha} = \frac{n}{6\mathrm{e}^2k^5\alpha}\,.\]
    On the other hand, it is clear that $\abs{\-{HD}(V')} \le \abs{\-{HD}(V)} \le \-e^{-k}\alpha^{-2}n$ by \Cref{prop:high-degree}. So we obtain that
    \[
    \abs{\-{HD}(V')} \le \frac{6\-e^2k^5\alpha}{\-e^k\alpha^2} \abs{V'} \le \frac{6k^5}{p_1}\abs{V'}\,. \qedhere
    \]
\end{proof}

The following result is adapted from \cite[Lemma 2.4]{CF14Walksat} and \cite[Lemma A.2]{HWY23}, but uses tighter parameters.
\begin{proposition}\label{prop:peeling}
Assume $k\ge12$, $q\ge 2$ and $\alpha\le q^k$ are constants.
Then with probability $1-o(1/n)$ over the random formula $\Phi$, the following holds for $H_\Phi = (V, \+E)$:
Fix an arbitrary $\+E' \subseteq \+E$ of size $\abs{\+E'}\le \-4^{-k}\alpha^{-1/2}n$. Let $e_{i_1},\ldots,e_{i_\ell}\in\+E\setminus\+E'$ be hyperedges of distinct indices. For each $s\in[\ell]$, define $V_s=(\bigcup_{e\in\+E'} e)\cup(\bigcup_{j=1}^{s-1}e_{i_j})$.
If $|e_{i_s}\cap V_s|\ge6$ holds for all $s\in[\ell]$, then $\ell\le \abs{\+E'}$.
\end{proposition}

\begin{proof}
We prove the statement by contradiction. Hence, assume that $\+E'$ and $e_{i_1},...,e_{i_\ell}$ violate the statement for $\ell>|\+E'|$. We can discard additional clauses from $\E'$ and in fact, assume that $\ell=|\+E'|+1$ for which the statement is violated. Further, observe that we can assume that $|\+E'|\leq m-\ell$. Define $\varepsilon>0$ such that $\ell=\lfloor\varepsilon n\rfloor +1$. Since $|\+E'|\geq 1$, we have that $\varepsilon\geq \frac{1}{n}$. From the bound on $|\+E'|$, we get that $\varepsilon\leq 4^{-k}\alpha^{-1/2}+\frac{1}{n}$.

Define $Y:=\bigcup_{j=1}^{j=\ell}\vbl(e_{i_j})\setminus \bigcup_{e\in \+E'}\vbl(e)$. Then, the following are true:
\begin{itemize}
    \item $|Y|=\sum_{r=1}^{r=\ell}|\vbl(e_{i_r})|-|\vbl(e_{i_r})\cap V_r|\leq (k-6)\ell$;
    \item There is a $\tilde{\+E}\subset \+E\setminus \+E'$ such that $|\tilde{\+E}|=\ell$ and $\vbl(\tilde{e})\subseteq Y\cup \bigcup_{e\in \+E'}\vbl(e)$ for all $e\in \tilde{\+E}$. For this, we can just take $\tilde{\+E}=\{e_{i_1},...,e_{i_\ell}\}$.
\end{itemize}
Now, consider fixed $\+E', Y, \tilde{\+E}$ such that $|Y|=t\leq (k-6)\ell$, $|\tilde{\+E}|=\ell$ and $|\+E'|=\min\{m-\ell,\ell-1\}$. Define $\mathcal{A}(\+E',Y,\tilde{\+E})$ to be the event that $\vbl(\tilde{\+E})\subset Y\cup \vbl(\+E')$. We have that
\[
\Pr{\mathcal{A}(\+E',Y,\tilde{\+E})}\leq \left(\frac{k|\+E'|+|Y|}{n}\right)^{k\tilde{\+E}}\leq \left(\frac{k(\ell-1)+(k-6)\ell}{n}\right)^{k\ell}\leq (4k\varepsilon)^{k\ell}.
\]
The last inequality follows from the fact that $\ell\leq \varepsilon n+1\leq 2\varepsilon n$.
Now, by union bound over all choices of $(\+E',Y,\tilde{\+E})$, we get that
\[
\Pr{\exists \text{ such an } (\+E',Y,\tilde{\+E})}\leq \sum_{t=1}^{t=k-6}\binom{m}{\ell-1}\cdot \binom{m}{\ell}\cdot \binom{n}{t}(4k\varepsilon)^{k\ell}.
\]
We can upper bound both $\binom{m}{\ell}$ and $\binom{m}{\ell-1}$ by $(\frac{
\-em}{\ell-1})^\ell$. Further, $(k-6)\ell\leq (k-6)(\varepsilon n+1)\leq 2k\varepsilon n\leq n/2$ if we assume $\varepsilon\leq 1/(4k)$. But we have that $\varepsilon\leq 4^{-k}\alpha^{-1/2}+\frac{1}{n}$ and we can use that $k\geq 12$. Then, $\binom{n}{t}\leq 
\left(\frac{\-en}{(k-6)\ell}\right)^{(k-6)\ell}$. Now,
\begin{align*}
 \Pr{\exists \text{ such an } (\+E',Y,\tilde{\+E})}&\leq n\cdot \left(\frac{\-e\alpha n}{\ell-1}\right)^{2\ell}\cdot \left(\frac{
 \-en}{(k-6)\ell}\right)^{(k-6)\ell}\cdot (4k\varepsilon \ell)^{k\ell}   \\
 &\leq n\cdot \left(\frac{
 \-e^{k-4}\cdot \alpha^2\cdot 2^{2k}\cdot n^{k-4}\cdot\varepsilon^k\cdot k^k}{\ell^{k-6}\cdot (\ell-1)^2 (k-6)^{k-6}}\right)^{\ell}.
\end{align*}
Now, to further evaluate the R.H.S., we use the fact that $\frac{k^k}{(k-6)^{k-6}}\leq (\-ek)^6$. Also, $\ell\leq \varepsilon n$ and $\ell-1\geq \varepsilon n/2$. Then, we get that
\begin{align*}
 \Pr{\exists \text{ such an } (\+E',Y,\tilde{\+E})}&\leq n\cdot \left(\-e^{k+2}2^{2k}\cdot \alpha^{2}\cdot \varepsilon^4\cdot k^6 \right)^{\ell} :=\tilde{p}. 
 %\Pr{\exists \text{ such an } (\+E',Y,\tilde{\+E})}&\leq n\cdot \left(\right)^{l}  \\
\end{align*}
Observe that since $\varepsilon\leq 4^{-k}\alpha^{-1/2}+\frac{1}{n}$, we get that $\-e^{k+2}2^{2k}\cdot \alpha^{2}\cdot \varepsilon^4\cdot k^6\leq \frac{1}{2}$ since $k\geq 12$.
Now, based on whether $\varepsilon n\geq 5\log n$ or $\varepsilon n<5\log n$, we make two cases: 
\begin{itemize}
    \item If $\varepsilon n\geq 5\log n$, we get that $\tilde{p}\leq n\cdot (\frac{1}{2})^{\varepsilon n}=o(1/n^3)$;
    \item If $\varepsilon n<5\log n$, we assume $n\geq 2^{\Omega(k)}$ (for an appropriately large constant), which implies that $\-e^{k+2}2^{2k}\cdot \alpha^{2}\cdot \varepsilon^4\cdot k^6=o(1/n^3)$ and we also have that $\ell\geq 2$, which again implies that $\tilde{p}=o(1/n^3)$. 
\end{itemize}
Observe that we have shown the above probability bound for $\+E'$ of a specific size, which, in our case was $\min\{\ell-1,m-\ell\}$. Then, by union bound over all possible sizes of $\+E'$, we get that $\Pr{\exists \text{ no } (\+E',Y,\tilde{\+E})}\geq 1-o(1/n)$.
\end{proof}
Using \Cref{prop:peeling}, we can bound the size of bad vertices and bad hyperedges in any connected sets.
Given a hypergraph $H = (V, \+E)$, we now simplify $\vbad(V)$ to $\vbad$ and $\ebad(V)$ to $\ebad$. We first present the following facts.

\begin{fact}
    For any $V' \subseteq V$, it holds that $\vbad(V') \subseteq \vbad$ and $\ebad(V') \subseteq \ebad$.
\end{fact}

\begin{fact}
    [\text{\cite[Lemma 8.9]{GGGY21}}]\label{fact:connected-component-Hbad}
    For any $V' \subseteq \vbad$ such that $V'$ is a connected component in the induced subgraph $H[\vbad]$, it holds that $\vbad(V') = V'$.
\end{fact}

If \Cref{prop:peeling} holds for $H_\Phi$, we can prove that \Cref{property:bounded-bad-vertices} is also satisfied, which states that the number of bad variables is bounded by the number of high-degree variables. 

\begin{lemma}\label{lemma:bounded-bad-vertices}
    For any fixed $k$ and $\alpha$, if $\eta k \ge 4$, $\-e(\rho k \alpha)^\eta \le 1$, $\eps_1 > \eta$, and $6k^5 \le p_1 \le \-e^{k-2}\alpha$, then with probability $1 - o(1/n)$ over the choice of random formula $\Phi$, it holds for $H_\Phi = (V, \+E)$ that:
    \begin{enumerate}
        \item for any $V' \subseteq V$, we have $\abs{\vbad(V')} \le 2\abs{\-{HD}(V')}/(\eps_1 - \eta)$, in particular, $H_\Phi$ satisfies \Cref{property:bounded-bad-vertices} if $\eps_1 = 2\eta$;
        \item for any $V' \subseteq \vbad$ such that $V'$ consists of some connected components in $H_\Phi[\vbad]$, we have $\abs{V'} \le 2\abs{\-{HD}(V')}/(\eps_1 - \eta)$.
    \end{enumerate}
\end{lemma}

\begin{proof}
    We only need to prove item (1). Item (2) is a direct corollary of item (1) and \Cref{fact:connected-component-Hbad}.

    By \Cref{prop:high-degree}, we have $\abs{\-{HD}(V')} \le \abs{\-{HD}(V)} \le \-e^{-k}\alpha^{-2}n$. Let
    \[
        \+E' = \{e \in \+E \mid \abs{e \cap \-{HD}(V')} \ge \eps_1 k\}.
    \]
    By setting $b = \eps_1$ in \Cref{lemma:kSAT-expansion} and \Cref{cor:bkvars}, it follows that \[\abs{\-{HD}(V')} \ge (\eps_1 - \eta)k \abs{\+E'}\,.\]
    In particular, $\abs{\+E'}\le \abs{\-{HD}(V')} \le \-e^{-k}\alpha^{-2}n$. Observe that, in \Cref{alg:identify-bad}, if there exist more than one hyperedge that can be added into $\ebad(V')$, we can add them in an arbitrary order without changing the output of \Cref{alg:identify-bad}. So we start by adding all hyperedges in $\+E'$ into $\ebad(V')$ first. After that, each hyperedge newly added to $\ebad(V')$ intersects at least $\eps_1 k$ vertices with existing hyperedge in $\ebad(V')$. Hence, applying \Cref{prop:peeling}, we obtain that $\abs{\ebad(V')\setminus \+E'} \le \abs{\+E'}$, which implies that \[\abs{\ebad(V')} \le 2\abs{\+E'} \le \frac{2}{(\eps_1 - \eta)k} \abs{\-{HD}(V')}\]
    and thus \[\abs{\vbad(V')} \le k\abs{\ebad(V')} \le \frac{2}{\eps_1 - \eta} \abs{\-{HD}(V')}\,.\qedhere\]
\end{proof}

Now we can prove \Cref{property:bounded-bad-fraction}. The following two results are adapted from \cite[Lemma 4.4 \& Corollary 4.2]{HWY23} using our parameters.

\begin{proposition}\label{lemma:bounded-vbad-fraction}
    For any fixed $k$ and $\alpha$, assume $\eta, \rho, p_1, \eps_1$ are parameters satisfying $\eta k \ge 4$, $\-e(\rho k \alpha)^\eta \le 1$, $\eps_1 \ge \eta + 1/k$ and $6k^5 \le p_1 \le \-e^{k-2}\alpha$. Then with probability $1 - o(1/n)$ over the choice of random formula $\Phi$, for any $V' \subseteq V$ of size $\abs{V'} \ge \log n$ connected in $H_\Phi = (V, \+E)$, it holds that $\abs{V' \cap \vbad} \le \frac{12k^5}{(\eps_1-\eta)p_1}\abs{V'}$.
\end{proposition}

\begin{proof}
    Let $V_1, V_2, \ldots, V_\ell$ be connected components in $H_\Phi[\vbad]$ that intersects $V'$, and let \[\tilde{V} = V' \cup V_1 \cup V_2 \cup \cdots \cup V_\ell\,.\]
    Note that $\tilde{V}$ is connected in $H_\Phi$, and $\-{HD}(\tilde{V}) = \-{HD}(V_1) \cup \-{HD}(V_2) \cup\cdots \cup \-{HD}(V_\ell)$.

    By \Cref{lemma:bounded-bad-vertices}, we have $\abs{V_i} \le 2\abs{\-{HD}(V_i)}/(\eps_1 - \eta)$. By \Cref{prop:fraction_of_high-degrees}, we have $\abs{\-{HD}(\tilde{V})} \le 6k^5\abs{\tilde{V}}/p_1$. Thus, it follows that
    \[
    \abs{\tilde{V} \cap \vbad} = \sum_{i = 1}^\ell \abs{V_i} \le \frac{2}{\eps_1 - \eta} \sum_{i = 1}^\ell \abs{\-{HD}(V_i)} \le \frac{2}{\eps_1 - \eta} \cdot \frac{6k^5 \abs{\tilde{V}}}{p_1}\,.
    \]
    Since $\tilde{V} \setminus V' \subseteq \vbad$, we conclude that
    \[
    \frac{\abs{V' \cap \vbad}}{\abs{V'}} \le \frac{\abs{V' \cap \vbad} + \abs{\tilde{V} \setminus V'}}{\abs{V'} + \abs{\tilde{V}}} = \frac{\abs{\tilde{V} \cap \vbad}}{\abs{\tilde{V}}} \le \frac{12k^5}{(\eps_1-\eta)p_1}\,.\qedhere
    \]
\end{proof}

As a corollary, we obtain the proof of \Cref{property:bounded-bad-fraction}.

\begin{lemma}
    \label{lemma:bounded-ebad-fraction}
    For any fixed $k$ and $\alpha$, assume $\eta, \rho, p_1, \eps_1$ are parameters satisfying $\eta k \ge 4$, $\-e(\rho k \alpha)^\eta \le 1$, $\eps_1 \ge \eta + 1/k$ and $6k^5 \le p_1 \le \-e^{k-2}\alpha$. Then with probability $1 - o(1/n)$ over the choice of random formula $\Phi$, for any $\+E' \subseteq \+E$ of size $\abs{\+E'} \ge \log n$ connected in the line graph of $H_\Phi = (V, \+E)$ (namely, connected in $G_\Phi$), it holds that $\abs{\+E' \cap \ebad} \le \frac{12k^5}{(1-\eta)(\eps_1 - \eta)p_1} \abs{\+E'}$. In particular, $H_\Phi$ satisfies \Cref{property:bounded-bad-fraction} if $\eps_1 = 2\eta$ and $\eps_2 = \frac{12k^5}{(1-\eta)\eta p_1}$.
\end{lemma}

\begin{proof}
    Let $V' = \cup_{e\in \+E'} e$. Clearly, $V'$ is connected in $H_\Phi$, and $\abs{V'} \le k \abs{\+E'}$.
    
    We first assume that $\abs{\+E'} \le \rho \abs{\+E}$. Then \Cref{lemma:kSAT-expansion} applies. It implies that $\abs{V'} \ge (1 - \eta)k\abs{\+E'} \ge \log n$, and thus \Cref{lemma:bounded-vbad-fraction} applies. Therefore, we obtain
    \[
        \abs{V' \cap \vbad} \le \frac{12k^5}{(\eps_1-\eta)p_1} \abs{V'}\le \frac{12k^6}{(\eps_1-\eta)p_1} \abs{\+E'}\,.
    \]
    Note that, each hyperedge in $\+E' \cap \ebad$ is a subset of $V' \cap \vbad$. Applying \Cref{lemma:kSAT-expansion} again, we have
    \[
        \abs{\+E' \cap \ebad} \le \frac{\abs{V' \cap \vbad}}{(1-\eta) k} \le \frac{12k^5}{(1-\eta)(\eps_1 - \eta)p_1} \abs{\+E'}\,.
    \]
    
    Now we consider the case where $\abs{\+E'} > \rho \abs{\+E}$. Again, note that each hyperedge in $\+E' \cap  \ebad$ is a subset of $\vbad$, and $\abs{\vbad} \le 2\-e^{-k}\alpha^{-2}\abs{V}/(\eps_1 - \eta) < \rho \abs{\+E}/((1-\eta)k)$ by \Cref{prop:high-degree} and \Cref{lemma:bounded-bad-vertices} if we set $\eta \ge 4/k$ and $\-e(\rho k \alpha)^\eta = 1$. So \Cref{lemma:kSAT-expansion} applies to $\vbad$, and we conclude that
    \[
        \abs{\+E' \cap \ebad} \le \frac{\abs{\vbad}}{(1 - \eta)k} \le \frac{2\abs{\+E}}{\-e^kk\alpha^3(1-\eta)(\eps_1-\eta)} < \frac{12k^5}{(1-\eta)(\eps_1 - \eta)p_1} \abs{\+E'}\,.\qedhere
    \]
\end{proof}

Finally, it is a direct corollary of \Cref{property:bounded-bad-fraction} that $H_\Phi$ has no connected components of size $\ell \ge \log n$ in the line graph induced by bad hyperedges.

\end{document}